\documentclass[11pt]{article}
\usepackage[utf8]{inputenc}
\usepackage[authoryear]{natbib}
\usepackage{bibunits}
\defaultbibliographystyle{apalike}
\defaultbibliography{ref}
\usepackage[a4paper]{geometry}
\usepackage[title]{appendix}
\usepackage{xcolor}
\usepackage{setspace}
\usepackage{algorithm}
\usepackage{algpseudocode}
\algrenewcommand\algorithmicrequire{\textbf{Input:}}
\algrenewcommand\algorithmicensure{\textbf{Output:}}
\usepackage{amsmath}

\usepackage{amsfonts}
\usepackage{bbm}
\usepackage{graphicx}
\usepackage{booktabs, multirow}
\usepackage{makecell}
\usepackage{amsthm}
\theoremstyle{plain}
\newtheorem{lemma}{Lemma}
\newtheorem{remark}{Remark}
\newtheorem{proposition}{Proposition}
\usepackage{url}
\usepackage{comment}
\usepackage{subcaption}
\numberwithin{equation}{section}

\title{Stabilised weighted data subsampling for accelerated Bayesian inference in models with recursive likelihoods}
\author{Matias Quiroz\thanks{School of Mathematical and Physical Sciences, University of Technology Sydney, Sydney, NSW, Australia. \hspace*{1.5em} Corresponding author Matias Quiroz. Email: quiroz.matias@gmail.com.}$\;\,^\ddagger$ \and Aishwarya Bhaskaran\thanks{School of Mathematics and Statistics, University of New South Wales, Sydney, NSW, Australia.} \and Zixuan Wang$^*$  \and Thomas Goodwin\thanks{Human Technology Institute, University of Technology Sydney.}}
\date{}

\begin{document}
\begin{bibunit}

\maketitle

\begin{abstract}
Bayesian inference for models with recursively defined likelihoods is computationally demanding, limiting scalability to large datasets. We propose a stabilised weighted subsampling methodology for accelerated Bayesian inference based on an unbiased estimator of the log-likelihood. By assigning higher sampling probabilities to early observations, the method reduces the effective depth of recursive likelihood evaluations and hence computational cost. This creates a challenging variance-cost trade-off: sampling probabilities that decay too slowly yield limited computational savings, while overly aggressive decay can substantially inflate estimator variance. We develop a stabilisation framework, supported by theory, that controls this trade-off and leads to a practical, theory-informed tuning strategy. We also derive an unbiased subsampling estimator of the log-likelihood gradient, enabling gradient-based Bayesian inference. The methodology can be embedded within a range of Bayesian computational algorithms. We illustrate its use in variational Bayes and subsampling Markov chain Monte Carlo for conditional volatility models, including leverage effects. Empirical results show substantial computational speed-ups relative to full-data methods while maintaining inferential accuracy. We also compare with recent stochastic gradient MCMC and divide-and-conquer MCMC methods for temporally dependent data, observing favourable empirical performance.
     \\ 
 \textbf{Keywords}: Log-likelihood estimation; variance control; computational efficiency.
 
\end{abstract}
%\newpage
\section{Introduction}
Inference for models with recursively defined likelihoods is computationally demanding, since evaluating the likelihood requires sequential propagation of recursively updated quantities across the entire dataset. This recursive structure limits scalability in large samples and settings requiring repeated likelihood or gradient evaluations, such as optimisation-based and Monte Carlo inference methods.

Methods to estimate the log-likelihood via data subsampling have been proposed both in the classical paradigm \citep{Wang2018subsampling, yao2021review, ai2021optimal} and the Bayesian paradigm \citep{korattikara2014austerity, bardenet2014towards, bardenet2017markov, quiroz2019speeding}. In this paper, we focus on Bayesian inference. For subsampling methods to offer a significant speed-up over using the full dataset, the computation on the subsample must be significantly faster. This is difficult to achieve for models with recursively defined likelihoods: when an observation is subsampled, all preceding time steps must still be computed. Consequently, if an observation late in the sequence is subsampled, there is little computational gain, as nearly the entire dataset must still be evaluated, even with a single subsampled observation. As a result, uniform sampling offers limited computational benefit in such settings. We address this issue by proposing a weighted subsampling scheme that samples observations early in the series with a higher probability, while preserving unbiasedness of the log-likelihood estimator. Although such decay can substantially reduce expected computational cost by shortening the effective depth of recursive likelihood evaluations, overly aggressive decay may lead to severe variance inflation. Conversely, insufficient decay approaches uniform sampling and inherits its limited computational efficiency in this setting. We develop theoretical results and practical guidelines that characterise this variance-cost trade-off and provide a stabilised probability design for recursive likelihoods.

We illustrate the proposed methodology using a prominent class of conditional volatility models with recursively defined likelihoods. Generalised autoregressive conditional heteroskedasticity (GARCH) models \citep{bollerslev1986generalized} provide a flexible framework for modelling and forecasting time-varying volatility in financial time series, capturing key features such as volatility clustering and persistence that are commonly observed in asset returns. GARCH models pose significant computational challenges due to the recursive structure of the conditional variance equation. Posterior-based inference thus requires the recursive likelihood or its gradient to be evaluated repeatedly. Several approaches for estimating GARCH models have been proposed. In the original paper, \cite{bollerslev1986generalized} considers an iterative procedure for maximum likelihood estimation assuming Gaussian innovation distributions. More generally, the Gaussian likelihood may be used as a quasi-likelihood for unknown innovation distributions and results in a well-behaved estimator when the fourth moment is finite \citep{huang2008estimating}. In the Bayesian inference literature, exact inference for GARCH models has been carried out using algorithms such as the griddy-Gibbs sampler \citep{Bauwens1998Bayesian}, and sequential Monte Carlo methods \citep{Li2021GARCHSMC}, while approximate inference has been approached using stochastic variational Bayes methods \citep{magris2023variational, xuan2024stochastic}. \cite{aicher2025stochastic} propose a stochastic gradient Markov chain Monte Carlo (MCMC) method \citep{chen2014stochastic,nemeth2021stochastic} for general state space models that extends stochastic gradient MCMC beyond the hidden Markov model assumption in \cite{ma2017stochastic}. This extended class includes GARCH models with noise; we compare with this approach in the supplementary material. However, both \cite{aicher2025stochastic, ma2017stochastic} use buffered estimators of the gradient of the log-likelihood that truncate the data and induce a controllable bias. In contrast, our approach provides an unbiased estimator of both the log-likelihood and its gradient. We illustrate the latter using a doubly stochastic variational Bayes \citep{Titsias2014doubly} algorithm for GARCH models. \cite{xuan2024stochastic} exclude stochastic variational inference via data subsampling for GARCH models: ``$\dots$ there is no clear and unbiased way to subsample time series data for GARCH models $\dots$'' (p.\ 2). Our contribution thus enables stochastic variational inference through data subsampling for recursive likelihoods, since stochastic optimisation requires an unbiased gradient estimator \citep{robbins1951stochastic}.

Our paper makes the following contributions. We propose data subsampling estimators of the log-likelihood and its gradient that are well suited for likelihood functions with recursive components. The key methodological contribution is a stabilised weighted sampling design based on truncated decaying sampling probabilities. The decay substantially reduces computational cost by shortening the effective depth of recursive likelihood evaluations, but has the non-trivial consequence that overly aggressive decay can induce severe variance inflation, whereas insufficient decay yields little computational gain. We develop theoretical results that characterise this variance-cost trade-off and motivate a simple stabilisation safeguard, together with practical tuning guidelines based on both estimator variance and expected computational cost. In addition, we derive general recursive expressions for gradients and Hessians of GARCH-type log-likelihoods of arbitrary order, including threshold extensions, using a fully vectorised matrix-calculus formulation, together with a reparameterisation that maps the constrained model parameters to an unconstrained space. Empirically, we show that the proposed stabilised weighted subsampling methodology, when embedded within variational Bayes and Markov chain Monte Carlo algorithms, yields substantial computational speed-ups relative to their full-data counterparts while retaining comparable inferential accuracy, and outperforms uniform subsampling. We further compare our approach, embedded within subsampling MCMC \citep{quiroz2019speeding}, with stochastic gradient MCMC \citep{aicher2025stochastic} and a divide-and-conquer approach for temporally dependent data \citep{ou2025scalable}. For the latter comparison, we additionally consider a higher-dimensional regression-GARCH setting to demonstrate scalability in the number of model parameters, since standard GARCH models are naturally low-dimensional. In all cases, we observe favourable empirical performance.

The rest of the paper is organised as follows. Section \ref{sec:methodology} presents our novel stabilised weighted subsampling methodology. Section \ref{sec:VB_application} presents a subsampling variational Bayes application for a GARCH model. Section \ref{sec:conclusion_and_future_research} concludes and outlines future research. The paper also includes an appendix and supplementary material. Appendix \ref{sec:VolatilityModels} describes the conditional volatility models that serve as our working examples, while Appendix \ref{app:derivatives_GARCH} collects the derivatives required for implementing our method. Proofs of our theoretical results are collected in Section \ref{supp:Proofs} of the supplementary material, while Section \ref{supp:moved_main_material} provides additional results and discussion supporting the stabilised weighted subsampling methodology. The derivations of the derivatives in Appendix \ref{app:derivatives_GARCH} are provided in Section \ref{supp:derivations}. The supplementary material also contains additional results for subsampling Markov chain Monte Carlo posterior inference, including an out-of-sample model evaluation exercise involving GARCH models and a threshold extension thereof in Section \ref{sec:MCMC_application}. Sections \ref{sec:comparison_SG_MCMC} and \ref{sec:comparison_divide-and-conquer_MCMC} contain comparisons with, respectively, stochastic gradient MCMC and a divide-and-conquer approach. The code used to reproduce all experiments and other figures is publicly available\footnote{Code available at \url{https://github.com/matiasq/stabilised-weighted-subsampling}.}. We refer to equations, sections, lemmas, and related results in the main paper as (1.1), Section 1, Lemma 1, and so on. Corresponding items in the appendix and supplementary material are denoted by (A1.1) and (S1.1), Appendix A and Section S1, and Lemma S1, and so on.

\section{Methodology}\label{sec:methodology}

\subsection{A recursive likelihood example: General conditional volatility models}\label{subsec:general_model}
As illustrative working examples, we consider deterministic models for the conditional variance, known as conditional heteroscedasticity models, in which the conditional variance is modelled as a function of historical quantities, such as past (squared) returns. These models typically consist of two components: 
\begin{itemize}
    \item[(i)] a measurement equation that models the (random) time series via a distribution on the error term, and
    \item[(ii)] a conditional variance equation that models its (deterministic) dynamics.
\end{itemize}
We defer details of the error distribution and the volatility model to Appendix \ref{sec:VolatilityModels} and now turn to our methodological framework, which we illustrate using the following general conditional volatility model and notation.

Let $\{y_t\}_{t=1}^T$ be a time series whose conditional variance is modelled by an equation governed by a variance parameter vector $\boldsymbol{\theta}_v$. Consider the general model
\begin{align}\label{eq:measurement_plus_deterministic_state}
    y_t & = \mu + z_t, \,\, z_t = \sigma_t \varepsilon_t, \,\, \varepsilon_t \overset{\mathrm{ind}}{\sim} \mathcal{D}(0, 1),  \,\, \mu \in \mathbb{R}, \,\, \sigma_t > 0, \nonumber \\
    \sigma^2_t & = g(\boldsymbol{\theta}_v;\mathcal{F}_{t-1}),
\end{align}
where $z_t$ denotes the innovation (or shock), and $\varepsilon_t$ the standardised innovation (or standardised shock). In \eqref{eq:measurement_plus_deterministic_state}, $\mu$ and $\sigma_t^2$ are the conditional mean $\mathbb{E}(y_t | \mathcal{F}_{t-1})$ and variance $\mathbb{V}(y_t | \mathcal{F}_{t-1})$, respectively, $\mathcal{D}(0, 1)$ is a distribution with zero mean and unit variance, $g:\mathbb{R}^{\dim_{\boldsymbol{\theta}_v}}\rightarrow (0,\infty)$, and $\mathcal{F}_{t-1}$ is the information set at time $t-1$. The information set may include, for example, past returns and conditional variances. The parameter vector $\boldsymbol{\theta}_v$, governing the conditional variance recursion includes,  amongst other parameters, $\mu$, and we assume without loss of generality that this is the first element in $\boldsymbol{\theta}_v$. Conditional on the required presample conditional variances, which we treat as fixed and known, $\sigma_t^2$, $t=1,\dots,T$, is determined recursively.

Let $\boldsymbol{\theta}=(\boldsymbol{\theta}_v^\top,\theta_\varepsilon)^\top
\in\boldsymbol{\Theta}\subseteq\mathbb{R}^{\dim_{\boldsymbol{\theta}}}$ denote the original parameter vector, where $\theta_\varepsilon$ is an optional parameter governing the error distribution. The parameter space $\boldsymbol{\Theta}$ is generally constrained by restrictions on the conditional variance process and, when applicable, the error distribution. We impose priors and model constraints in this original parameterisation. Our subsampling methodology constructs Taylor control variates using gradients and Hessians of the log-density. For this purpose, we work with the unconstrained reparameterisation
$$
\boldsymbol{\phi}
=(\boldsymbol{\phi}_v^\top,\phi_\varepsilon)^\top
=h(\boldsymbol{\theta}),
\qquad
h:\boldsymbol{\Theta}\rightarrow\mathbb{R}^{\dim_{\boldsymbol{\phi}}},
\qquad
\dim_{\boldsymbol{\phi}}=\dim_{\boldsymbol{\theta}},
$$
where $h$ is one-to-one with inverse $h^{-1}$. Throughout, we use the shorthand
$\ell_t(\boldsymbol{\phi})=\ell_t(h^{-1}(\boldsymbol{\phi}))$,
with $\nabla_{\boldsymbol{\phi}}\ell_t(\boldsymbol{\phi})$ and
$\nabla^2_{\boldsymbol{\phi}}\ell_t(\boldsymbol{\phi})$ denoting the corresponding gradient and Hessian with respect to $\boldsymbol{\phi}$; formal definitions are given in Appendix \ref{sec:VolatilityModels}.

\subsection{The log-likelihood and its gradient}\label{subsec:likelihood}

We now introduce the main quantities to be estimated by subsampling, namely the log-likelihood and its gradient, using the conditional volatility model in Section \ref{subsec:general_model} as a working example. The likelihood factorises into conditional densities as
\begin{align}\label{eq:likelihood}
    p(y_1, \dots, y_T | \boldsymbol{\phi})
    &=
    \prod_{t=1}^T
    p(y_t|\mathcal{F}_{t-1},\boldsymbol{\phi}),
\end{align}
where $p(y_t|\mathcal{F}_{t-1},\boldsymbol{\phi})$ denotes the conditional density of observation $y_t$. Hence, the log-likelihood $\ell(\boldsymbol{\phi})$ and its gradient decompose as
\begin{align}
    \ell(\boldsymbol{\phi})
    &=
    \sum_{t=1}^T \ell_t(\boldsymbol{\phi}),
    \qquad
    \ell_t(\boldsymbol{\phi})
    =
    \log p(y_t|\mathcal{F}_{t-1},\boldsymbol{\phi}),
    \label{eq:log-likelihood}
    \\
    \nabla_{\boldsymbol{\phi}}\ell(\boldsymbol{\phi})
    &=
    \sum_{t=1}^T \nabla_{\boldsymbol{\phi}}\ell_t(\boldsymbol{\phi}).
    \label{eq:grad-log-likelihood}
\end{align}
Each contribution $\ell_t(\boldsymbol{\phi})$ depends on the recursively defined conditional variance $\sigma_t^2$. Consequently, evaluating a contribution at a late time index requires propagating the recursion through the preceding observations. Our subsampling methodology exploits this structure by assigning larger sampling probabilities to earlier time indices and using Taylor control variates to reduce estimator variance. Specifically, $q_t(\boldsymbol{\phi})$ is a second-order Taylor approximation to $\ell_t(\boldsymbol{\phi})$ around a central value $\boldsymbol{\phi}^\star$, with $\nabla_{\boldsymbol{\phi}}q_t(\boldsymbol{\phi})$ providing the corresponding control variate for the gradient. The required gradients and Hessians are given in Appendices \ref{sec:VolatilityModels} and \ref{app:derivatives_GARCH}.

For GARCH-family models, evaluating the log densities or their derivatives separately for each $t$ would unnecessarily repeat recursive computations. Appendix \ref{subsec:computational_considerations}, and in particular Algorithm \ref{alg:online_eval}, describes an efficient sequential evaluation scheme.

\subsection{Preliminaries and notation for subsampling-based estimation}

Our methodology builds on subsampling-based log-likelihood estimation methods that use uniform sampling, such as \cite{bardenet2017markov} and \cite{quiroz2019speeding}, but addresses an additional computational challenge arising from recursive likelihoods. In these models, the cost of evaluating an individual log-likelihood contribution $\ell_t(\boldsymbol{\phi})$ depends strongly on its index $t$: evaluating a contribution at a late time point requires propagating the recursion through the preceding observations. Consequently, uniform subsampling may provide little computational saving, since late and computationally expensive contributions are sampled with the same probability as early ones. A natural response is therefore to assign larger sampling probabilities to earlier indices. While constructing an unbiased weighted estimator under decaying sampling probabilities is straightforward, the choice of decay rate is not: probabilities that decay too slowly provide little computational saving, whereas probabilities that decay too aggressively can substantially inflate estimator variance. The main methodological challenge is therefore to design decaying sampling probabilities that exploit the recursive computational structure while controlling the resulting variance inflation.

Our paper focuses on Bayesian inference and conditions on the observed data throughout, consistent with a Bayesian interpretation. Conditional on the observed data and for a fixed $\boldsymbol{\phi}$, an explicit notation for the expectation and variance of a subsampling estimator $\widehat{\ell}(\boldsymbol{\phi})$ of $\ell(\boldsymbol{\phi})$, with randomness governed by the subsampling indices $\mathbf{u}$, is
$$
\mathbb{E}_{\mathbf{u}}\!\left(
\widehat{\ell}(\boldsymbol{\phi})
\,\big|\,
y_{1:T},\boldsymbol{\phi}
\right)
\quad\text{and}\quad
\mathbb{V}_{\mathbf{u}}\!\left(
\widehat{\ell}(\boldsymbol{\phi})
\,\big|\,
y_{1:T},\boldsymbol{\phi}
\right).
$$
For notational simplicity, we suppress the conditioning and subsampling subscript and write $\mathbb{E}(\widehat{\ell}(\boldsymbol{\phi}))$ and $\mathbb{V}(\widehat{\ell}(\boldsymbol{\phi}))$ throughout.

\subsection{Weighted subsampling-based estimator}\label{subsec:WDE}

We extend the difference estimator in \cite{bardenet2017markov,quiroz2019speeding} to allow for unequal sampling probabilities. Let $u_1,\dots,u_m$ be independent subsampling indices with
$\Pr(u_i=t)=p_t>0$, $t=1,\dots,T$. The weighted difference estimator (WDE) is
\begin{align}\label{eq:weighted_diff_estimator}
    \widehat{\ell}_{\mathrm{WDE}}(\boldsymbol{\phi})
    =
    \sum_{t=1}^T q_t(\boldsymbol{\phi})
    +
    \frac{1}{m}\sum_{i=1}^m
    \frac{\ell_{u_i}(\boldsymbol{\phi})-q_{u_i}(\boldsymbol{\phi})}{p_{u_i}},
\end{align}
where $q_t(\boldsymbol{\phi})$ is a control variate for $\ell_t(\boldsymbol{\phi})$. The next lemma gives the expectation and variance of the estimator. The proof is straightforward and omitted.

\begin{lemma}\label{lem:expectation_variance_WDE}
For the estimator in \eqref{eq:weighted_diff_estimator},
\begin{enumerate}
    \item[(i)]
    $\mathbb{E}\!\left(\widehat{\ell}_{\mathrm{WDE}}(\boldsymbol{\phi})\right)
    =\ell(\boldsymbol{\phi})$.

    \item[(ii)]
    With $e_t=\ell_t(\boldsymbol{\phi})-q_t(\boldsymbol{\phi})$ and
    $e=\sum_{t=1}^T e_t$,
    $$
    \mathbb{V}\!\left(\widehat{\ell}_{\mathrm{WDE}}(\boldsymbol{\phi})\right)
    =
    \frac{1}{m}
    \sum_{t=1}^T
    \left(\frac{e_t}{p_t}-e\right)^2p_t.
    $$
\end{enumerate}
\end{lemma}
Following \cite{bardenet2017markov}, we use a second-order Taylor control variate around a central value $\boldsymbol{\phi}^\star$,
\begin{align}\label{eq:2nd_order_cv}
    q_t(\boldsymbol{\phi})
    =
    \ell_t(\boldsymbol{\phi}^\star)
    + \nabla^\top_{\boldsymbol{\phi}}\ell_t(\boldsymbol{\phi}^\star)
      (\boldsymbol{\phi}-\boldsymbol{\phi}^\star)
    + \frac{1}{2}
      (\boldsymbol{\phi}-\boldsymbol{\phi}^\star)^\top
      \nabla^2_{\boldsymbol{\phi}}\ell_t(\boldsymbol{\phi}^\star)
      (\boldsymbol{\phi}-\boldsymbol{\phi}^\star).
\end{align}
After a one-off pass through the data to precompute the required quantities,
$\sum_{t=1}^T q_t(\boldsymbol{\phi})$ can be evaluated at a cost of
$\mathcal{O}(1)$ in $T$. From an asymptotic error perspective, as $T\rightarrow\infty$, the order of the control variate residual
$e_t(\boldsymbol{\phi})=\ell_t(\boldsymbol{\phi})-q_t(\boldsymbol{\phi})$
is determined by the remainder term of the second-order Taylor expansion. The following lemma characterises the bound of the remainder term, which underlies all subsequent results on estimator variance and its stability.

\begin{lemma}\label{lem:taylor_remainder}
Assume that there exist constants $r >0$ and $K>0$ such that $$\mathop{\max\vphantom{\sup}}\limits_{1 \le t \le T}\;
\mathop{\sup\vphantom{\max}}\limits_{\|\boldsymbol{\phi} -\boldsymbol{\phi}^\star\|\leq r}\| \nabla^3_{\boldsymbol{\phi}}\ell_t(\boldsymbol{\phi})\| \leq K,$$
where $\nabla^3_{\boldsymbol{\phi}}\ell_t(\boldsymbol{\phi})$ denotes the third derivative of $\ell_t$, and its norm is defined as
$$\|\nabla^3_{\boldsymbol{\phi}}\ell_t(\boldsymbol{\phi})\|=\sup_{\| \mathbf{u}\|=\| \mathbf{v}\|=\| \mathbf{w}\|=1}|D^3\ell_t(\boldsymbol{\phi})(\mathbf{u},\mathbf{v},\mathbf{w})|,$$
with $D^3\ell_t(\boldsymbol{\phi})(\mathbf{u},\mathbf{v},\mathbf{w})$ denoting the third-order directional derivative of $\ell_t$ at $\boldsymbol{\phi}$ in directions $\mathbf{u}, \mathbf{v}, \mathbf{w}\in \mathbb{R}^{\dim_{\boldsymbol{\phi}}}$. 
Let $C_0>0$ and define the ball $$\mathcal{B}_T=\{\boldsymbol{\phi}: \| \boldsymbol{\phi} - \boldsymbol{\phi}^\star\|\leq C_0T^{-1/2} \},$$
such that $\mathcal{B}_T\subseteq  \{\boldsymbol{\phi}: \| \boldsymbol{\phi} - \boldsymbol{\phi}^\star\|\leq r \}$ for all sufficiently large $T$. Then the control variate residual $e_t(\boldsymbol{\phi})=\ell_t(\boldsymbol{\phi})-q_t(\boldsymbol{\phi})$ satisfies
\begin{align}\label{eq:taylor_remainder_bound}
\mathop{\max\vphantom{\sup}}\limits_{1 \le t \le T}\;
\mathop{\sup\vphantom{\max}}\limits_{\boldsymbol{\phi}\in \mathcal{B}_T}
|e_t(\boldsymbol{\phi})|
= \mathcal{O}(T^{-3/2}).
\end{align}
\end{lemma}

\begin{remark}\label{rem:why_not_quiroz}
Related results in \cite{quiroz2019speeding} also imply that $e_t=\mathcal{O}(T^{-3/2})$, but under substantially stronger assumptions that entail a Bernstein--von Mises framework as in \cite{chen1985asymptotic}. These stronger assumptions are needed in \cite{quiroz2019speeding} to bound the total variation distance between the perturbed posterior and the exact posterior.
\end{remark}

In our methodology, a central objective is to establish the asymptotic order of the variance (with respect to the subsampling randomness, conditional on the observed data and for fixed $\boldsymbol{\phi}$) of the proposed estimators. For this purpose, it is sufficient to verify the assumption in Lemma \ref{lem:taylor_remainder}, without appealing to the stronger Bernstein-von Mises assumptions in \cite{chen1985asymptotic}; see Section \ref{supp:verify_Lemma2_ass} for a verification of the assumption in a standard GARCH($1,1$) model, which is the application considered in Section \ref{sec:VB_application}. Lemma \ref{lem:taylor_remainder} establishes the bound in \eqref{eq:taylor_remainder_bound}, which holds uniformly on the shrinking neighbourhood $\mathcal{B}_T$. In what follows, we write $e_t=\mathcal{O}(T^{-3/2})$ with this local interpretation.

When $p_t=1/T$, corresponding to uniform sampling, \eqref{eq:weighted_diff_estimator} reduces to the difference estimator (DE) of \cite{quiroz2019speeding},
\begin{align}\label{eq:diff_estimator}
    \widehat{\ell}_{\mathrm{DE}}(\boldsymbol{\phi}) = \sum_{t=1}^T q_t(\boldsymbol{\phi}) + \frac{T}{m}\sum_{i=1}^m \left(\ell_{u_i}(\boldsymbol{\phi}) - q_{u_i}(\boldsymbol{\phi})\right).
\end{align}
As discussed above, uniform sampling offers limited computational savings for recursive likelihoods, motivating the decaying sampling probabilities introduced next.

%\newpage
\subsection{Truncated power-law decaying sampling probabilities}\label{subsec:TPD_probs}

Decaying sampling probabilities provide a straightforward way to favour early, computationally cheaper likelihood contributions. However, as shown in Section \ref{supp:decaying_probabilities}, sufficiently aggressive decay can result in an estimator variance that grows unboundedly with $T$, creating a significant challenge. Uniform sampling, although computationally expensive for recursive likelihoods, has favourable variance asymptotics; see Lemma \ref{lem:asymtotic_variance_estimators}(i). This motivates our truncated decaying sampling probability strategy, in which the probabilities decay up to a chosen cut-off $t^\star$ and are subsequently set to be uniform. As shown in Lemma \ref{lem:asymptotic_variance_estimators_truncated}, our proposed scheme retains the favourable asymptotic variance order of uniform sampling.

Let $\gamma$ denote the decay rate of the sampling probabilities and let $\varepsilon(\gamma) \in (0,1)$ be the total probability mass of the tail (typically small) $$\mathcal{T}=\{t^\star +1, \dots, T\},$$ and $1-\varepsilon(\gamma)$ be that of the head $$\mathcal{H}=\{1,\dots, t^\star \}.$$ Our proposed truncated power-law decaying (TPD) weights are
\begin{equation}\label{eq:trunc_probs_TPD}
  p^{\mathrm{TPD}}_t =
  \begin{cases}
    \displaystyle (1-\varepsilon(\gamma))\,\frac{t^{-\gamma}}{\sum_{j=1}^{t^\star} j^{-\gamma}},
    & t \in \mathcal{H},\quad \gamma \geq 0,\ \\
    \displaystyle \frac{\varepsilon(\gamma)}{T - t^\star},
    & t \in \mathcal{T}.
  \end{cases}
\end{equation}
We determine $\varepsilon(\gamma)$ by requiring that $p_{t^\star}$ (the last head probability) equals $p_{t^\star+1}$ (the probability of each tail element), which gives
\begin{equation}\label{eq:eps_gamma}
\varepsilon(\gamma) =
\frac{{t^\star}^{-\gamma}(T-t^\star)}
{{t^\star}^{-\gamma}(T-t^\star) + \sum_{j=1}^{t^\star}j^{-\gamma}}.
\end{equation}
Figure \ref{fig:TPD_probabilities} illustrates our truncated power-law decaying probability scheme. 

\begin{figure}[H]
    \centering
    \includegraphics[width=0.90\textwidth]{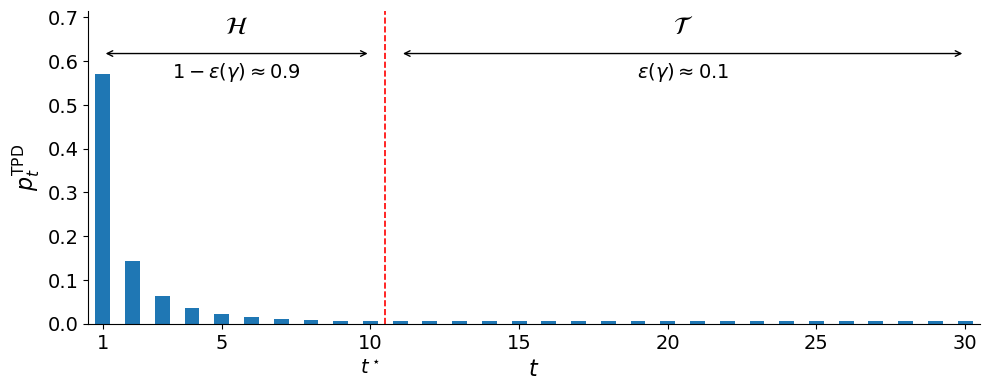}
    \caption{Truncated power-law decaying (TPD) sampling probabilities in \eqref{eq:trunc_probs_TPD} with $T=30$, $t^\star = 10$, and $\gamma$ such that the tail mass is $\varepsilon(\gamma)\approx 0.1$. The vertical line indicates the transition from the head $\mathcal{H}$ $(t\leq t^\star)$ to the tail $\mathcal{T}$ $(t>t^\star)$.}
    \label{fig:TPD_probabilities}
\end{figure}

\begin{remark}
  When $\gamma=0$ and $\varepsilon(0)=\frac{T-t^\star}{T}$, the probabilities in \eqref{eq:trunc_probs_TPD} are exactly uniform, i.e.\ $p_t=1/T$ for all $t$. As discussed in the introduction, this uniform sampling regime offers little computational benefit for recursively defined likelihoods. Section \ref{subsec:computational_cost} analyses the implications for the expected computational cost as $\gamma$ decreases towards zero.
\end{remark}

The next lemma establishes several basic properties of the truncated power-law decaying sampling scheme. In particular, it shows the monotonicity of $\varepsilon(\gamma)$, that each tail element has probability no larger than that of uniform sampling, and that each head probability is at least as large as the per-element tail probability.
\begin{lemma}\label{lem:tail_smaller}
For the truncated sampling scheme in \eqref{eq:trunc_probs_TPD},
\begin{enumerate}
    \item[(i.)] $\varepsilon(\gamma)$ is non-increasing in $\gamma \geq 0$, and strictly decreasing if $t^\star \geq 2$.
    \item[(ii.)] The per-element tail probabilities obey $$\frac{\varepsilon(\gamma)}{T-t^\star} \leq \frac{1}{T} \quad \text{with equality if and only if } \gamma=0\,\,\text{(uniform sampling)}.$$
    \item[(iii.)] The  per-element head probabilities obey
    \begin{align*}
        p_t \ge \frac{\varepsilon(\gamma)}{T-t^\star},
    \,\, t \in \mathcal{H}.        
    \end{align*}
\end{enumerate}
\end{lemma}

Lemma \ref{lem:tail_smaller}(ii) shows that the probability assigned to each tail element under our scheme is always smaller than that of uniform sampling. This reduction in tail probabilities comes at a cost; under a homogeneity assumption on the control variate residuals, the next lemma shows that increasing $\gamma$, which shrinks the tail mass $\varepsilon(\gamma)$ (Lemma \ref{lem:tail_smaller}(i)), increases the contribution of the tail indices to the variance of the estimator.
\begin{lemma}\label{lem:tail_smaller_higher_variance}
    Consider the estimator $\widehat{\ell}_{\mathrm{WDE}}$ in \eqref{eq:weighted_diff_estimator} under homogeneous control variate residuals, $$e_t(\boldsymbol{\phi}) = \ell_t(\boldsymbol{\phi}) - q_t(\boldsymbol{\phi}) = \overline{e}(\boldsymbol{\phi}), \quad t= 1, \dots, T,$$
    where $\overline{e}(\boldsymbol{\phi})$ denotes the common residual value.
    Let $\mathbb{V}_\mathcal{T}(\gamma)$ denote the contribution of the tail $\mathcal{T}$ to the variance of $\widehat{\ell}_{\mathrm{WDE}}$ under the truncated sampling scheme in \eqref{eq:trunc_probs_TPD} with decay parameter $\gamma$. Then, for any $0 \leq \gamma_1 < \gamma_2$,
    $$\mathbb{V}_{\mathcal{T}}(\gamma_1) < \mathbb{V}_{\mathcal{T}}(\gamma_2).$$
\end{lemma}

Importantly, truncation alone does not resolve the variance-cost trade-off. Although the probabilities are uniform beyond $t^\star$, increasing $\gamma$ reduces the probability assigned to each tail observation and can therefore substantially inflate the estimator variance, as shown in Lemma \ref{lem:tail_smaller_higher_variance}. Thus, even within the class of truncated decaying weights, allowing the decay to become too aggressive can lead to an unstable estimator. To control this behaviour and stabilise the estimator, we impose the per-element tail probability floor
\begin{align}\label{eq:safe_guard_floor}
    \frac{\varepsilon(\gamma)}{T-t^\star} & \geq c\frac{1}{T}, \quad 0 < c \leq 1,
\end{align}
where $c$ is a tail floor fraction controlling the minimum per-element tail probability relative to uniform sampling. This condition is equivalent to imposing the tail mass floor 
$$\varepsilon(\gamma) \geq c\frac{T-t^\star}{T}.$$ 
The case $c=1$ corresponds to uniform sampling over all indices. There is a direct relationship between the tail floor fraction $c$ and the decay rate $\gamma$; imposing a larger tail floor restricts the admissible decay rate. This relationship is exploited in Section \ref{subsec:balancing_variance_and_computational_cost} to tune $c$ rather than $\gamma$. Through its effect on $\gamma$, $c$ also controls the tail mass $\varepsilon(\gamma)$, which plays a central role in the expected computational cost analysed in Section \ref{subsec:computational_cost}.

We now quantify the variance implications of imposing the safeguard. The next lemma provides a worst-case bound on how much the variance of the weighted difference estimator can be inflated relative to the baseline difference estimator (with uniform weights), expressed as a function of the tail floor fraction $c$.
\begin{lemma}\label{lem:upper_bound_variance}
Let $0< c \leq 1$ and suppose that $p_t \geq \frac{c}{T}$ for all $t$ in the estimator $\widehat{\ell}_\mathrm{WDE}$ defined in \eqref{eq:weighted_diff_estimator},  with $e_t = \ell_t(\boldsymbol{\phi})-q_t(\boldsymbol{\phi})$ and $e=\sum_{t=1}^T e_t$. Then, relative to $\widehat{\ell}_\mathrm{DE}$ in \eqref{eq:diff_estimator},
\begin{equation}\label{eq:variance_ratio_bound}
    \frac{\mathbb{V}(\widehat{\ell}_\mathrm{WDE})}{\mathbb{V}(\widehat{\ell}_\mathrm{DE})} \leq \frac{\frac{1}{c}-\rho}{1-\rho}, \quad \text{ where } \rho = \frac{e^2}{T\sum_{t=1}^Te^2_t} \in (0, 1).
\end{equation}
\end{lemma}

See Section~\ref{supp:stabilised_sampling_properties} for further discussion of the following properties. In practice, $\rho$ is typically small, so the variance inflation bound in Lemma \ref{lem:upper_bound_variance} is well approximated by $1/c$. Empirically, the variance ratio closely tracks $1/c$ and is largely insensitive to the values of $\boldsymbol{\phi}$ explored by the algorithm. We therefore introduce the design parameter $R_{\max}=1/c_{\min}$ and restrict $c\geq c_{\min}$, providing direct control over the allowable variance inflation relative to uniform sampling. This restriction also prevents the tail probability $\varepsilon$ from becoming too small, which can hinder exploration of tail indices under a fixed computational budget. The tuning strategy is described in Section \ref{subsec:balancing_variance_and_computational_cost}.

Finally, the next lemma shows that the safeguard preserves the favourable asymptotic variance order of the uniform-sampling difference estimator in \eqref{eq:diff_estimator}.
\begin{lemma}\label{lem:asymptotic_variance_estimators_truncated}
Suppose that $e_t = \mathcal{O}(T^{-3/2})$. Then, for the weighted difference estimator in \eqref{eq:weighted_diff_estimator}, with $p_t$ given by \eqref{eq:trunc_probs_TPD} and the safeguard in \eqref{eq:safe_guard_floor},
$$\mathbb{V}\left(\widehat{\ell}_{\mathrm{WDE}}(\boldsymbol{\phi})\right) = \frac{1}{m}\mathcal{O}\left(T^{-1}\right).$$ 
\end{lemma}

\subsection{Expected computational cost}\label{subsec:computational_cost}

A naive implementation of \eqref{eq:weighted_diff_estimator} evaluates the sampled log-likelihood contributions separately and unnecessarily repeats recursive computations. Instead, let
$u_{\max}=\max\{u_1,\dots,u_m\}$ and evaluate the recursion sequentially only up to $u_{\max}$, using Algorithm \ref{alg:online_eval} with $T$ replaced by $u_{\max}$. Since the computational cost is linear in the length of the recursion, we use $\mathbb{E}(u_{\max})$ as our measure of expected computational cost.

The tail floor fraction $c$ introduced in Section \ref{subsec:TPD_probs} controls variance inflation through Lemma \ref{lem:upper_bound_variance} and, through its effect on the tail mass $\varepsilon(\gamma)$, also affects expected computational cost. The next lemma characterises $\mathbb{E}(u_{\max})$ and its dependence on $\varepsilon$ and $m$.
\begin{lemma}\label{lem:average_cost}
Let $u_1, \dots, u_m$ be independent with $u_i \in \{1, \dots, T\}$ and $\Pr(u_i=t)=p_t$. Define $u_{\max} = \max \{ u_1, \dots, u_m\}$. Then the following hold.
\begin{enumerate}
    \item[(i)]
    \[
    \mathbb{E}(u_{\max})
    =
    T-\sum_{k=1}^T
    \left(
    \sum_{t=1}^{k-1} p_t
    \right)^m
    \leq T.
    \]

    \item[(ii)]
    Let $p_t$ be the truncated decaying probabilities in \eqref{eq:trunc_probs_TPD}, and denote the tail mass by
    $\varepsilon=\sum_{t\in\mathcal{T}}p_t=\varepsilon(\gamma)$. Then
    \begin{align}\label{eq:exact_compute_bound_TPD}
        \mathbb{E}(u_{\max})
        \leq
        t^\star
        +
        (T-t^\star)
        \left(
        1-(1-\varepsilon)^m
        \right).
    \end{align}
    The upper bound in \eqref{eq:exact_compute_bound_TPD} is strictly increasing in $\varepsilon\in(0,1)$ for any $m\geq1$.

    \item[(iii)]
    For the truncated decaying probabilities in \eqref{eq:trunc_probs_TPD}, $\mathbb{E}(u_{\max})$ is strictly increasing in $m$.
\end{enumerate}
\end{lemma}
When $\gamma=0$, corresponding to uniform sampling,
$\varepsilon=(T-t^\star)/T$, and \eqref{eq:exact_compute_bound_TPD} becomes
\begin{align}\label{eq:exact_compute_bound_uniform}
    \mathbb{E}(u_{\max})
    \leq
    t^\star+(T-t^\star)
    \left(1-\left(\frac{t^\star}{T}\right)^m\right).
\end{align}
Since $\varepsilon(\gamma)$ is strictly decreasing in $\gamma$ for $t^\star\geq2$ by Lemma \ref{lem:tail_smaller}(i), and the bound in Lemma \ref{lem:average_cost}(ii) is strictly increasing in $\varepsilon$, any $\gamma>0$ yields a strictly smaller upper bound than uniform sampling. Thus, the tail floor fraction $c$ jointly governs the variance-cost trade-off. Decreasing $c$ permits more aggressive decay, reducing the upper bound on expected computational cost but allowing greater variance inflation. Conversely, increasing $c$ restricts the decay, increasing the computational-cost bound while reducing the bound on variance inflation. The tuning strategy in Section \ref{subsec:balancing_variance_and_computational_cost} balances these competing objectives.

Figure \ref{fig:E_umax_c_m} illustrates the expected computational cost as a function of $(c,m)$ for three values of $t^\star$. The figure shows that smaller values of $c$ and $m$ can yield substantial computational savings. Moreover, for the values of $t^\star$ considered in Figure \ref{fig:E_umax_c_m}, the expected computational cost landscape is relatively insensitive to the choice of $t^\star$, with very similar level sets across the three panels. This motivates treating $t^\star$ as fixed in our tuning procedure. Note that the largest computational savings occur for small values of $c$ and $m$, which may result in excessive estimator variance. The tuning strategy in Section \ref{subsec:balancing_variance_and_computational_cost} addresses this trade-off by imposing a variance tolerance and the safeguard $c\geq c_{\min}$.

\begin{figure}[H]
    \centering
    \includegraphics[width=1.00\textwidth]{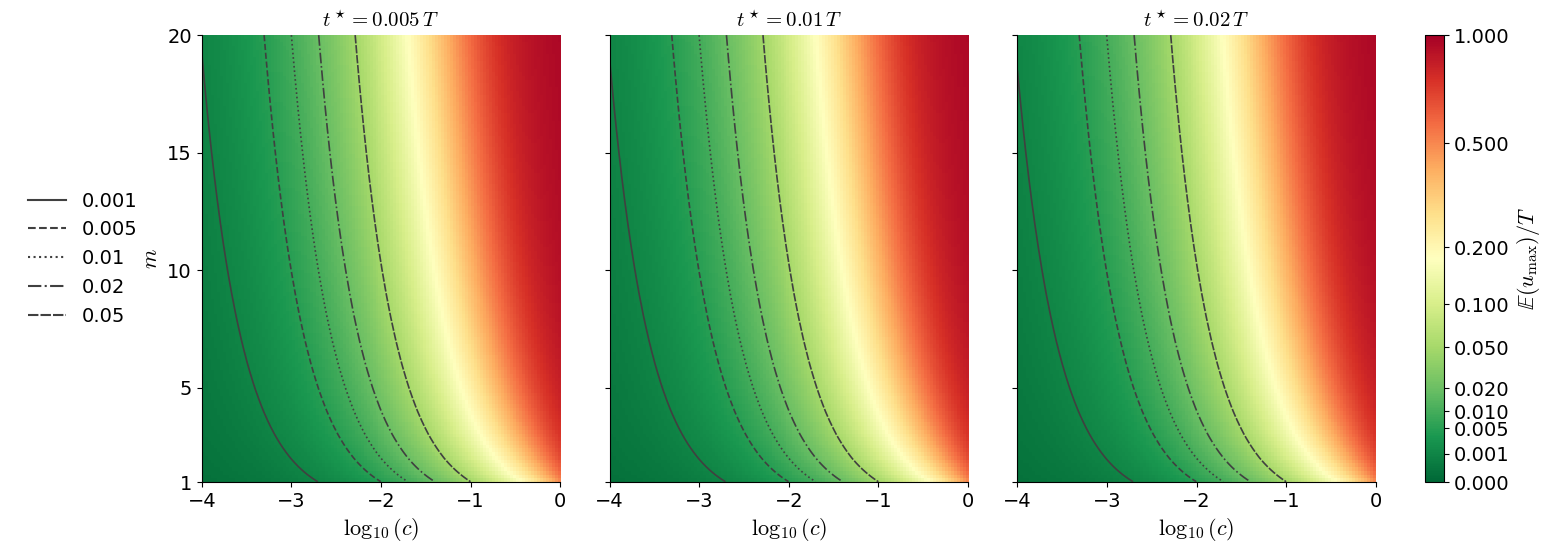}
    \caption{$\mathbb{E}(u_{\max})$ in Lemma \ref{lem:average_cost}(i) (normalised by $T$, with $T=100{,}000$) as a function of $c$ and $m$ (where $m$ is treated as continuous for visualisation), shown for three values of $t^\star$ (see panel titles). The value of $T$ corresponds to that in the applications, as does $t^\star$ in the middle panel. The colour scale (right) represents the expected computational cost ratio, where red indicates no compute savings relative to the full dataset of size $T$, and green indicates increasing compute savings. Superimposed level curves (left) correspond to fixed values of $\mathbb{E}(u_{\max})/T$; for each curve, hyperparameter combinations lying below it yield at least as much compute savings.}
    \label{fig:E_umax_c_m}
\end{figure}

Section \ref{supp:heuristic_scaling} of the supplementary material discusses the asymptotic interpretation of the computational-cost bounds and provides additional illustrations of expected computational cost. It also presents a heuristic argument suggesting that, for fixed $m$ and $t^\star$, the scaling $\varepsilon=\mathcal{O}(T^{-1})$ can keep the expected computational cost bounded as $T$ grows. This scaling, however, introduces an important caveat. As $\varepsilon$ decreases, the sampling probabilities assigned to tail observations become smaller, potentially leading to substantial variance inflation and, ultimately, an estimator that is impractical for inference. In our applications, however, we obtained substantial computational savings while maintaining acceptable estimator variance.

\subsection{Cost-variance considerations and tuning guidelines }\label{subsec:balancing_variance_and_computational_cost}

A natural first attempt to balance computational effort and statistical efficiency is to consider the cost-variance product,
\begin{align}\label{eq:cost_variance}
    \mathrm{CV}
    =
    \mathbb{E}(u_{\max})
    \mathbb{V}(\widehat{\ell}_{\mathrm{WDE}}).
\end{align}
However, in our setting, this criterion typically favours uniform sampling ($c=1$), since the reduction in expected computational cost achieved with decaying probabilities does not sufficiently offset the associated increase in estimator variance. Nevertheless, the inferential algorithm in which the estimator is embedded may still benefit from substantial computational savings despite the higher variance, provided it can tolerate the additional noise. More fundamentally, this behaviour reflects a mismatch between the cost-variance product and the design principle underlying our decaying weights: expected computational cost is the quantity we seek to reduce, while variance is an unavoidable by-product that must be controlled. This leads naturally to a constrained optimisation formulation in which expected computational cost is minimised subject to a variance tolerance and the safeguard $c\geq c_{\min}$, rather than combining cost and variance into a single objective.

Our tuning strategy treats $t^\star$ as fixed and optimises the tail floor fraction $c$ and subsample size $m$. Fixing $t^\star$ avoids introducing an additional tuning dimension and is motivated by the relatively weak sensitivity of expected computational cost to its choice, as illustrated in Figure \ref{fig:E_umax_c_m}. For each candidate $c$, we select the largest decay rate $\gamma_{\max}(c)$ compatible with the tail probability floor introduced in Section \ref{subsec:TPD_probs}. Since $\varepsilon(\gamma)$ is strictly decreasing in $\gamma$ by Lemma \ref{lem:tail_smaller}(i), this rate is determined by
\[
    \varepsilon(\gamma_{\max}(c))
    =
    c\frac{T-t^\star}{T},
\]
which can be solved numerically using root-finding. Thus, $c$ determines the sampling probabilities $p_t(c)$ and hence the expected computational cost and estimator variance for a given subsample size, while $m$ is selected to satisfy the variance tolerance for those probabilities. We now formulate the resulting constrained optimisation problem.

We minimise expected computational cost subject to a variance tolerance $V>0$ and the safeguard $c\geq c_{\min}$,
\begin{align}\label{eq:miminises_compute_subject_variance}
    \min_{c\in[c_{\min},1],\,m\in\{1,\dots,T\}}
    \mathbb{E}(u_{\max};c,m)
    \quad\text{subject to}\quad
    \mathbb{V}\left(\widehat{\ell}_{\mathrm{WDE}}(\boldsymbol{\phi});c,m\right)
    \leq V.
\end{align}
Using Lemma \ref{lem:expectation_variance_WDE}(ii), the estimator variance can be written as
\begin{equation}\label{eq:var_WDE_c_m}
    \mathbb{V}\left(\widehat{\ell}_{\mathrm{WDE}}(\boldsymbol{\phi});c,m\right)
    =
    \frac{\sigma^2(c;\boldsymbol{\phi})}{m},
    \quad
    \sigma^2(c;\boldsymbol{\phi})
    =
    \sum_{t=1}^T
    \left(
    \frac{e_t(\boldsymbol{\phi})}{p_t(c)}
    -
    e(\boldsymbol{\phi})
    \right)^2p_t(c).
\end{equation}
In practice, we evaluate the variance constraint at a reference value $\boldsymbol{\phi}^\dagger$ obtained from a short pilot run, as described below, and write $\sigma^2(c)=\sigma^2(c;\boldsymbol{\phi}^\dagger)$.

For a fixed $c$, expected computational cost is strictly increasing in $m$ by Lemma \ref{lem:average_cost}(iii), whereas the estimator variance decreases as $1/m$. Consequently, the optimal subsample size for a given $c$ is the smallest integer satisfying the variance constraint,
\[
    m_V(c)
    =
    \min\left\{
    m\in\{1,\dots,T\}:
    \frac{\sigma^2(c)}{m}\leq V
    \right\}
    =
    \max\left\{
    1,\left\lceil\frac{\sigma^2(c)}{V}\right\rceil
    \right\},
\]
provided that $\sigma^2(c)\leq VT$. The feasible values of $c$ are therefore
\[
    \mathcal{C}_V
    =
    \left\{
    c\in[c_{\min},1]:
    \sigma^2(c)\leq VT
    \right\}.
\]
Substituting $m_V(c)$ into the objective reduces the original two-dimensional optimisation problem to
\begin{align}\label{eq:miminises_compute_subject_variance_1dim}
    c^\star
    =
    \arg\min_{c\in\mathcal{C}_V}
    \mathbb{E}\left(u_{\max};c,m_V(c)\right),
\end{align}
with corresponding subsample size $m^\star=m_V(c^\star)$. The reduced problem can be solved numerically, for example by evaluating the objective over a grid of feasible values of $c$. Unlike $m$, which is chosen as the smallest feasible integer for each $c$, the optimal value of $c$ need not lie on the boundary of its feasible set.

We calibrate the variance tolerance $V$ using the design parameter $R_{\max}=1/c_{\min}>1$ introduced in Section \ref{subsec:TPD_probs}. Let $\{\boldsymbol{\phi}^{(j)}\}_{j=1}^{N_{\mathrm{tune}}}$ denote parameter draws obtained from a short pilot Markov chain Monte Carlo run using the full data. For each draw, we compute
\begin{align}\label{eq:scaled_intrinsic_variability}
    V^{(j)} = R_{\max}\sigma^2(1;\boldsymbol{\phi}^{(j)}).
\end{align}
Here, $\sigma^2(1;\boldsymbol{\phi}^{(j)})$ represents the intrinsic variability of a single weighted subsample contribution under uniform sampling and thus reflects the accuracy of the control variates. The factor $R_{\max}$ specifies the nominal allowable variance inflation relative to this baseline. As discussed in Section \ref{supp:stabilised_sampling_properties}, the variance inflation ratio is empirically approximately insensitive to $\boldsymbol{\phi}$, motivating the use of a common $R_{\max}$ across the pilot draws. We summarise the resulting variance tolerances by their median,
\[
    V=\mathrm{median}\left(V^{(1)},\dots,V^{(N_{\mathrm{tune}})}\right).
\]
We select the reference parameter value $\boldsymbol{\phi}^\dagger$ as the pilot draw whose variance tolerance is closest to the median,
\begin{align}\label{eq:phi_tune}
    \boldsymbol{\phi}^\dagger
    =
    \boldsymbol{\phi}^{(j^\star)},
    \qquad
    j^\star
    =
    \arg\min_{1\leq j\leq N_{\mathrm{tune}}}
    \left|V^{(j)}-V\right|.
\end{align}
For uniform sampling ($c=1$), the same calibration procedure applies with $R_{\max}=1$.

Finally, we emphasise that the proposed tuning strategy is not globally optimal. Fixing $t^\star$ avoids an additional tuning dimension but may sacrifice computational efficiency, since jointly optimising $t^\star$, $c$, and $m$ could yield further savings. Moreover, the variance constraint is evaluated at a reference parameter value obtained from a pilot run rather than enforced uniformly over the parameter space. Consequently, the computational speed-ups reported in our applications should be interpreted as achievable gains rather than upper limits on the computational efficiency of the proposed methodology.

\subsection{Weighted subsampling-based gradient estimator}\label{subsec:grad_WDE}
The weighted subsampling methodology developed above extends directly to estimation of the gradient of the log-likelihood in \eqref{eq:grad-log-likelihood}, which is required in gradient-based algorithms such as variational Bayes.

The weighted difference estimator of the gradient follows directly from the weighted difference estimator of the log-likelihood in \eqref{eq:weighted_diff_estimator},  
\begin{align}\label{eq:weighted_diff_estimator_grad}
    \nabla_{\boldsymbol{\phi}}\widehat{\ell}_{\mathrm{WDE}}(\boldsymbol{\phi}) = \sum_{t=1}^T \nabla_{\boldsymbol{\phi}} q_t(\boldsymbol{\phi}) + \frac{1}{m}\sum_{i=1}^m \frac{\nabla_{\boldsymbol{\phi}}\ell_{u_i}(\boldsymbol{\phi}) - \nabla_{\boldsymbol{\phi}}q_{u_i}(\boldsymbol{\phi})}{p_{u_i}},
\end{align}
where $\nabla_{\boldsymbol{\phi}}q_t(\boldsymbol{\phi})$ is a control variate for the gradient contribution $\nabla_{\boldsymbol{\phi}}\ell_t(\boldsymbol{\phi})$. Differentiating the second-order Taylor control variate in \eqref{eq:2nd_order_cv} gives
\begin{align*}
    \nabla_{\boldsymbol{\phi}} q_t(\boldsymbol{\phi}) = \nabla_{\boldsymbol{\phi}}\ell_t(\boldsymbol{\phi}^\star) + \nabla^2_{\boldsymbol{\phi}} \ell_t(\boldsymbol{\phi}^\star) (\boldsymbol{\phi} -\boldsymbol{\phi}^\star).
\end{align*}
The next lemma establishes unbiasedness and gives the gradient estimator's covariance matrix.
\begin{lemma}\label{lem:expectation_variance_WDE_grad}
For the estimator in \eqref{eq:weighted_diff_estimator_grad},
\begin{enumerate}
    \item[(i)] $\mathbb{E}\left(\nabla_{\boldsymbol{\phi}}\widehat{\ell}_{\mathrm{WDE}}(\boldsymbol{\phi})\right) = \nabla_{\boldsymbol{\phi}}\ell(\boldsymbol{\phi})$, with $\nabla_{\boldsymbol{\phi}}\ell(\boldsymbol{\phi})$ in \eqref{eq:grad-log-likelihood}.
    \item[(ii)] With $\mathbf{e}_t = \nabla_{\boldsymbol{\phi}}\ell_t(\boldsymbol{\phi})-\nabla_{\boldsymbol{\phi}}q_t(\boldsymbol{\phi})\in \mathbb{R}^{\dim_{\boldsymbol{\phi}}}$ and $\mathbf{e}=\sum_{t=1}^T \mathbf{e}_t$, $$\mathbb{V}\left(\nabla_{\boldsymbol{\phi}}\widehat{\ell}_{\mathrm{WDE}}(\boldsymbol{\phi})\right) =  \frac{1}{m}\sum_{t=1}^T\left(\frac{\mathbf{e}_t}{p_t} - \mathbf{e}\right) \left(\frac{\mathbf{e}_t}{p_t} - \mathbf{e}\right)^\top p_t.$$
\end{enumerate}
\end{lemma}
In our methodology, the sampling probabilities are tuned using the variance of the log-likelihood estimator, as described in Section \ref{subsec:balancing_variance_and_computational_cost}. These probabilities and the corresponding subsamples are then reused to estimate the gradient, without imposing a separate constraint on its variance.

\section{Application: Subsampling variational Bayes inference for a GARCH($1,1$) model}\label{sec:VB_application}

\subsection{Background, model, and data}
We showcase our stabilised weighted subsampling methodology by considering a doubly stochastic variational inference approach \citep{Titsias2014doubly} for a GARCH($1,1$) model with Gaussian measurement errors, given by
\begin{align*}
    y_t &= \mu + z_t, \,\, z_t = \sigma_t \varepsilon_t, \,\, \varepsilon_t \overset{\mathrm{ind}}{\sim} N(0,1), \,\, \mu \in \mathbb{R}, \,\, \sigma_t > 0, \\
    \sigma_t^2 &= \omega + \alpha z_{t-1}^2 + \beta \sigma_{t-1}^2,
\end{align*}
where $\omega>0$, $\alpha\geq 0$, and $\beta\geq 0$. The GARCH($1,1$) specification is a widely used benchmark model for conditional volatility in financial time series \citep{hansen2005GARCH}. Moreover, under the parameterisation described in Section \ref{subsec:spec_settings_App1}, the assumption underlying Lemma \ref{lem:taylor_remainder} is verified for this model in Section \ref{supp:verify_Lemma2_ass}, making it a natural and theoretically well-supported benchmark for our methodology.

We fit variational posteriors to a GARCH($1,1$) model whose parameters are reparameterised to an unrestricted space, using log-returns of the Dow Jones Industrial Average, which comprises 30 large US firms. Figure \ref{fig:DowJones} displays the data, with further details provided in the figure caption. 

\begin{figure}[h]
    \centering
    \includegraphics[width=0.65\textwidth]{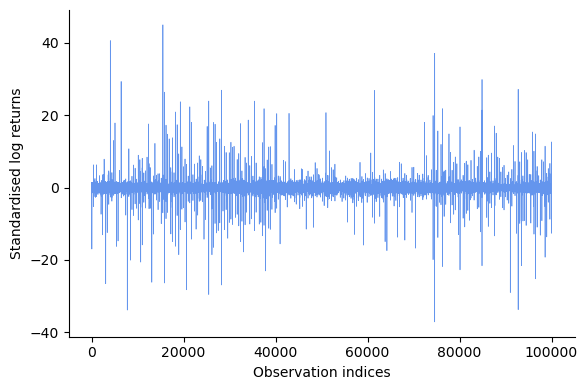}
    \caption{Observations $y_t$ for the Dow Jones 30 index, constructed as one-minute log-returns over the period 2023-05-15 to 2024-06-28 ($T=100{,}000$) and rescaled to have unit sample standard deviation.}
    \label{fig:DowJones}
\end{figure}

The figure shows that the log-returns can exhibit deviations of 30--40 unconditional standard deviations from their unconditional mean, which is difficult to reconcile with a homoskedastic Gaussian specification. Clear volatility clustering and potential heavy-tailed behaviour are also evident, motivating the use of GARCH dynamics and more flexible measurement error distributions.

\subsection{Gaussian variational approximation}\label{subsec:GVA}
We consider variational approximations to the posterior distribution
\begin{align}\label{eq:posterior}
    \pi_T(\boldsymbol{\phi}) & = \frac{p(\mathbf{y}_{1:T}|\boldsymbol{\phi})\pi(\boldsymbol{\phi})}{p(\mathbf{y}_{1:T})}, \,\, p(\mathbf{y}_{1:T}) = \int_{\boldsymbol{\phi}}p(\mathbf{y}_{1:T}|\boldsymbol{\phi})\pi(\boldsymbol{\phi})d\boldsymbol{\phi},
\end{align}
where $\pi(\boldsymbol{\phi})$ denotes the prior distribution, $p(\mathbf{y}_{1:T}|\boldsymbol{\phi})$ is the likelihood function in \eqref{eq:likelihood} conditional on the observed data $\mathbf{y}_{1:T}=(y_1,\dots,y_T)^\top$, and $p(\mathbf{y}_{1:T})$ is the marginal likelihood. Specifically, we approximate the posterior distribution in \eqref{eq:posterior} using a Gaussian variational approximation with density
$$\mathcal{Q}_{\boldsymbol{\lambda}}(\boldsymbol{\phi}) = \mathcal{N}(\boldsymbol{\phi}|\boldsymbol{\mu}_{\mathrm{VB}},\boldsymbol{\Sigma}_{\mathrm{VB}}),$$
where the variational mean $\boldsymbol{\mu}_{\mathrm{VB}}$ and covariance matrix $\boldsymbol{\Sigma}_{\mathrm{VB}}$ constitute the variational parameters, which we collectively denote by
$\boldsymbol{\lambda}$. 

Gaussian variational approximation aims to find the variational parameters that minimise the Kullback-Leibler divergence from the variational posterior $\mathcal{Q}_{\boldsymbol{\lambda}}(\boldsymbol{\phi})$ to the posterior $\pi_T(\boldsymbol{\phi})$. Equivalently, the problem can be formulated as maximising the so-called evidence lower bound (ELBO); see, for example, \cite{blei2017variational}. The ELBO is defined as
\begin{align}\label{eq:lower_bound}
\mathrm{ELBO}(\boldsymbol{\lambda})=\mathbb{E}_{\mathcal{Q}_{\boldsymbol{\lambda}}}\left(\log p(\mathbf{y}_{1:T}|\boldsymbol{\phi}) + \log \pi(\boldsymbol{\phi})-\log \mathcal{Q}_{\boldsymbol{\lambda}}(\boldsymbol{\phi})\right),
\end{align}
where the expectation is taken with respect to the variational distribution with fixed variational parameters. The ELBO satisfies the inequality $$\log p(\mathbf{y}_{1:T})\geq \mathrm{ELBO}(\boldsymbol{\lambda}),$$
for all $\boldsymbol{\lambda}$, with the marginal likelihood $p(\mathbf{y}_{1:T})$ defined in \eqref{eq:posterior}. The optimal variational approximation is given by $\mathcal{Q}_{\boldsymbol{\lambda}^\star}(\boldsymbol{\phi})$, where
\begin{align}\label{eq:optim_lower_bound}
    \boldsymbol{\lambda}^\star =  \arg \max_{\boldsymbol{\lambda}} \mathrm{ELBO}(\boldsymbol{\lambda}).
\end{align}
Maximising \eqref{eq:optim_lower_bound} poses several computational challenges. First, the number of variational parameters grows quadratically with the dimension of $\boldsymbol{\phi}$ due to the covariance matrix, motivating parsimonious covariance structures \citep{ong2018gaussian,quiroz2023gauss}. Second, the expectations appearing in the ELBO and its gradient generally cannot be evaluated analytically and must be estimated using Monte Carlo methods, often employing the reparameterisation trick \citep{Kingma2013,xu2019variance}. Third, even after reparameterisation, repeated evaluation of the full-data log-likelihood and its gradient can dominate computational cost for large datasets. To address this challenge, \cite{Titsias2014doubly} propose estimating the log-likelihood and its gradient using uniform data subsampling within a doubly stochastic variational framework. We incorporate our stabilised weighted subsampling methodology into this framework by replacing uniform sampling with the truncated power-law decaying probabilities proposed in Section \ref{subsec:TPD_probs}, retaining unbiased estimators of the log-likelihood and its gradient. We additionally follow \cite{gunawan2017fast} by restricting the variational covariance matrix to a rank-1 plus isotropic structure, yielding a tractable expression for the natural gradient \citep{amari1998natural}, which incorporates the local geometry induced by the Fisher information metric \citep{martens2020insights}. For comparison, we also consider full-data variational Bayes, for which the only source of stochasticity is the Monte Carlo estimation of the expectation using the reparameterisation trick; we refer to this as the singly stochastic approach.

\subsection{Specifications and settings}\label{subsec:spec_settings_App1}
To improve the Gaussian variational approximation and to facilitate the ELBO optimisation, we work with an unconstrained parameterisation 
$$\boldsymbol{\phi}=(\phi_\mu,\phi_\omega,\phi_{\alpha_1} ,\phi_{\beta_1})^\top,$$
where
\begin{equation}\label{eq:transformations_main_paper}
\phi_\mu=\mu, \,\, \phi_{\omega}=\log(\omega), \,\, \phi_{\alpha_1}=\log(\alpha_1), \,\, \phi_{\beta_1}=\log(\beta_1).
\end{equation}
These transformations map the original parameters $\mu\in \mathbb{R},\omega > 0, \alpha_1 \geq 0$, and $\beta_1 \geq 0$, to an unrestricted space. 

For interpretability, the prior distribution is imposed on the original parameter vector, $$\boldsymbol{\theta} =(\mu,\omega,\alpha_1,\beta_1)^\top.$$
Specifically, we assign independent marginal priors,
$$\mu \sim \mathcal{N}(0,10^2),\,\, \omega \sim \text{Half-Normal}(1),\,\,\alpha_1 \sim \text{Half-Normal}(0.2),\,\, \beta_1 \sim \text{Half-Normal}(0.8),$$
where \text{Half-Normal} denotes the half-normal distribution parameterised in terms of its scale (standard deviation). The joint prior is supported on the stationary region
$\alpha_1 + \beta_1 < 1$ and assigns zero density outside this region. The implied prior distribution on the reparameterised parameter vector $\boldsymbol{\phi}$ is induced via the transformations in \eqref{eq:transformations_main_paper}.

We follow \cite{gunawan2017fast} and parameterise the variational covariance matrix using a rank-1 plus isotropic structure,
$$\boldsymbol{\Sigma}_{\mathrm{VB}}=\mathbf{B}\mathbf{B}^\top + d^2\mathbf{I}, \,\, \mathbf{B} \in \mathbb{R}^{\dim_{\boldsymbol{\phi}}\times 1}, \,\, d\in \mathbb{R},$$
where $\mathbf{I}\in\mathbb{R}^{\dim_{\boldsymbol{\phi}}\times \dim_{\boldsymbol{\phi}}}$ is the identity matrix. For optimisation of the ELBO in \eqref{eq:lower_bound}, the expectation appearing in its gradient is reparameterised using the reparameterisation trick and estimated via Monte Carlo sampling with 1 draw per iteration. We run all optimisers for $5{,}000$ iterations. The variational mean $\boldsymbol{\mu}_{\mathrm{VB}}$ is initialised at the maximum a posteriori (MAP) estimate, obtained using the basin-hopping global optimisation algorithm \citep{wales1997global} with five iterations to reduce the risk of converging to poor local modes. The variational covariance parameters are initialised with $\mathbf{B}=\mathbf{0}$ and $d^2$ is obtained from a down-scaled version of the diagonal of the posterior covariance estimate from a Laplace approximation. These MAP and Laplace calculations are common to all variational methods and are therefore omitted from the reported CPU times. For stability, the first $500$ iterations use the Euclidean gradient together with the Adam optimiser \citep{KingmaBa2015Adam} with its default hyperparameter settings, while the remaining iterations employ the natural gradient derived in \cite{gunawan2017fast}. To monitor optimisation progress, the reparameterised ELBO is estimated every 25 iterations using 5 Monte Carlo draws. The final variational parameter values are computed using Polyak averaging \citep{polyak1992acceleration} over the last 500 iterations.

For our stabilised weighted subsampling methodology, we construct the control variates at the MAP estimate $\boldsymbol{\phi}^\star$ and include their construction time in the total computational cost. We apply the tuning guidelines in Section \ref{subsec:balancing_variance_and_computational_cost} for three values of $R_{\max}=1/c_{\min}$, corresponding to $c_{\min}=0.001,0.01,0.1$. The algorithm is tuned at $\boldsymbol{\phi}^{\dagger}$ in \eqref{eq:phi_tune} using $N_{\mathrm{tune}}=20$ samples obtained by thinning a pilot full-data Markov chain Monte Carlo run of length 100, whose computational cost and associated tuning are also included. The resulting variance tolerances are approximately $V \approx 250,25,2.5$ for $R_{\max}=1000,100,10$, respectively. Figure \ref{fig:expected_cost_var_constraint} displays the optimisation objective $\mathbb{E}(u_{\max})$, normalised by $T$, as a function of the tail-floor fraction $c$ and subsample size $m$, treating $m$ as continuous for visualisation.

\begin{figure}[H]
    \centering
    \includegraphics[width=1.0\textwidth]{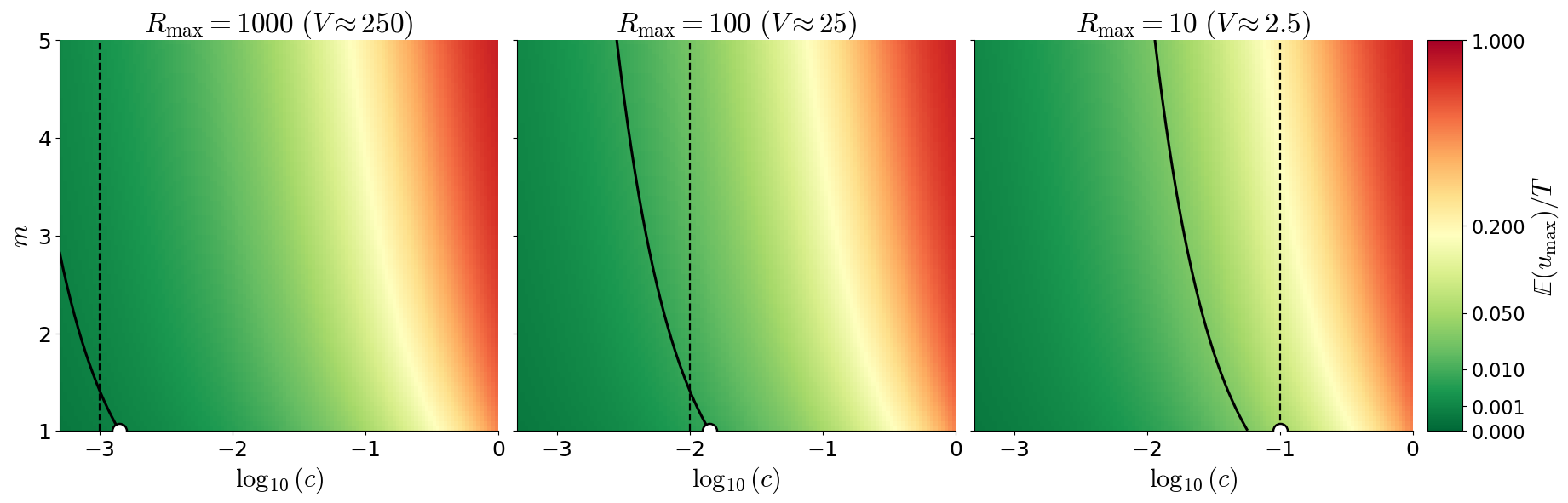}\caption{Constrained objective function $\mathbb{E}(u_{\max})$ defined in \eqref{eq:miminises_compute_subject_variance} (normalised by $T$, with $T=100{,}000$) and described in Section \ref{subsec:balancing_variance_and_computational_cost}, shown as a function of $c$ and $m$ (where $m$ is treated as continuous for visualisation). The black curve marks the variance constraint, and the dashed vertical line marks the safeguard lower bound constraint. Results are shown for three choices of $R_{\max}=1/c_{\min}$ (see panel titles) and $t^\star=1{,}000$. The white circles indicate the optima for each setting. The colour scale represents the expected computational cost ratio, where red indicates no compute savings relative to the full dataset of size $T$, and green indicates increasing compute savings.} \label{fig:expected_cost_var_constraint}
\end{figure}

    Figure \ref{fig:expected_cost_var_constraint} extends Figure \ref{fig:E_umax_c_m} by additionally showing the variance tolerance constraint (black curve), induced by the variance tolerances $V \approx 250, 25, 2.5$, obtained via variance inflation factors $R_{\max}=1000,100,10$ corresponding to $c_{\min}=0.001, 0.01, 0.1$, as well as the safeguard lower bound $c_{\min}$ (dashed vertical line). The feasibility region corresponds to the set of $(c,m)$ values lying above the variance constraint curve and to the right of the safeguard lower bound. The optimal tuning parameters (white circles) lie on the boundary of the feasibility region (either the variance constraint or the safeguard lower bound), and the expected computational cost decreases as the variance tolerance increases and the safeguard lower bound decreases. The variance constraint is binding at the optimum for $V \approx 250, 25$, whereas for $V \approx 2.5$ the safeguard lower bound binds first, resulting in a variance of $1.39$ at the optimum. These findings are consistent with the theoretical results in Section \ref{subsec:balancing_variance_and_computational_cost}. At the resulting optimal values, the expected computational costs are approximately 5,005, 708, and 72 for $V \approx 2.5, 25, 250$, respectively, which are substantially lower than the cost of processing the full dataset ($T=100{,}000$). Note that in general the optimal $m$ need not lie at the boundary $m=1$; its location depends on the interaction between the expected computational cost and the variance constraint, the latter depending on the control variates and the parameter value $\boldsymbol{\phi}^\dagger$ at which the algorithm is tuned.

We run doubly stochastic variational optimisers targeting the three variance tolerances above (by tuning $c,m$, which determine the sampling scheme). In addition, we run a doubly stochastic variational optimiser with uniform sampling weights, corresponding to the implementation in \cite{gunawan2017fast}. Applying the tuning guidelines for uniform sampling in Section \ref{subsec:balancing_variance_and_computational_cost} yields $m^{\star}=1$ for this optimiser under a variance tolerance of $\approx 0.25$ (obtained using $R_{\max}=1$). The variance at the optimum is $0.14$, and the variance tolerance constraint cannot bind since $m$ cannot be smaller than $1$. Finally, as a reference, we run a full-data singly stochastic variational optimiser. Note that for all subsampling-based variational configurations considered, the estimator of the gradient of the log-likelihood is unbiased. Consequently, the estimator of the ELBO gradient is also unbiased, and the optimisation falls within the Robbins–Monro stochastic approximation framework \citep{robbins1951stochastic}. We run a full-data Markov chain Monte Carlo with an adaptive proposal \citep{haario2001adaptive}, initialised with a scaled covariance matrix from a Laplace approximation, and obtain $10{,}000$ post-burn-in samples after discarding $2{,}000$ iterations. This chain serves as ground truth for the posterior distribution.

All experiments are conducted on a Linux workstation equipped with an Intel Core Ultra 7 268V processor and 32 gigabytes of random-access memory (RAM).

\subsection{Results}\label{subsec:results_VB_app1}

We first provide some empirical results concerning the theoretical upper bound introduced in Section \ref{subsec:TPD_probs}. The approximate invariance of $R_{\max}$ across $\boldsymbol{\phi}$ is confirmed by the empirical variance inflation ratios (TPD variance with $c_{\min}$ relative to uniform sampling) evaluated over the $N_{\mathrm{tune}}=20$ pilot samples. The resulting ranges are $[9.89,10.11]$, $[98.87,101.27]$, $[988.71,1012.80]$ for $R_{\max}=10,100,1000$, respectively, which are very close to their targeted values. Moreover, the tightness of the approximate bound $1/c$ is supported by the empirical results. In particular, across the pilot samples, the variance inflation ratios (TPD variance with $c^{\star}(\boldsymbol{\phi})$ relative to uniform sampling) closely track the corresponding bounds $1/c^{\star}(\boldsymbol{\phi})$, with median ratios within approximately $0.2\%$ of the corresponding bound in all three cases and maximum absolute deviations of 0.11, 1.25, and 12.64 for $R_{\max}=10,100,1000$, respectively.

We now turn to the stochastic optimisation results. To account for variability in the stochastic optimisation procedures, we repeat the optimisation $n_{\mathrm{rep}}=6$ times with different random seeds, with all runs using the same initialisation described in Section \ref{subsec:spec_settings_App1}. The ELBOs (estimated using the full-data log-posterior at the variational optima), CPU times, and compute fraction (CF) are summarised in Table~\ref{tab:VB_app_results}; see the caption for details. 

\begin{table}[h]
\centering
\small
\caption{Results for the GARCH($1,1$) model with normal errors fitted to the Dow Jones data. We compare subsampling variational Bayes (VB) under different sampling schemes, tuned via the safeguard parameter $c_{\min}$ (equivalently $R_{\max}=1/c_{\min}$), which determines the implied variance tolerance $V$. Variational optima are obtained using $5{,}000$ optimisation iterations with Polyak averaging over the final $500$ iterations. For each method, \(n_{\mathrm{rep}}=6\) independent stochastic optimisations are performed from a common initialisation. All ELBO values are estimated via the reparameterisation trick using $100$ Monte Carlo draws from the fitted variational posterior (obtained via either full-data VB or subsampling VB), evaluating the full-data log-likelihood at each draw. Best ELBO denotes the largest ELBO across the $n_{\mathrm{rep}}$ runs, with Monte Carlo standard error of the ELBO estimator (based on the reparameterisation trick) shown in parentheses. Median CPU is the median wall-clock optimisation time (seconds) across runs. Compute fraction (CF) accounts for log-density evaluations (used for ELBO monitoring), gradient evaluations, and tuning overhead for the subsampling algorithms, and is normalised so that full-data VB equals 1. TPD denotes truncated power-law decaying weights. Uniform corresponds to constant (non-decaying) weights (TPD with $\gamma=0$), i.e.\ uniform sampling, and coincides with the original algorithm in \cite{gunawan2017fast}.}
\begin{tabular}{l c c c c r r}
\toprule
Method
& $R_{\max}$
& $V$
& \makecell[c]{Best \\ ELBO}
& \makecell[c]{Median \\ ELBO}
& \makecell[c]{Median \\ CPU}
& \makecell[c]{Median \\ CF} \\
\midrule

Full-data VB
& --
& --
& \makecell[c]{-114663.6 \\ (0.39)}
& -114664.9
& 73375
& 1.00000 \\
\midrule

Subsampling VB
& & & & & & \\
\midrule

\quad Uniform
& $1$
& $\approx 0.25$
& \makecell[c]{-114664.0 \\ (0.35)}
& -114667.1
& 23418
& 0.51993 \\
\addlinespace[4pt]

\quad TPD
& $10$
& $\approx 2.5$
& \makecell[c]{-114663.1 \\ (0.43)}
& -114665.1
& 2576
& 0.06520 \\
\addlinespace[4pt]

\quad TPD
& $100$
& $\approx 25$
& \makecell[c]{-114663.7 \\ (0.36)}
& -114664.2
& 623
& 0.02359 \\
\addlinespace[4pt]

\quad TPD
& $1000$
& $\approx 250$
& \makecell[c]{-114663.4 \\ (0.33)}
& -114664.3
& 315
& 0.01733 \\
\bottomrule
\end{tabular}\label{tab:VB_app_results}
\end{table}

The table shows that subsampling VB with truncated power-law decaying weights yields substantial computational speed-ups relative to both full-data VB and subsampling VB under uniform sampling. The TPD schemes also attain median ELBO values comparable to the full-data VB median, with the two larger variance settings yielding slightly higher median values. Relative to the full-data method, the TPD variants are approximately 28, 118, and 233 times faster for $R_{\max}=10,100,1000$, respectively, while using only 1.7--6.5\% of the full-data computation. This speed-up is attributed to the decay mechanism, which produces extremely small realised values of the largest subsample index $u_{\max}$ on average; see Figure \ref{fig:u_max} for an illustration.

Figure \ref{fig:App1_VB_delta_elbo} assesses the robustness of the variational optimisers. It displays boxplots of
\begin{align}\label{eq:delta_ELBO}
\Delta \mathrm{ELBO} & = \mathrm{ELBO} - \operatorname{median}\bigl(\mathrm{ELBO}_{\text{Full-data}}\bigr),
\end{align}
where the median is taken over $n_{\mathrm{rep}}=6$ full-data variational Bayes runs. 

\begin{figure}[h]
    \centering    \includegraphics[width=1\textwidth]{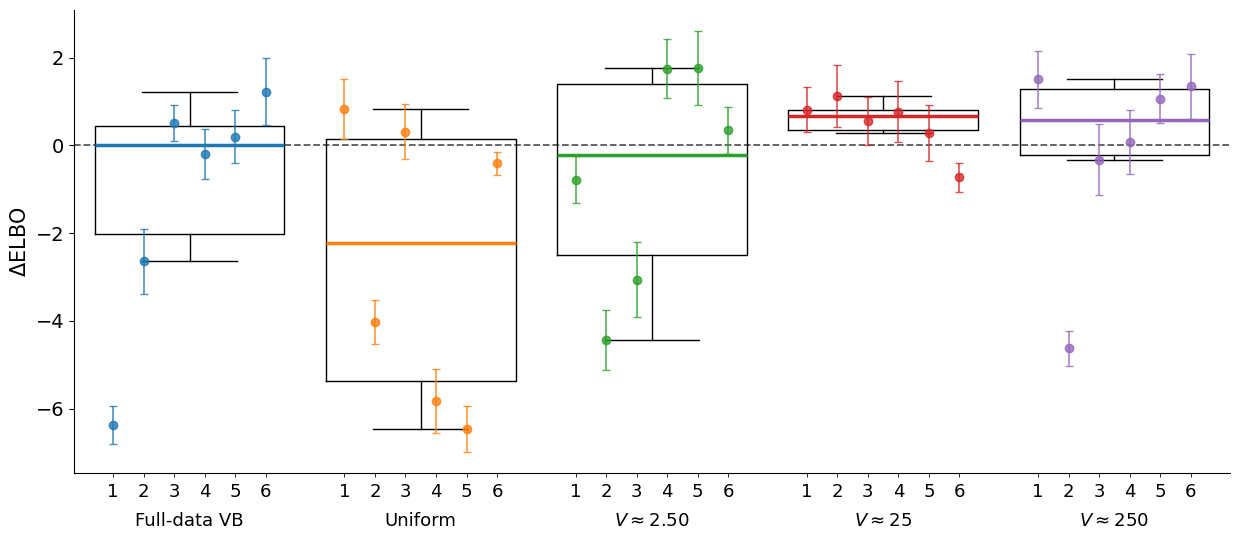}
    \caption{Results for the GARCH($1,1$) model with normal errors fitted to the Dow Jones data. The figure shows boxplots of the evidence lower bound (ELBO) differences $\Delta \mathrm{ELBO}$ defined in \eqref{eq:delta_ELBO}, where positive values indicate ELBO values higher than the median full-data VB ELBO, which is marked with a dashed horizontal line. Subsampling variational Bayes (VB) results are shown for different sampling schemes, with full-data VB as the reference. Points correspond to the $n_{\mathrm{rep}}=6$ optimisation runs, with vertical bars indicating approximate 95\% confidence intervals for the ELBO estimates.}
    \label{fig:App1_VB_delta_elbo}
\end{figure}

The individual replication estimates in Figure~\ref{fig:App1_VB_delta_elbo} are overlaid together with approximate $95\%$ confidence intervals for the ELBO estimates based on asymptotic normality. Positive values of $\Delta \mathrm{ELBO}$ correspond to ELBO values higher than the median of the full-data variational Bayes runs, whereas negative values correspond to lower ELBO values. We first consider the full-data singly stochastic VB case, where the only source of noise in the ELBO estimator arises from the Monte Carlo estimation using the reparameterisation trick. The figure shows that the six optimisation runs converge to several distinct ELBO levels, indicating that the stochasticity in the optimisation procedure can lead to convergence to different local solutions. This behaviour is consistent with the non-concave optimisation landscape induced by the combination of a rank-1 plus isotropic covariance structure in the Gaussian variational approximation and the ridge-like geometry of the GARCH($1,1$) posterior. For uniform sampling, additional stochasticity is introduced through the subsampling-based gradient estimator. However, the resulting optimisation outcomes are both more variable and systematically worse than those of full-data VB, with some runs converging to lower ELBO values. It is worth noting that lower variance of the stochastic gradient estimator does not necessarily imply better optimisation outcomes in this non-convex setting. As illustrated in Figure \ref{fig:App1_VB_delta_elbo}, moderate additional stochasticity can help the optimiser escape poorer local optima, whereas excessive stochasticity may instead increase the risk of convergence to substantially poorer solutions. Consistent with this interpretation, the TPD schemes show that allowing greater stochasticity can improve typical optimisation performance, although the improvement is not monotone in $R_{\max}$. Relative to $V \approx 2.5$, the two larger variance settings shift the bulk of the $\Delta$ELBO distribution upwards and yield positive median values, with the outcomes for $V \approx 25$ particularly concentrated in the six-run experiment. At the same time, occasional poor optimisation outcomes remain possible, especially for $V \approx 250$. This downside risk is examined more systematically in Table \ref{tab:vb_robustness}, which summarises the results over 50 independent runs ($n_{\mathrm{rep}}=50$).

\begin{table}[h]
\centering
\caption{Robustness summary over 50 independent runs for the two larger variance settings of the truncated power-law decaying (TPD) sampling scheme, complementing the results in Figure \ref{fig:App1_VB_delta_elbo}. The table reports the median of $\Delta$ELBO, defined in \eqref{eq:delta_ELBO}, together with empirical probabilities of $\Delta$ELBO falling below selected thresholds.}
\begin{tabular}{lrrrr}
\toprule
Method & $\Delta$ELBO & $\Pr(\Delta \mathrm{ELBO} < -5)$ & $\Pr(\Delta \mathrm{ELBO} < -10)$ & $\Pr(\Delta \mathrm{ELBO} < -15)$ \\
\midrule
$V \approx 25$  & 0.71 & 0.02 & 0.02 & 0.02 \\
$V \approx 250$ & 0.79 & 0.14 & 0.12 & 0.06 \\
\bottomrule
\end{tabular}\label{tab:vb_robustness}
\end{table}

The table reports the median $\Delta$ELBO together with estimated probabilities of lower-tail events. The median $\Delta$ELBO is positive for both settings, indicating improved typical optimisation outcomes relative to the full-data median baseline. For $V \approx 25$, lower-tail outcomes are rare, with an estimated probability of only $0.02$ at each of the three thresholds considered. In contrast, for $V \approx 250$, the corresponding probabilities increase to $0.14$, $0.12$, and $0.06$, respectively. This reflects a trade-off: smaller safeguard values $c_{\min}$ (equivalently, larger $R_{\max}$, implying higher variance tolerance) can yield greater computational gains while increasing the risk of occasional poor optimisation outcomes.

In the following posterior accuracy assessment, we select a representative run corresponding to the upper median ELBO across the $n_{\mathrm{rep}}=6$ repetitions. Figure \ref{fig:App1_VB_normal_kde} shows marginal kernel density estimates of the posterior distribution in the original parameter space $\boldsymbol{\theta}$, obtained from full-data Markov chain Monte Carlo (MCMC), which serves as ground truth, and from the variational approximations.

\begin{figure}[H]
    \centering    \includegraphics[width=0.9\textwidth]{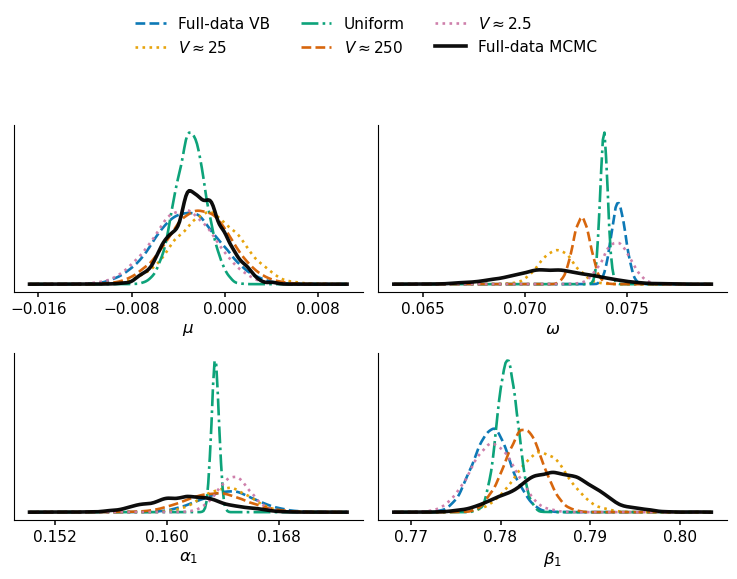}
    \caption{Posterior marginal distributions under the original parameterisation $\boldsymbol{\theta}$ for the GARCH($1,1$) model fitted to the Dow Jones data. The subsampling variational Bayes (VB) method with truncated decaying sampling probabilities is shown for $R_{\max}=10, 100, 1000$, corresponding to variance tolerances $V\approx 2.5, 25, 250$. The uniform sampling case corresponds to $R_{\max}=1$ with $V \approx 0.25$. The figure also shows the full-data VB approximate posterior, obtained using full-data gradients in the optimisation. The solid black curves represent the full-data Markov chain Monte Carlo (MCMC) posterior for reference.}
    \label{fig:App1_VB_normal_kde}
\end{figure}

We first note that the full-data VB and uniform subsampling approaches provide relatively poor fits compared to the full-data MCMC posterior, with the uniform subsampling approximation showing particularly marked discrepancies. This is broadly consistent with the optimisation results in Table \ref{tab:VB_app_results}, where uniform sampling attains the lowest median ELBO, while the TPD schemes with the two larger variance tolerances attain slightly higher median ELBO values than full-data VB. In contrast, the subsampling VB methods with truncated power-law decaying weights yield posterior approximations that are generally more closely aligned with the full-data MCMC posterior, although some discrepancies remain. These discrepancies are likely due to the restricted covariance structure imposed by the rank-1 plus isotropic variational family. Using a more flexible covariance structure could potentially improve the approximation and is left for future research. We note that the subsampling Markov chain Monte Carlo approach using truncated power-law decaying weights, considered in Sections \ref{sec:MCMC_application}--\ref{sec:comparison_divide-and-conquer_MCMC}, provides more accurate uncertainty quantification and should be preferred when accurate posterior inference is the primary objective; see for example Figures \ref{fig:App2_MCMC_GARCH_11_norm_kde}--\ref{fig:App2_MCMC_TGARCH_11_Student-t_kde},  \ref{fig:kde_SG_MCMC_subsampling_full_data}, and \ref{fig:kde_DC_BATS_subsampling_full_data}--\ref{fig:regression_GARCH_50}. The purpose of the current example is to demonstrate that our methodology enables doubly stochastic variational inference for GARCH models and can substantially
accelerate variational optimisation relative to singly stochastic variational inference. Our experiments show that this approach often leads to improved typical optimisation outcomes, at the cost of an increased risk of occasional poor optimisation outcomes. Table \ref{tab:vb_robustness} shows that such lower-tail events occur with relatively low probability, indicating that the downside risk, while present, is small Overall, we find that $R_{\max}=100$ provides a good balance between improved typical performance and downside risk, and may serve as a useful default in practice.

Finally, Section \ref{subsec:additional_results_App1} contains additional results for this application, including ELBO traces, $u_{\max}$ over iterations, and an extension with Student-$t$ errors.

\section{Conclusion and future research}\label{sec:conclusion_and_future_research}

We propose a stabilised weighted data subsampling methodology for accelerating Bayesian inference in models with recursively defined likelihoods. Our estimators of the log-likelihood and its gradient are unbiased with respect to the subsampling randomness, conditional on the observed data and for fixed parameter values. Central to our approach is a truncated power-law decaying sampling scheme that favours earlier observations, reducing the computational cost of recursive likelihood evaluations while controlling variance inflation through a tail probability floor. Our theoretical results characterise this variance-cost trade-off and motivate practical tuning guidelines that minimise expected computational cost subject to a variance tolerance, for a fixed head-tail boundary $t^\star$. Under fixed tuning parameters, the expected computational cost admits an $\mathcal{O}(T)$ upper bound with a smaller leading constant than the corresponding bound for uniform sampling.

We demonstrate the methodology within a subsampling-based variational inference framework applied to a GARCH model, enabling unbiased log-likelihood estimation in a setting previously considered infeasible in the literature. The proposed method substantially outperforms both full-data variational inference and variational inference using uniform subsampling. The supplementary material further considers subsampling Markov chain Monte Carlo for GARCH models and a threshold extension thereof, including out-of-sample model selection. We observe median speed-ups of around $48$--$66\times$ in CPU time, using only $1.8$--$2.4\%$ of the log-density evaluations required by full-data MCMC, while maintaining comparable inferential accuracy. Our approach also compares favourably with stochastic gradient MCMC (Section \ref{sec:comparison_SG_MCMC}) and divide-and-conquer MCMC (Section \ref{sec:comparison_divide-and-conquer_MCMC}) in the settings considered. To demonstrate scalability in the number of model parameters, the latter comparison additionally considers higher-dimensional regression-GARCH models, substantially extending the parameter dimension beyond that of the standard GARCH models used in our other applications.

Our analysis also highlights an important limitation. For fixed $m$ and $t^\star$, scaling the tail mass as $\varepsilon=\mathcal{O}(T^{-1})$ is sufficient to keep expected computational cost bounded as $T$ grows. However, achieving this scaling requires weakening the variance safeguard as $T$ increases, potentially leading to substantial variance inflation. In our applications, including datasets with $T=100{,}000$, we obtained substantial computational savings without encountering prohibitively large estimator variance, although this does not constitute a general guarantee.

The proposed likelihood estimators provide generic building blocks for subsampling-based Bayesian inference with recursive likelihoods. Our theoretical analysis is conditional on the observed data, consistent with our Bayesian focus, and does not establish the frequentist properties of the resulting parameter estimators. Investigating these properties would require probabilistic results that account jointly for randomness in the data and the subsampling scheme. Establishing consistency, convergence rates, asymptotic normality, and valid confidence intervals under truncated decaying sampling probabilities therefore remains an important direction for future research. Another promising direction is to extend our methodology to models with latent state processes, thereby approaching the level of generality considered in \cite{aicher2025stochastic}.

%\section*{Use of artificial intelligence tools}

%ChatGPT (OpenAI, GPT-5.5) was used during manuscript preparation to assist with language editing, wording suggestions, document formatting, and software implementation. All scientific content, methodology, analyses, theoretical results, and conclusions were developed and verified by the authors, who take full responsibility for the final manuscript.

\putbib
\end{bibunit}

% Uncomment old references
%\bibliographystyle{apalike}
%\addcontentsline{toc}{section}{\refname}\bibliography{ref}

\clearpage
\appendix

\setcounter{page}{1}
\renewcommand{\thepage}{A\arabic{page}}

\begin{bibunit}

% For appendix: prefix with A
% Section = A, B, C, ...
\renewcommand{\thesection}{\Alph{section}}

% Subsection = A.1, A.2, ...
\renewcommand{\thesubsection}{\thesection.\arabic{subsection}}

% For equations: prefix with S<section>.<equation>
\renewcommand{\theequation}{A\arabic{section}.\arabic{equation}}

\renewcommand{\theproposition}{A\arabic{proposition}}

\renewcommand{\theremark}{A\arabic{remark}}

\renewcommand{\thelemma}{A\arabic{lemma}}

\renewcommand{\thealgorithm}{A\arabic{algorithm}}

\setcounter{algorithm}{0}
\setcounter{section}{0}
\setcounter{equation}{0}
\setcounter{lemma}{0}
\setcounter{remark}{0}
\setcounter{proposition}{0}

\section{Recursive likelihood examples: the GARCH family}\label{sec:VolatilityModels}
\subsection{Preliminaries}
Recall our working example, the general volatility model from Section \ref{subsec:general_model},
\begin{align}\label{eq:measurement_plus_deterministic_state_restated}
    y_t & = \mu + z_t, \,\, z_t = \sigma_t \varepsilon_t, \,\, \varepsilon_t \overset{\mathrm{ind}}{\sim} \mathcal{D}(0, 1),  \,\, \mu \in \mathbb{R}, \,\, \sigma_t > 0, \nonumber \\
    \sigma^2_t & = g(\boldsymbol{\theta}_v;\mathcal{F}_{t-1}).
\end{align}
Note that the conditional variance $\sigma^2_t$ in \eqref{eq:measurement_plus_deterministic_state_restated} is modelled deterministically (given $\mathcal{F}_{t-1}$); only the measurement equation has random errors $\varepsilon_t$. However, not conditioning on $\mathcal{F}_{t-1}$, we can think of the variance process $\sigma^2_t$ as random, where the randomness comes from past random errors $\varepsilon_s$, $s = 1, \dots t-1$. We denote the expectation of this variance process as $\mathbb{E}(\sigma^2_t)$, which can be derived assuming it does not depend on $t$. Moreover, it follows from the law of total expectation and variance that the unconditional mean $\mathbb{E}(y_t)=\mu$ and the unconditional variance $\mathbb{V}(y_t)=\mathbb{E}(\sigma^2_t)$ for $y_t$ in \eqref{eq:measurement_plus_deterministic_state_restated}. 

Below we present the conditional variance models used as working examples and derive the corresponding gradient and Hessian with respect to the original variance parameter vector $\boldsymbol{\theta}_v$, as well as the reparameterised $\boldsymbol{\phi}_v$. We also derive the associated gradients and Hessians of the log-densities under Gaussian and Student-$t$ innovations. These quantities are key components of our methodology for constructing the control variates. 

\cite{xuan2024stochastic} derive the gradient of the log-likelihood under Gaussian and Student-$t$ innovations, but assume a standard GARCH($1,1$) conditional variance structure and exclude a conditional mean. In addition to the gradient, \cite{fiorentini1996derivatives} derive the Hessian of the log-likelihood in a linear regression model with Gaussian errors following a standard GARCH($p,q$) specification. In contrast, we consider more complex GARCH processes, such as TGARCH($p,q$). We also note that automatic differentiation (see \citealp{baydin2018automatic} for a survey), when applied naively to GARCH-type log-densities, does not automatically exploit the recursive structure of the conditional variance updates, leading to unnecessary recomputation and increased computational cost.

\subsection{Conditional variance models}\label{subsec:conditional_variance}
\subsubsection{ARCH and GARCH}
In the autoregressive conditional heteroscedasticity (ARCH) \citep{engle1982autoregressive}, the conditional variance in \eqref{eq:measurement_plus_deterministic_state_restated} depends on past squared innovations $z_t$, 
\begin{align}\label{eq:ARCH}
    \sigma^2_t(\boldsymbol{\theta}_v) & = \omega + \sum_{i=1}^p \alpha_i z^2_{t-i} \nonumber \\
       & = \omega + \sum_{i=1}^p \alpha_i (y_{t-i} - \mu)^2,\,\,  \omega > 0, \,\, \alpha_i \geq 0,
\end{align}
denoted as ARCH($p$). Note that the conditional variance in \eqref{eq:ARCH} is not recursive and requires only $p$ lags of the series to evaluate for any $t$. Since $\mathbb{E}\left((y_{t-i} - \mu)^2 \right)=\mathbb{V}\left(y_{t-i}\right)=\mathbb{E}(\sigma^2_{t-i})$, and assuming $$\mathbb{E}(\sigma^2_t(\boldsymbol{\theta}_v)) = \mathbb{E}(\sigma^2_{t-1}(\boldsymbol{\theta}_v)) = \dots =\mathbb{E}(\sigma^2_{1}(\boldsymbol{\theta}_v)),$$
the unconditional variance is
$$\mathbb{E}(\sigma^2_{t}(\boldsymbol{\theta}_v)) = \frac{\omega}{1 - \sum_{i = 1}^p \alpha_i},$$
provided that $\sum_{i=1}^p \alpha_i < 1$.

The ARCH($p$)  may require a large $p$ to capture long-term volatility dynamics. \cite{bollerslev1986generalized} proposes adding lagged conditional variances in the ARCH structure. Specifically, with $q$ lags of the conditional variance, the GARCH($p,q$) model is defined as
\begin{align}\label{eq:GARCH}
    \sigma^2_t(\boldsymbol{\theta}_v) & = \omega + \sum_{i=1}^p \alpha_i (y_{t-i} - \mu)^2 +  \sum_{j=1}^q \beta_j \sigma^2_{t-j}, \,\,  \omega > 0, \,\, \alpha_i \geq 0, \,\, \beta_j \geq 0,
\end{align}
 and the variance parameters are $\boldsymbol{\theta}_v = (\mu, \omega, \boldsymbol{\alpha}^\top, \boldsymbol{\beta}^\top)^\top$, with $\boldsymbol{\alpha} = (\alpha_1, \dots, \alpha_p)^\top$ and $\boldsymbol{\beta} = (\beta_1, \dots, \beta_q)^\top$. Similarly to the ARCH case, if the unconditional variance is constant, it follows that 
 \begin{align}\label{eq:GARCH_uncond_var}
     \mathbb{E}(\sigma^2_{t}(\boldsymbol{\theta}_v)) = \frac{\omega}{1 - \sum_{i = 1}^p \alpha_i - \sum_{j = 1}^q \beta_j},
 \end{align}
 provided that $\sum_{i=1}^p \alpha_i + \sum_{j=1}^q \beta_j < 1$, which defines the stationary region of the GARCH($p,q$) model. As argued in Section \ref{subsec:likelihood}, a recursive structure of the conditional variance, such as that in \eqref{eq:GARCH}, implies that evaluating the log-density of an observation $y_t$, which contributes a single term in the log-likelihood, requires access to $y_1, \dots, y_t$, leading to a growing computational cost as $t$ increases.

Proposition \ref{prop:grad_Hess_GARCH} in Appendix \ref{app:GARCH_grad_Hess} provides the recursive gradient and Hessian computation for the GARCH model.

\subsubsection{TGARCH}\label{subsec:TGARCH}
The GARCH($p,q$) model treats positive and negative shocks $z_t = y_t - \mu$ symmetrically: the square in \eqref{eq:GARCH} removes the sign of the shock. In real applications, negative shocks (for example, bad news) tend to increase volatility more than positive shocks (good news). 

One approach to capture asymmetric effects of shocks on volatility, and in particular emphasise the effects of a negative shock, is the threshold GARCH model \citep{glosten1993relation, zakoian1994threshold}. The TGARCH($p,q$) model is defined as
\begin{align}\label{eq:TGARCH}
    \sigma^2_t(\boldsymbol{\theta}_v) & = \omega + \sum_{i=1}^p \left\{ \alpha_i (y_{t-i} - \mu)^2 + \gamma_i (y_{t-i} - \mu)^2\mathbbm{1}(y_{t-i} < \mu)  \right\}  + \sum_{j=1}^q \beta_j \sigma^2_{t-j},
\end{align}
where $\boldsymbol{\theta}_v = (\mu, \omega, \boldsymbol{\alpha}^\top, \boldsymbol{\gamma}^\top, \boldsymbol{\beta}^\top)^\top$, with $\boldsymbol{\alpha} = (\alpha_1, \dots, \alpha_p)^\top$, $\boldsymbol{\gamma} = (\gamma_1, \dots, \gamma_p)^\top$,
and $\boldsymbol{\beta} = (\beta_1, \dots, \beta_q)^\top$, with $\omega > 0$, $\alpha_i, \gamma_i \geq 0$, and $\beta_j \geq 0$, and $\mathbbm{1}(A)$ is 1 if the event $A$ occurs, and 0 otherwise. In \eqref{eq:TGARCH}, $\gamma_i$ controls how much a negative shock (defined as the observation being below its overall mean) at period $t-i$ affects the present volatility.

Unlike the standard GARCH model, deriving the unconditional variance for the TGARCH process is less straightforward, since the indicator function in the volatility process introduces an additional nonlinear component. \cite{zakoian1994threshold} derives this result for a general setting. In our case, the derivation is simplified because the distribution $\mathcal{D}(0, 1)$ in \eqref{eq:measurement_plus_deterministic_state_restated} is symmetric. Then, as shown in Section \ref{supp:TGARCH_unconditional_var}, the unconditional variance is given by
\begin{align}
\mathbb{E}(\sigma^2_t(\boldsymbol{\theta}_v))
= \frac{\omega}{1 - \sum_{i = 1}^p \alpha_i - \frac{1}{2} \sum_{i = 1}^p \gamma_i - \sum_{j = 1}^q \beta_j},    
\end{align}
under the stationary condition $ \sum_{i = 1}^p \alpha_i + \frac{1}{2} \sum_{i = 1}^p \gamma_i +\sum_{j = 1}^q \beta_j <1$.

Proposition \ref{prop:grad_Hess_TGARCH} in Appendix \ref{app:TGARCH_grad_Hess} provides the recursive gradient and Hessian computation for the threshold GARCH model.

\subsubsection{Reparameterised gradients and Hessians of the variance process}\label{sec:grad_hess_sigma2_reparam}

Propositions \ref{prop:grad_Hess_GARCH} and \ref{prop:grad_Hess_TGARCH} give recursive formulas for computing the gradients and Hessians of the GARCH and TGARCH conditional variance processes with respect to the original parameter $\boldsymbol{\theta}_v$. Lemma \ref{lem:reparam_grad_Hess_variance} in Appendix \ref{app:reparam_grad_Hess} shows how to obtain the corresponding gradients and Hessians with respect to the reparameterised parameter $\boldsymbol{\phi}_v$. Section \ref{supp:reparam_ex_GARCH} provides an example of the reparameterised expressions under commonly used transformations for the GARCH process, including those employed in the application in Section \ref{sec:VB_application}.

\subsection{Conditional log-density of the measurements}\label{subsec:grad_Hess_log_dens}

The distribution of the error term $\varepsilon_t$ determines the conditional distribution of the measurements in \eqref{eq:measurement_plus_deterministic_state_restated}, denoted $p(y_t|\mathcal{F}_{t-1}, \boldsymbol{\theta})$. Recall that the original parameter vector is $$\boldsymbol{\theta} = (\boldsymbol{\theta}_v^\top, \theta_\varepsilon)^\top \in \boldsymbol{\Theta} \subseteq \mathbb{R}^{\dim_{\boldsymbol{\theta}}},$$ where $\boldsymbol{\theta}_v$ collects the original parameters entering the conditional variance equation, and $\theta_\varepsilon$ denotes a possible additional original measurement parameter (for example the degrees of freedom in Student-$t$; absent in the Gaussian case). 

Define the conditional log-density at time $t$ as
\begin{align}\label{eq:log_dens_general}
    \ell_t(\boldsymbol{\theta}) = \log p(y_t|\mathcal{F}_{t-1}, \boldsymbol{\theta}).
\end{align}
As shown below, $\ell_t$ in \eqref{eq:log_dens_general} depends on the conditional mean $\mu$ and the conditional variance $\sigma^2_t$ in Appendix \ref{subsec:conditional_variance}, as well as on $\theta_{\varepsilon}=\nu$ when a Student-$t$ density is employed, where $\nu$ denotes the degrees of freedom.

\subsubsection{Gaussian model}\label{eq:Gaussian_model}
For the standard normal distribution, $\varepsilon_t \sim \mathcal{N}(0, 1)$, $\boldsymbol{\theta} = \boldsymbol{\theta}_v$,
\begin{align*}
  p(y_t|\mathcal{F}_{t-1}, \boldsymbol{\theta}) = \frac{1}{\sqrt{2\pi\sigma^2_t}}\exp\left( -\frac{1}{2\sigma_t^2}(y_t - \mu)^2\right),
\end{align*}
or equivalently,
\begin{align}\label{eq:Gauss_log_density}
  \ell_t(\sigma^2_t, \mu) = \log p(y_t|\mathcal{F}_{t-1}, \boldsymbol{\theta}) = -\frac{1}{2}\log(2\pi) - \frac{1}{2}\log(\sigma^2_t) - \frac{1}{2\sigma^2_t}(y_t - \mu)^2.
\end{align}
Proposition \ref{supp:prop_partial_Gauss} in Appendix \ref{app:measurement_equation} gives the partial derivatives with respect to $\mu$ and $\sigma^2_t$.

\subsubsection{Student-$t$ model}
For the standard Student-$t$ distribution, $\varepsilon_t \sim t_\nu(0, 1)$, with degrees of freedom $\nu$, $\boldsymbol{\theta} = (\boldsymbol{\theta}_v,\nu)^\top$,
$$f_{\varepsilon_t}(x) = \frac{\Gamma\left(\frac{\nu + 1}{2}\right)}{\Gamma\left(\frac{\nu}{2}\right)\sqrt{\pi(\nu - 2)}}\left( 1 +  \frac{x^2}{\nu - 2} \right)^{-\frac{\nu + 1}{2}},$$
where $\Gamma()$ is the gamma function. Using a change of variables to standardise,
\begin{align*}
  p(y_t|\mathcal{F}_{t-1}, \boldsymbol{\theta}) & = \frac{1}{\sigma_t}f_{\varepsilon_t}\left(\frac{y_t - \mu}{\sigma_t} \right) = \frac{\Gamma\left(\frac{\nu + 1}{2}\right)}{\sigma_t\Gamma\left(\frac{\nu}{2}\right)\sqrt{\pi(\nu - 2)}}\left( 1 +  \frac{(y_t - \mu)^2}{\sigma_t^2(\nu - 2)} \right)^{-\frac{\nu + 1}{2}},
\end{align*}
or equivalently,
\begin{align}\label{eq:Studentt_log_density}
  \ell(\sigma_t^2, \mu, \nu) = \log p(y_t|\mathcal{F}_{t-1}, \boldsymbol{\theta}) & = 
  \log \Gamma\left(\frac{\nu + 1}{2}\right) - \log\Gamma\left(\frac{\nu}{2}\right) - \frac{1}{2}\log(\pi(\nu - 2)) \nonumber \\
  & \phantom{-} -\frac{1}{2}\log(\sigma^2_t) 
  -\frac{(\nu + 1)}{2}\log \left( 1 +  \frac{(y_t - \mu)^2}{\sigma_t^2(\nu - 2)} \right).
\end{align}
Proposition \ref{supp:prop_partial_Studentt} in Appendix \ref{app:measurement_equation} gives the partial derivatives with respect to $\mu$,  $\sigma^2_t$, and $\nu$.

\subsubsection{Gradient and Hessian of the log-density}
The gradient and Hessian of $\ell_t(\boldsymbol{\theta})$ depend on the choice of GARCH process and the assumed measurement error. The corresponding expressions are given in Proposition \ref{prop:grad_Hess_log_dens} in Appendix \ref{app:measurement_equation}. They are functions of the gradient and Hessian of the variance process with respect to the original parameters $\boldsymbol{\theta}_v$, which are computed recursively as in Propositions
\ref{prop:grad_Hess_GARCH} and \ref{prop:grad_Hess_TGARCH} in Appendix \ref{app:GARCH_grad_Hess} and Appendix \ref{app:TGARCH_grad_Hess}.

\subsubsection{Reparameterised gradients and Hessians of the log-density}
Lemma \ref{lem:reparam_grad_Hess_log_dens} in Appendix \ref{app:reparam_grad_Hess} shows how to obtain gradients and Hessians of the log-density with respect to the reparameterised parameter $\boldsymbol{\phi}$. These quantities depend, among other terms, on the gradient and Hessian of the reparameterised variance process, obtained by the reparameterisation mappings in Lemma \ref{lem:reparam_grad_Hess_variance}.

\subsection{Computational considerations}\label{subsec:computational_considerations}
Naively evaluating each log-density $\ell_t(\boldsymbol{\theta})$, or its gradient $\nabla_{\boldsymbol{\theta}}\ell_t(\boldsymbol{\theta})$ and Hessian $\nabla^2_{\boldsymbol{\theta}}\ell_t(\boldsymbol{\theta})$ separately for all $t$ incurs unnecessary repeated recursive computations. Algorithm \ref{alg:online_eval} outlines an online evaluation of the log-densities and their derivatives, where at each time point $t$ only a small number of recent values are required, thereby avoiding recomputation over the full history. The per-time-step computation in Algorithm \ref{alg:online_eval} is $\mathcal{O}(\max(p, q))$, accounting for the recursive evaluation of the conditional variance and its derivatives (cost $\mathcal{O}({\max(p, q)})$, together with the evaluation of the current log-likelihood contribution $\ell_t$ and its derivatives. The total cost over $T$ time steps is $\mathcal{O}(T\max(p, q ))$, which simplifies to $\mathcal{O}(T)$ under the assumption that $T\gg p, q$.

Finally, the control variates in our subsampling methodology include derivatives reparameterised with respect to the original parameter $\boldsymbol{\phi}$, evaluated at a central value $\boldsymbol{\phi}^\star$. These are constructed by applying the reparameterisation mappings in Lemma \ref{lem:reparam_grad_Hess_log_dens} to the stored gradient and Hessian values produced by Algorithm \ref{alg:online_eval}. 

\begin{algorithm}[h]
\caption{Evaluating $\ell(\boldsymbol{\theta})=\sum_{t=1}^T \ell_t(\boldsymbol{\theta})$ and its derivatives for a GARCH-type variance model of order $(p,q)$. The notation $\mathbf{x}_{a:b}$, with $a>b$, denotes the vector $(x_{a}, x_{a-1}, \dots, x_{b})^\top$.}
\label{alg:online_eval}
\begin{algorithmic}[1]
\Require Observations $\{y_t\}_{t=1}^T$, parameters $\boldsymbol{\theta}=(\boldsymbol{\theta}_v^\top,\theta_\varepsilon)^\top$, initial values $\mathbf{y}_{0:-p}$ and $\boldsymbol{\sigma}^2_{0:-q}$
\State Initialise $\ell \gets 0$, \quad $\nabla_{\boldsymbol{\theta}}\ell \gets \mathbf{0}_{d}$, \quad $\nabla^2_{\boldsymbol{\theta}}\ell \gets \mathbf{0}_{d\times d}$
\Comment{$d=\dim_{\boldsymbol{\theta}}$}
\State Initialise storage $\{\ell_t\}_{t=1}^T$, $\{\nabla_{\boldsymbol{\theta}}\ell_t\}_{t=1}^T$ and $\{\nabla^2_{\boldsymbol{\theta}}\ell_t\}_{t=1}^T$
\Comment{For control variates}
\For{$j=1$ to $q$}
    \State Initialise $\nabla_{\boldsymbol{\theta}_v}\sigma^2_{1-j} \gets \mathbf{0}_{d_v}$ and
    $\nabla^2_{\boldsymbol{\theta}_v}\sigma^2_{1-j} \gets \mathbf{0}_{d_v\times d_v}$
\EndFor
\Comment{$d_v=\dim_{\boldsymbol{\theta}_v}$}
\For{$t = 1$ to $T$}
    \State $\sigma_t^2 \gets f_{\sigma_t^2}\!\left(\boldsymbol{\theta}_v, \mathbf{y}_{t-1:t-p}, \boldsymbol{\sigma}^{2}_{t-1:t-q}\right)$
    \Comment{\eqref{eq:GARCH} or \eqref{eq:TGARCH}}
    \State $\nabla_{\boldsymbol{\theta}_v}\sigma_t^2 \gets 
    f_{\nabla\sigma_t^2}\!\left(\boldsymbol{\theta}_v, \mathbf{y}_{t-1:t-p}, \boldsymbol{\sigma}^{2}_{t-1:t-q}, \nabla_{\boldsymbol{\theta}_v}\boldsymbol{\sigma}^{2}_{t-1:t-q}\right)$
    \Comment{Prop. \ref{prop:grad_Hess_GARCH}--\ref{prop:grad_Hess_TGARCH}}
    \State $\nabla^2_{\boldsymbol{\theta}_v}\sigma_t^2 \gets 
    f_{\nabla^2\sigma_t^2}\!\left(\boldsymbol{\theta}_v, \mathbf{y}_{t-1:t-p}, \boldsymbol{\sigma}^{2}_{t-1:t-q}, \nabla_{\boldsymbol{\theta}_v}\boldsymbol{\sigma}^{2}_{t-1:t-q}, \nabla^2_{\boldsymbol{\theta}_v}\boldsymbol{\sigma}^{2}_{t-1:t-q}\right)$
    \Comment{Prop. \ref{prop:grad_Hess_GARCH}--\ref{prop:grad_Hess_TGARCH}}
    \State $\ell_t \gets f_{\ell_t}\!\left(y_t,\mu,\sigma_t^2,\theta_\varepsilon\right)$
    \Comment{\eqref{eq:Gauss_log_density} or \eqref{eq:Studentt_log_density}}
    \State $\ell \gets \ell + \ell_t$
    \State $\nabla_{\boldsymbol{\theta}}\ell_t \gets 
    f_{\nabla\ell_t}\!\left(\sigma_t^2, \nabla_{\boldsymbol{\theta}_v}\sigma_t^2\right)$
    \Comment{Prop. \ref{prop:grad_Hess_log_dens}}
    \State $\nabla^2_{\boldsymbol{\theta}}\ell_t \gets 
    f_{\nabla^2\ell_t}\!\left(\sigma_t^2, \nabla_{\boldsymbol{\theta}_v}\sigma_t^2, \nabla^2_{\boldsymbol{\theta}_v}\sigma_t^2\right)$
    \Comment{Prop. \ref{prop:grad_Hess_log_dens}}
    \State Store $\ell_t$, $\nabla_{\boldsymbol{\theta}}\ell_t$ and $\nabla^2_{\boldsymbol{\theta}}\ell_t$
    \State $\nabla_{\boldsymbol{\theta}}\ell \gets \nabla_{\boldsymbol{\theta}}\ell + \nabla_{\boldsymbol{\theta}}\ell_t$
    \State $\nabla^2_{\boldsymbol{\theta}}\ell \gets \nabla^2_{\boldsymbol{\theta}}\ell + \nabla^2_{\boldsymbol{\theta}}\ell_t$
\EndFor
\State \Return $\ell$, $\nabla_{\boldsymbol{\theta}}\ell$, $\nabla^2_{\boldsymbol{\theta}}\ell$, $\{\ell_t\}_{t=1}^T$, $\{\nabla_{\boldsymbol{\theta}}\ell_t\}_{t=1}^T$, $\{\nabla^2_{\boldsymbol{\theta}}\ell_t\}_{t=1}^T$
\end{algorithmic}
\end{algorithm}

\section{Derivatives for implementation}\label{app:derivatives_GARCH}
\subsection{GARCH}\label{app:GARCH_grad_Hess}
Rewrite \eqref{eq:GARCH} in vector form,
\begin{align}\label{eq:GARCH_vectorised_mainpaper}
    \sigma^2_t(\boldsymbol{\theta}_v) & = \omega + {\mathbf{z}^{2^\top}_{t-1:t-p}}\boldsymbol{\alpha} +  {\boldsymbol{\sigma}_{t-1:t-q}^{2^\top}} \boldsymbol{\beta},
\end{align}
where
$$\mathbf{z}_{t-1:t-p} = \mathbf{y}_{t-1:t-p} -\mathbf{1}_p \mu,$$
where $\mathbf{1}_p$ denotes a unit column vector of length $p$ and for a vector $\mathbf{x}_{a:b}$, with $a>b$, $(x_{a}, x_{a-1}, \dots, x_{b})^\top$.

\begin{proposition}[Recursive evaluation of gradients and Hessians for the GARCH($p,q$) process]\label{prop:grad_Hess_GARCH}
Consider the conditional variance $\sigma_t^2(\boldsymbol{\theta}_v)$ in \eqref{eq:GARCH_vectorised_mainpaper}.
\begin{enumerate}
\item[(i)] The gradient with respect to
$\boldsymbol{\theta}_v$ is evaluated recursively as
\begin{align*}
\nabla_{\boldsymbol{\theta}_v} \sigma_t^2(\boldsymbol{\theta}_v) = \begin{pmatrix} -2\boldsymbol{\alpha}^\top\mathbf{z}_{t-1:t-p} \\ 1 \\ \mathbf{z}^{2}_{t-1:t-p} \\ {\boldsymbol{\sigma}_{t-1:t-q}^{2}} \end{pmatrix} +  \sum_{j=1}^q  \beta_j \nabla_{\boldsymbol{\theta}_v} \sigma_{t-j }^2(\boldsymbol{\theta}_v),
\end{align*}
where, for $t=1$, $$\nabla_{\boldsymbol{\theta}_v}\sigma_{1-j}^2(\boldsymbol{\theta}_v)=\mathbf{0}_d, \,\,  \,\,j=1,\dots,q.$$

\item[(ii)] The Hessian with respect to
$\boldsymbol{\theta}_v$ is evaluated recursively as
\begin{align*}
     \nabla^2_{\boldsymbol{\theta}_v} \sigma_t^2(\boldsymbol{\theta}_v) = \left(\mathbf{A} + \sum_{j=1}^q \nabla_{\boldsymbol{\theta}_v}\sigma^2_{t-j}(\boldsymbol{\theta}_v) \otimes \mathbf{e}^\top_k \right) + \mathbf{B} + \sum_{j = 1}^q \beta_j \nabla^2_{\boldsymbol{\theta}_v}\sigma^2_{t-j}(\boldsymbol{\theta}_v),
\end{align*}
with 
\begin{align*}
    \mathbf{A} =  \begin{pmatrix} 2\boldsymbol{\alpha}^\top\mathbf{1}_p & 0 & -2\mathbf{z}^{\top}_{t-1:t-p} & \mathbf{0}_q^\top \\
    0 & 0 & \mathbf{0}_p^\top & \mathbf{0}_q^\top \\
    -2 \mathbf{z}_{t-1:t-p} & \mathbf{0}_p & \mathbf{0}_{p\times p} & \mathbf{0}_{p \times q} \\
    \mathbf{0}_{q} & \mathbf{0}_{q} & \mathbf{0}_{q\times p} & \mathbf{0}_{q \times q}\end{pmatrix} \text{ and } \mathbf{B} =  \begin{pmatrix} \mathbf{0}^\top_d\\\mathbf{0}^\top_d \\ \mathbf{0}_{p\times d} \\ \nabla^\top_{\boldsymbol{\theta}_v}\sigma^2_{t-1}(\boldsymbol{\theta}_v) \\
    \nabla^\top_{\boldsymbol{\theta}_v}\sigma^2_{t-2}(\boldsymbol{\theta}_v) \\
    \vdots\\
    \nabla^\top_{\boldsymbol{\theta}_v}\sigma^2_{t-q}(\boldsymbol{\theta}_v) 
    \end{pmatrix},
\end{align*}
where, for $t=1$, $$\nabla^2_{\boldsymbol{\theta}_v}\sigma_{1-j}^2(\boldsymbol{\theta}_v)=\mathbf{0}_{d\times d}, \,\,\,\, j=1,\dots,q,$$
and
$\mathbf{e}_k$ denotes the $k$th canonical basis vector, with $k=2+p+j$.
\end{enumerate}
\end{proposition}

\subsection{TGARCH}\label{app:TGARCH_grad_Hess}

Rewrite \eqref{eq:TGARCH} in vector form,
\begin{align}\label{eq:TGARCH_vectorised_mainpaper}
    \sigma^2_t(\boldsymbol{\theta}_v) & = \omega + {\mathbf{z}^{2^\top}_{t-1:t-p}}\boldsymbol{\alpha} + \boldsymbol{\gamma}^{\top} \mathrm{diag}({\mathbf{z}^{2}_{t-1:t-p}})\mathbf{G}(\mu) +  {\boldsymbol{\sigma}_{t-1:t-q}^{2^\top}} \boldsymbol{\beta},
\end{align}
where
\begin{align*}
    \mathbf{G}(\mu) = (H(\mu - y_{t-1}), \dots, H(\mu - y_{t-p}))^{\top},
\end{align*}
with $H(x)$ being the Heaviside function, defined by 
\begin{equation*}
H(x) =
\begin{cases}
1, & x > 0,\\
0, & \text{otherwise}.
\end{cases}
\end{equation*}

\begin{proposition}[Recursive evaluation of gradients and Hessians for the TGARCH($p,q$) process]\label{prop:grad_Hess_TGARCH}
Consider the conditional variance $\sigma_t^2(\boldsymbol{\theta}_v)$ in \eqref{eq:TGARCH_vectorised_mainpaper}.
\begin{enumerate}
\item[(i)] The gradient with respect to
$\boldsymbol{\theta}_v$ is evaluated recursively as
\begin{align*}
\nabla_{\boldsymbol{\theta}_v} \sigma_t^2(\boldsymbol{\theta}_v) &= \begin{pmatrix} -2\boldsymbol{\alpha}^\top\mathbf{z}_{t-1:t-p} - \boldsymbol{\gamma}^{\top}\{2\mathrm{diag}(\mathbf{z}_{t-1:t-p})\mathbf{G}(\mu) - \mathrm{diag}(\mathbf{z}^{2}_{t-1:t-p})\mathbf{G}'(\mu)\}\\ 1 \\ \mathbf{z}^{2}_{t-1:t-p} \\ 
\mathrm{diag}(\mathbf{z}^{2}_{t-1:t-p})\mathbf{G}(\mu)\\
{\boldsymbol{\sigma}_{t-1:t-q}^{2}} \end{pmatrix} \nonumber\\
&\quad+  \sum_{j=1}^q  \beta_j \nabla_{\boldsymbol{\theta}_v} \sigma_{t-j }^2(\boldsymbol{\theta}_v),
\end{align*}
where, for $t=1$, $$\nabla_{\boldsymbol{\theta}_v}\sigma_{1-j}^2(\boldsymbol{\theta}_v)=\mathbf{0}_d, \,\,  \,\,j=1,\dots,q,$$
and

    $$\mathbf{G}'(\mu) = (\delta(\mu - y_{t-1}), \dots, \delta(\mu - y_{t-p}))^{\top},$$
where $\delta(\cdot)$ denotes the Dirac delta distribution.

\item[(ii)] The Hessian with respect to
$\boldsymbol{\theta}_v$ is evaluated recursively as
\begin{align*}
     \nabla^2_{\boldsymbol{\theta}_v} \sigma_t^2(\boldsymbol{\theta}_v) = \left(\mathbf{A} + \sum_{j=1}^q \nabla_{\boldsymbol{\theta}_v}\sigma^2_{t-j}(\boldsymbol{\theta}_v) \otimes \mathbf{e}^\top_k \right) + \mathbf{B} + \sum_{j = 1}^q \beta_j \nabla^2_{\boldsymbol{\theta}_v}\sigma^2_{t-j}(\boldsymbol{\theta}_v),
\end{align*}
with 
\begin{align*}
    \mathbf{A} =  \begin{pmatrix} \mathbf{A}_{\mu\mu} & 0 & -2\mathbf{z}^{\top}_{t-1:t-p} & \mathbf{A}_{\mu\boldsymbol{\gamma}}  & \mathbf{0}_q^\top \\
    0 & 0 & \mathbf{0}_p^\top & \mathbf{0}_p^\top & \mathbf{0}_q^\top \\
    -2 \mathbf{z}_{t-1:t-p} & \mathbf{0}_p & \mathbf{0}_{p\times p} & \mathbf{0}_{p\times p} &  \mathbf{0}_{p \times q} \\ \mathbf{A}_{\boldsymbol{\gamma}\mu}
    & \mathbf{0}_p & \mathbf{0}_{p\times p} & \mathbf{0}_{p\times p} &  \mathbf{0}_{p \times q}\\ \mathbf{0}_{q} &
    \mathbf{0}_{q} & \mathbf{0}_{q\times p} & \mathbf{0}_{q\times p} & \mathbf{0}_{q \times q}\end{pmatrix} \text{ and }    \mathbf{B} =  \begin{pmatrix} \mathbf{0}^\top_d\\\mathbf{0}^\top_d \\ \mathbf{0}_{p\times d} \\ 
    \mathbf{0}_{p\times d} \\\nabla^\top_{\boldsymbol{\theta}_v}\sigma^2_{t-1}(\boldsymbol{\theta}_v) \\
    \nabla^\top_{\boldsymbol{\theta}_v}\sigma^2_{t-2}(\boldsymbol{\theta}_v) \\
    \vdots\\
    \nabla^\top_{\boldsymbol{\theta}_v}\sigma^2_{t-q}(\boldsymbol{\theta}_v) 
    \end{pmatrix},
\end{align*}
and
\begin{align*}
    \mathbf{A}_{\mu\mu} = 2\boldsymbol{\alpha}^\top\mathbf{1}_p + \boldsymbol{\gamma}^\top \{2\mathbf{G}(\mu) - 4\mathrm{diag}(\mathbf{z}_{t-1:t-p})\mathbf{G}'(\mu) + \mathrm{diag}(\mathbf{z}^{2}_{t-1:t-p})\mathbf{G}''(\mu)\},
\end{align*}
\begin{align*}
     \mathbf{A}_{\mu\boldsymbol{\gamma}} = (\mathbf{A}_{\boldsymbol{\gamma}\mu})^\top= -\{2\mathrm{diag}(\mathbf{z}_{t-1:t-p})\mathbf{G}(\mu) + \mathrm{diag}(\mathbf{z}^{2}_{t-1:t-p})\mathbf{G}'(\mu)\}^\top,
\end{align*}
\begin{align*}
\mathbf{G}''(\mu) &= (\delta'(\mu - y_{t-1}), \dots, \delta'(\mu - y_{t-p}))^{\top},
\end{align*}
where, for $t=1$, $$\nabla^2_{\boldsymbol{\theta}_v}\sigma_{1-j}^2(\boldsymbol{\theta}_v)=\mathbf{0}_{d\times d}, \,\,\,\, j=1,\dots,q,$$
and
$\mathbf{e}_k$ denotes the $k$th canonical basis vector, with $k=2+2p+j$.
\end{enumerate}
\end{proposition}

\begin{remark}
For numerical implementation, we replace $G(\mu)$ by a smooth logistic approximation; see Remark \ref{rem:computing_dirac} in Section \ref{supp:TGARCH_gradient}.
\end{remark}

\subsection{Log-density of the measurements}\label{app:measurement_equation}

The following propositions give the partial derivatives of the Gaussian and Student-$t$ models.

\begin{proposition}[partial derivatives of the Gaussian model] \label{supp:prop_partial_Gauss} The partial derivatives of the Gaussian model in \eqref{eq:Gauss_log_density} are
\begin{align*}%\label{eq:partial_derivatives_Gauss}
    \frac{\partial \ell_t(\sigma^2_t, \mu)}{\partial \mu} & = \frac{y_t - \mu}{\sigma^2_t} & \frac{\partial^2 \ell_t(\sigma^2_t, \mu)}{\partial \mu^2}  & = -\frac{1}{\sigma_t^2} \nonumber \\
\frac{\partial \ell_t(\sigma^2_t, \mu)}{\partial \sigma_t^2} & = -\frac{1}{2\sigma^2_t} + \frac{(y_t - \mu)^2}{2(\sigma^2_t)^2}    & \frac{\partial^2 \ell_t(\sigma^2_t, \mu)}{\partial (\sigma_t^2)^2} & = \frac{1}{2(\sigma^2_t)^2} - \frac{(y_t - \mu)^2}{(\sigma^2_t)^3} \nonumber \\
\frac{\partial^2 \ell_t(\sigma^2_t, \mu)}{\partial \sigma_t^2\partial \mu} & = -\frac{y_t - \mu}{(\sigma^2_t)^2} .   &  & 
\end{align*}
\end{proposition}

\begin{proposition}[partial derivatives of the Student-$t$ model] \label{supp:prop_partial_Studentt}
The partial derivatives for the Student-$t$ model in \eqref{eq:Studentt_log_density} are 
\begin{align*}%\label{eq:partial_derivatives_Student_t}
    \frac{\partial \ell_t(\sigma^2_t, \mu, \nu)}{\partial \mu} & = \frac{(\nu + 1)(y_t - \mu)}{\sigma_t^2(\nu - 2) + (y_t - \mu)^2} \nonumber\\
    \frac{\partial^2 \ell_t(\sigma^2_t, \mu, \nu)}{\partial \mu^2}  & =  (v+1)\Bigg[\frac{2(y_t - \mu)^2}{\{\sigma_t^2(\nu - 2) + (y_t - \mu)^2\}^2} - \frac{1}{\sigma_t^2(\nu - 2) + (y_t - \mu)^2}\Bigg] \nonumber \\
    \frac{\partial \ell_t(\sigma^2_t, \mu, \nu)}{\partial \sigma_t^2} & =  -\frac{1}{2\sigma_t^2} + \frac{(\nu + 1)}{2}\Bigg[\frac{(y_t- \mu)^2}{\sigma_t^2\{\sigma_t^2(\nu - 2) + (y_t - \mu)^2\}}\Bigg] \nonumber\\
    \frac{\partial^2 \ell_t(\sigma^2_t, \mu, \nu)}{\partial (\sigma_t^2)^2} & = \frac{1}{2(\sigma_t^2)^2} - \frac{(\nu+1)}{2}\Bigg[\frac{(y_t - \mu)^2}{(\sigma_t^2)^2\{\sigma_t^2(\nu - 2) + (y_t - \mu)^2\}} + \frac{(\nu-2)(y_t-\mu)^2}{\sigma_t^2\{\sigma_t^2(\nu - 2) + (y_t - \mu)^2\}^2}\Bigg] \nonumber \\
    \frac{\partial \ell_t(\sigma^2_t, \mu, \nu)}{\partial \nu} & = \psi(\nu) - \psi\bigg(\frac{\nu}{2}\bigg) - \ln(2) - \frac{1}{2(\nu-2)} - \frac{1}{2}\log\Bigg(1+ \frac{(y_t- \mu)^2}{\sigma_t^2(\nu-2)}\Bigg) \nonumber\\
    &\qquad + \frac{(\nu + 1)}{2}\Bigg[\frac{(y_t - \mu)^2(\sigma_t^2)}{\sigma_t^2(\nu-2)\{\sigma_t^2(\nu - 2) + (y_t - \mu)^2\}}\Bigg]\nonumber     \end{align*}

    \begin{align*}
    \frac{\partial^2 \ell_t(\sigma^2_t, \mu, \nu)}{\partial \nu^2} & = \psi_1(\nu) - \frac{1}{2}\psi_1\bigg(\frac{\nu}{2}\bigg) + \frac{1}{2(\nu-2)^2} + \frac{(y_t-\mu)^2}{(\nu-2)\{\sigma_t^2(\nu - 2) + (y_t - \mu)^2\}}\nonumber \\
    &\qquad - \frac{(\nu+1)}{2}\Bigg(\frac{2\sigma_t^2(\nu-2)(y_t-\mu)^2 + (y_t-\mu)^4}{[(\nu-2)\{\sigma_t^2(\nu - 2) + (y_t - \mu)^2\}]^2}\Bigg) \nonumber\\
    \frac{\partial^2 \ell_t(\sigma^2_t, \mu, \nu)}{\partial \sigma_t^2\partial \mu} & = (\nu+1) \Bigg[ \frac{(y_t-\mu)^3)}{\sigma_t^2\{\sigma_t^2(\nu - 2) + (y_t - \mu)^2\}^2} - \frac{(y_t-\mu)}{\sigma_t^2\{\sigma_t^2(\nu - 2) + (y_t - \mu)^2\}} \Bigg]\nonumber \\
    \frac{\partial^2 \ell_t(\sigma^2_t, \mu, \nu)}{\partial \sigma_t^2\partial \nu} & = \frac{1}{2}\Bigg[\frac{(y_t-\mu)^2}{\sigma_t^2\{\sigma_t^2(\nu - 2) + (y_t - \mu)^2\}}\Bigg] - \frac{(\nu+1)}{2}\Bigg[\frac{(y_t-\mu)^2}{\{\sigma_t^2(\nu - 2) + (y_t - \mu)^2\}^2}\Bigg] \nonumber \\
    \frac{\partial^2 \ell_t(\sigma^2_t, \mu, \nu)}{\partial \mu \partial \nu} & = \frac{y_t-\mu}{\sigma_t^2(\nu - 2) + (y_t - \mu)^2} - \frac{\sigma_t^2(\nu+1)(y_t-\mu)}{\{\sigma_t^2(\nu - 2) + (y_t - \mu)^2\}^2}.
\end{align*}
where $\psi(z) = \frac{d}{dz}\ln \Gamma(z)$ and $\psi_1(z) = \frac{d}{dz}\psi(z)$.
\end{proposition}

The following proposition gives the gradient and Hessian of $\ell_t(\boldsymbol{\theta})$.

\begin{proposition}[Gradients and Hessians for the log-density]\label{prop:grad_Hess_log_dens}
Let $\sigma_t^2(\boldsymbol{\theta}_v)$ be a conditional variance process with gradient $\nabla_{\boldsymbol{\theta}_v}\sigma_t^2(\boldsymbol{\theta}_v)$ and Hessian $\nabla^2_{\boldsymbol{\theta}_v}\sigma_t^2(\boldsymbol{\theta}_v)$, and let $\ell_t(\boldsymbol{\theta})=\log p(y_t|\mathcal{F}_{t-1}, \boldsymbol{\theta})$, with $$\boldsymbol{\theta}=(\boldsymbol{\theta}^\top_v, \theta_\varepsilon)^\top.$$
\begin{enumerate}
\item[(i)] The gradient of $\ell_t$ with respect to
$\boldsymbol{\theta}$ is
\begin{align*}
\nabla_{\boldsymbol{\theta}}\ell_t(\boldsymbol{\theta}) & = \left(\nabla_{\boldsymbol{\theta}_v}\ell_t(\boldsymbol{\theta})^\top, \frac{\partial \ell_t(\boldsymbol{\theta})}{\partial\theta_\varepsilon}\right)^\top,      
\end{align*}
with
\begin{align*}
\nabla_{\boldsymbol{\theta}_v}\ell_t(\boldsymbol{\theta})  & = \frac{\partial \ell_t}{\partial\sigma^2_t}\nabla_{\boldsymbol{\theta}_v}\sigma^2_t(\boldsymbol{\theta}_v) + \frac{\partial \ell_t}{\partial\mu} \mathbf{e}_1, 
\end{align*}
where $\mathbf{e}_1$ is the first canonical basis vector.
\item[(ii)] The Hessian of $\ell_t$ with respect to
$\boldsymbol{\theta}$ is
\begin{align*}
\nabla^2_{\boldsymbol{\theta}}\ell_t(\boldsymbol{\theta}) = & 
\begin{bmatrix}
\nabla^2_{\boldsymbol{\theta}_v}\ell_t(\boldsymbol{\theta}) & \nabla_{\boldsymbol{\theta}_v}\frac{\partial \ell_t(\boldsymbol{\theta})}{\partial\theta_\varepsilon} \\
\left(\nabla_{\boldsymbol{\theta}_v}\frac{\partial \ell_t(\boldsymbol{\theta})}{\partial\theta_\varepsilon}\right)^\top & \frac{\partial^2 \ell_t(\boldsymbol{\theta})}{\partial\theta_\varepsilon^2}
\end{bmatrix},
\end{align*}
with
\begin{align*}
\nabla_{\boldsymbol{\theta}_v}^2\ell_t(\boldsymbol{\theta}) & = \frac{\partial^2 \ell_t}{\partial(\sigma^2_t)^2}\nabla_{\boldsymbol{\theta}_v}\sigma^2_t(\boldsymbol{\theta}_v)\nabla_{\boldsymbol{\theta}_v}\sigma^2_t(\boldsymbol{\theta}_v)^\top + \frac{\partial^2 \ell_t}{\partial\sigma^2_t\partial\mu}\left( \nabla_{\boldsymbol{\theta}_v}\sigma^2_t(\boldsymbol{\theta}_v)\mathbf{e}_1^\top + \mathbf{e}_1\nabla_{\boldsymbol{\theta}_v}\sigma^2_t(\boldsymbol{\theta}_v)^\top \right) \nonumber   \\
 & +  \frac{\partial \ell_t}{\partial\sigma^2_t} \nabla^2_{\boldsymbol{\theta}_v}\sigma^2_t(\boldsymbol{\theta}_v) + \frac{\partial^2 \ell_t}{\partial\mu^2} \mathbf{e}_1\mathbf{e}_1^\top,
\end{align*}
and
\begin{align*}
\nabla_{\boldsymbol{\theta}_v}\frac{\partial\ell_t(\boldsymbol{\theta})}{\partial \theta_\varepsilon} & = \frac{\partial^2 \ell_t}{\partial\sigma^2_t\partial\theta_\varepsilon}\nabla_{\boldsymbol{\theta}_v}\sigma^2_t(\boldsymbol{\theta}_v) + \frac{\partial^2 \ell_t }{\partial\mu\partial \theta_\varepsilon}  \mathbf{e}_1.
\end{align*}
\end{enumerate}
\end{proposition}
\begin{remark}
Proposition \ref{prop:grad_Hess_log_dens} separates the contribution of the conditional variance dynamics from that of the measurement error. For a given choice of $p(y_t |\mathcal{F}_{t-1}, \boldsymbol{\theta})$, the gradient and Hessian are obtained by substituting the corresponding first and second derivatives of $\ell_t$ with respect to $\sigma_t^2$, $\mu$, and $\theta_\varepsilon$; see Propositions \ref{supp:prop_partial_Gauss} and \ref{supp:prop_partial_Studentt}. In the Gaussian case, the parameter $\theta_\varepsilon$ is absent, and the corresponding components of the gradient and Hessian are therefore not present. In particular, all terms involving $\theta_\varepsilon$ vanish, and the expressions reduce to the subblocks associated with $\boldsymbol{\theta}_v$ and $\mu$.
\end{remark}

\subsubsection{Reparameterised gradients and Hessians}\label{app:reparam_grad_Hess}

Lemma \ref{lem:reparam_grad_Hess_variance} provides the gradient and Hessian of the conditional variance process with respect to the reparameterised variance parameters. These expressions enter the derivatives of the log-density
contributions, whose reparameterised forms are given in
Lemma \ref{lem:reparam_grad_Hess_log_dens}.

\begin{lemma}\label{lem:reparam_grad_Hess_variance}
Consider a conditional variance process $\sigma_t^2(\boldsymbol{\theta}_v)$ with gradient $\nabla_{\boldsymbol{\theta}_v}\sigma_t^2(\boldsymbol{\theta}_v)$ and Hessian $\nabla^2_{\boldsymbol{\theta}_v}\sigma_t^2(\boldsymbol{\theta}_v)$, and let $$\boldsymbol{\phi}_v=h(\boldsymbol{\theta}_v)$$ 
denote a one-to-one, elementwise transformation to an unrestricted parameterisation, with inverse
$$\boldsymbol{\theta}_v=h^{-1}(\boldsymbol{\phi}_v).$$
\begin{enumerate}
\item[(i)] The gradient with respect to
$\boldsymbol{\phi}_v$ is
\begin{align*}
\nabla_{\boldsymbol{\phi}_v}\sigma^2_t(\boldsymbol{\theta}_v)  & =J^{\boldsymbol{\phi}_v\top} \nabla_{\boldsymbol{\theta}_v}\sigma^2_t(\boldsymbol{\theta}_v),   
\end{align*}
with
$$J^{\boldsymbol{\phi}_v} = \mathrm{diag}\left(h^{-1\, '}_1(\phi_{v1}), \dots, h^{-1\, '}_{d_v}(\phi_{vd_v})\right).$$

\item[(ii)] The Hessian with respect to $\boldsymbol{\phi}_v$ is
\begin{align*}
\nabla^2_{\boldsymbol{\phi}_v}\sigma^2_t(\boldsymbol{\theta}_v) & = J^{\boldsymbol{\phi}_v\top}\nabla^2_{\boldsymbol{\theta}_v}\sigma^2_t(\boldsymbol{\theta}_v)J^{\boldsymbol{\phi}_v} + \sum_{i=1}^{d_v} \left(\nabla_{\boldsymbol{\theta}_v}\sigma^2_t(\boldsymbol{\theta}_v)\right)_i  h^{-1\, ''}_i(\phi_{vi})\mathbf{e}_i\mathbf{e}^\top_i,
\end{align*}
where $\mathbf{e}_i$ denotes the $i$th canonical basis vector.
\end{enumerate}
\end{lemma}

\begin{lemma}\label{lem:reparam_grad_Hess_log_dens}
Consider a log-density $\ell_t(\boldsymbol{\theta})$ with conditional variance process $\sigma_t^2(\boldsymbol{\theta}_v)$, having gradient $\nabla_{\boldsymbol{\theta}_v}\sigma_t^2(\boldsymbol{\theta}_v)$ and Hessian $\nabla^2_{\boldsymbol{\theta}_v}\sigma_t^2(\boldsymbol{\theta}_v)$. Let $$\boldsymbol{\phi}=h(\boldsymbol{\theta}), \quad \boldsymbol{\theta} = (\boldsymbol{\theta}_v^\top, \theta_\varepsilon)^\top, \,\,
\boldsymbol{\phi} = (\boldsymbol{\phi}_v^\top, \phi_\varepsilon)^\top,$$ 
denote a one-to-one, elementwise transformation to an unrestricted parameterisation, with inverse
$$\boldsymbol{\theta}=h^{-1}(\boldsymbol{\phi}),$$
and let $\nabla_{\boldsymbol{\phi}_v}\sigma_t^2(\boldsymbol{\theta}_v)$ and Hessian $\nabla^2_{\boldsymbol{\phi}_v}\sigma_t^2(\boldsymbol{\theta}_v)$ be the reparameterisations obtained via Lemma \ref{lem:reparam_grad_Hess_variance}.
\begin{enumerate}
\item[(i)] The gradient of $\ell_t$ with respect to
$\boldsymbol{\phi}$ is
\begin{align*}
\nabla_{\boldsymbol{\phi}}\ell_t(\boldsymbol{\theta}) & = \left(\nabla_{\boldsymbol{\phi}_v}\ell_t(\boldsymbol{\theta})^\top, \frac{\partial \ell_t(\boldsymbol{\theta})}{\partial\phi_\varepsilon}\right)^\top,      
\end{align*}
where 
\begin{align*}
\nabla_{\boldsymbol{\phi}_v}\ell_t(\boldsymbol{\theta})  & = \frac{\partial \ell_t}{\partial\sigma^2_t}\nabla_{\boldsymbol{\phi}_v}\sigma^2_t(\boldsymbol{\theta}_v) + \frac{\partial \ell_t}{\partial\mu} \mathbf{e}_1 h^{-1\, '}_1(\phi_1), 
\end{align*}
and $$\frac{\partial \ell_t(\boldsymbol{\theta})}{\partial \phi_\varepsilon}=\frac{\partial \ell_t}{\partial\theta_\varepsilon}\frac{\partial \theta_\varepsilon}{\partial \phi_\varepsilon}.$$
\item[(ii)] The Hessian of $\ell_t$ with respect to $\boldsymbol{\phi}$ is

\begin{align*}
\nabla^2_{\boldsymbol{\phi}}\ell_t(\boldsymbol{\theta}) = & 
\begin{bmatrix}
\nabla^2_{\boldsymbol{\phi}_v}\ell_t(\boldsymbol{\theta}) & \nabla_{\boldsymbol{\phi}_v}\frac{\partial \ell_t(\boldsymbol{\theta})}{\partial\phi_\varepsilon} \\
\left(\nabla_{\boldsymbol{\phi}_v}\frac{\partial \ell_t(\boldsymbol{\theta})}{\partial\phi_\varepsilon}\right)^\top & \frac{\partial^2 \ell_t(\boldsymbol{\theta})}{\partial\phi_\varepsilon^2}
\end{bmatrix},
\end{align*}
where
\begin{align*}
\nabla^2_{\boldsymbol{\phi}_v}\ell_t(\boldsymbol{\theta}) &= \frac{\partial^2\ell_t}{\partial(\sigma^2_t)^2}\nabla_{\boldsymbol{\phi}_v}\sigma^2_t(\boldsymbol{\theta}_v)\nabla_{\boldsymbol{\phi}_v}\sigma^2_t(\boldsymbol{\theta}_v)^\top + \frac{\partial \ell_t}{\partial \sigma^2_t}\nabla^2_{\boldsymbol{\phi}_v}\sigma^2_t(\boldsymbol{\theta}_v)\\
&\quad + \frac{\partial^2 \ell_t}{\partial \sigma^2_t\partial \mu}\left(\mathbf{e}_1 \nabla_{\boldsymbol{\phi}_v}\sigma^2_t(\boldsymbol{\theta}_v) + \nabla_{\boldsymbol{\phi}_v}\sigma^2_t(\boldsymbol{\theta}_v) \mathbf{e}_1 ^\top \right) h^{-1\, '}_1(\phi_1) \\
&\quad + \frac{\partial^2 \ell_t}{\partial \mu^2} \mathbf{e}_1 \mathbf{e}_1 ^\top\left(h^{-1\, '}_1(\phi_1)\right)^2 + \frac{\partial \ell_t}{\partial \mu}\mathbf{e}_1\mathbf{e}_1^\top h^{-1\, ''}_1(\phi_1),
\end{align*}
\begin{align*}
    \nabla_{\boldsymbol{\phi}_v}\frac{\partial \ell_t(\boldsymbol{\theta})}{\partial\phi_\varepsilon} = \frac{\partial^2\ell_t}{\partial \sigma^2_t \partial \theta_\varepsilon}\frac{\partial\theta_\varepsilon}{\partial \phi_\varepsilon}(\nabla_{\boldsymbol{\phi}_v} \sigma^2_t\left(\boldsymbol{\theta}_v)\right)^\top + \frac{\partial^2\ell_t}{\partial \mu \partial \theta_\varepsilon}\frac{\partial\theta_\varepsilon}{\partial \phi_\varepsilon}\mathbf{e}_1^\top h^{-1\, '}_1(\phi_1),
\end{align*}
and
\begin{align*}
    \frac{\partial^2 \ell_t(\boldsymbol{\theta})}{\partial\phi_\varepsilon^2} = \frac{\partial^2\ell_t}{\partial \theta_\varepsilon^2}\frac{\partial\theta_\varepsilon}{\partial \phi_\varepsilon}\frac{\partial\theta_\varepsilon}{\partial \phi_\varepsilon}^\top + \frac{\partial \ell_t}{\partial \theta_\varepsilon}\frac{\partial^2\theta_\varepsilon}{\partial \phi_\varepsilon^2}.
\end{align*}
\end{enumerate}
\end{lemma}

\putbib
\end{bibunit}

\newpage

% For sections: prefix with S
\renewcommand{\thesection}{S\arabic{section}}

% For subsections (optional, remove if not needed)
\renewcommand{\thesubsection}{\thesection.\arabic{subsection}}

% For equations: prefix with S<section>.<equation>
\renewcommand{\theequation}{S\arabic{section}.\arabic{equation}}

% For figures and tables and lemmas
\setcounter{section}{0}
\setcounter{equation}{0}
\setcounter{lemma}{0}
\setcounter{remark}{0}
\setcounter{proposition}{0}
\setcounter{figure}{0}
\setcounter{table}{0}

\renewcommand{\thefigure}{S\arabic{figure}}
\renewcommand{\thetable}{S\arabic{table}}
\renewcommand{\thelemma}{S\arabic{lemma}}
\renewcommand{\theremark}{S\arabic{remark}}
\renewcommand{\theproposition}{S\arabic{proposition}}

\setcounter{page}{1}
\renewcommand{\thepage}{S\arabic{page}}

\begin{bibunit}

\section{Proofs}\label{supp:Proofs}

\begin{proof}[Proof of Lemma \ref{lem:taylor_remainder}]
For any $\boldsymbol{\phi}$ satisfying $$\|\boldsymbol{\phi} - \boldsymbol{\phi}^\star\|\leq r,$$ by Taylor's theorem with a Lagrange remainder, there exists a point $\widetilde{\boldsymbol{\phi}}$ on the line segment between $\boldsymbol{\phi}^\star$ and $\boldsymbol{\phi}$ such that
\begin{align*}
    \ell_t(\boldsymbol{\phi}) = \ell_t(\boldsymbol{\phi}^\star) 
+ \nabla^\top_{\boldsymbol{\phi}} \ell_t(\boldsymbol{\phi}^\star)(\boldsymbol{\phi} - \boldsymbol{\phi}^\star) + \frac{1}{2}(\boldsymbol{\phi} - \boldsymbol{\phi}^\star)^\top\nabla^2_{\boldsymbol{\phi}} \ell_t(\boldsymbol{\phi}^\star) (\boldsymbol{\phi} -\boldsymbol{\phi}^\star) + R_t(\boldsymbol{\phi}),
\end{align*}
where the remainder term satisfies
\begin{align}\label{eq:bounded_remainder_lagrange}
    |R_t(\boldsymbol{\phi})| & \leq \frac{1}{6}\| \nabla^3_{\boldsymbol{\phi}}\ell_t(\widetilde{\boldsymbol{\phi}})\|\|\boldsymbol{\phi} - \boldsymbol{\phi}^\star\|^3.
\end{align}
The norm $\| \nabla^3_{\boldsymbol{\phi}}\ell_t(\cdot)\|$, defined in Lemma \ref{lem:taylor_remainder}, measures the largest possible third-order directional derivative of $\ell_t$ at its argument over all unit directions in parameter space. This derivative is given by 
$$D^3\ell_t(\boldsymbol{\phi})(\mathbf{u},\mathbf{v},\mathbf{w})=\sum_{i,j,k}\frac{\partial^3\ell_t(\boldsymbol{\phi})}{\partial \phi_i \partial \phi_j \partial \phi_k}u_i v_j w_k.$$

Since $\widetilde{\boldsymbol{\phi}}$ lies on the line segment between $\boldsymbol{\phi}^\star$ and $\boldsymbol{\phi}$, there exists an $a\in[0,1]$ such that $\widetilde{\boldsymbol{\phi}}=\boldsymbol{\phi}^\star+a(\boldsymbol{\phi}-\boldsymbol{\phi}^\star)$. Hence,
$$\|\widetilde{\boldsymbol{\phi}} - \boldsymbol{\phi}^\star\| = a\|\boldsymbol{\phi} - \boldsymbol{\phi}^\star\| \leq \|\boldsymbol{\phi} - \boldsymbol{\phi}^\star\|.$$
Therefore, if $\boldsymbol{\phi} \in \mathcal{B}_T$, then also $\widetilde{\boldsymbol{\phi}} \in \mathcal{B}_T$. In particular, for sufficiently large $T$, $\widetilde{\boldsymbol{\phi}}$ lies in the region $\|\boldsymbol{\phi} - \boldsymbol{\phi}^\star\|\leq r$, where the assumed bound $\|\nabla_{\boldsymbol{\phi}}^3 \ell_t(\widetilde{\boldsymbol{\phi}})\|\leq K$ holds. Consequently, using \eqref{eq:bounded_remainder_lagrange}, 
$$|e_t(\boldsymbol{\phi})|=|R_t(\boldsymbol{\phi})| \leq \frac{K}{6}\| \boldsymbol{\phi} - \boldsymbol{\phi}^\star\|^3.$$
Then, for $\boldsymbol{\phi}\in \mathcal{B}_T$, we have
$$\|\boldsymbol{\phi}- \boldsymbol{\phi}^\star\| \leq C_0T^{-1/2},$$ and therefore
$$|e_t(\boldsymbol{\phi})| \leq \frac{K}{6}C_0^3T^{-3/2}.$$
Taking the supremum over $\boldsymbol{\phi}\in\mathcal{B}_T$ and the maximum over $1\leq t\leq T$ yields
\[
\mathop{\max\vphantom{\sup}}\limits_{1 \le t \le T}\;
\mathop{\sup\vphantom{\max}}\limits_{\boldsymbol{\phi}\in \mathcal{B}_T}
|e_t(\boldsymbol{\phi})|
= \mathcal{O}(T^{-3/2}).
\]
\end{proof}

\begin{proof}[Proof of Lemma \ref{lem:tail_smaller}]
To prove part (i), write
\[
\varepsilon(\gamma)
=
\frac{{t^\star}^{-\gamma}(T-t^\star)}
{{t^\star}^{-\gamma}(T-t^\star) + \sum_{j=1}^{t^\star}j^{-\gamma}}
=
\frac{1}{1+R(\gamma)},
\]
where
\[
R(\gamma)
=
\frac{\sum_{j=1}^{t^\star}j^{-\gamma}}{{t^\star}^{-\gamma}(T-t^\star)}
=
\frac{1}{T-t^\star}\sum_{j=1}^{t^\star}\left(\frac{t^\star}{j}\right)^{\gamma}.
\]
For each $j=1,\dots,t^\star$, we have $t^\star/j\ge 1$, so the function
$\left(\frac{t^\star}{j}\right)^{\gamma}$ is nondecreasing in $\gamma$.
Hence $R(\gamma)$ is nondecreasing in $\gamma$, and since $\varepsilon(\gamma)=1/(1+R(\gamma))$,
it follows that $\varepsilon(\gamma)$ is non-increasing in $\gamma$. If $t^\star \geq 2$, then for any $j\in \{1, \dots, t^\star - 1\}$ we have
$$\frac{t^{\star}}{j} > 1, $$
so the function $$\left(\frac{t^\star}{j}\right)^\gamma$$ is strictly increasing in $\gamma$ for such $j$. Since $R(\gamma)$ is a sum of non-decreasing terms with at least one strictly increasing term, $R(\gamma)$ is strictly increasing in $\gamma$, which implies that $\varepsilon(\gamma)$ is strictly decreasing in $\gamma$.

To prove part (ii), suppose $\gamma > 0$. Then, for each $j=1,\dots, t^\star$ ,
$$j^{-\gamma} \geq {t^\star}^{-\gamma} \quad \text{with equality only for } j=t^\star \text{},$$
and thus
$$\sum_{j=1}^{t^\star} j^{-\gamma} > \sum_{j=1}^{t^\star} {t^\star}^{-\gamma}= t^\star ({t^\star}^{-\gamma}).$$
Now, 
\begin{align*}
    \varepsilon(\gamma) & =
\frac{{t^\star}^{-\gamma}(T-t^\star)}
{{t^\star}^{-\gamma}(T-t^\star) + \sum_{j=1}^{t^\star}j^{-\gamma}} \\
& < \frac{{t^\star}^{-\gamma}(T-t^\star)}
{{t^\star}^{-\gamma}(T-t^\star) + t^\star ({t^\star}^{-\gamma})} \\
& < \frac{T-t^\star}
{T},
\end{align*}
which establishes the strict inequality for $\gamma > 0$. Evaluating $\varepsilon(0)$ gives $\varepsilon(0)=(T-t^\star)/T$, which proves the equality case.

To prove part (iii), recall that by \eqref{eq:trunc_probs_TPD}, for $t\in\mathcal{H}$,
\begin{align*}
p_t
&=(1-\varepsilon(\gamma))\,
\frac{t^{-\gamma}}{\sum_{j=1}^{t^\star}j^{-\gamma}}.
\end{align*}
Since $t^{-\gamma}$ is non-increasing in $t$ for $\gamma\ge 0$, the smallest
head probability occurs at $t=t^\star$, so $p_t \ge p_{t^\star}$ for all
$t\in\mathcal{H}$. Moreover, by \eqref{eq:eps_gamma}, it follows that
\begin{align*}
1-\varepsilon(\gamma)
& =
\frac{\sum_{j=1}^{t^\star}j^{-\gamma}}
{{t^\star}^{-\gamma}(T-t^\star)
+\sum_{j=1}^{t^\star}j^{-\gamma}}.
\end{align*}
Hence,
\begin{align*}
p_{t^\star}
&=
(1-\varepsilon(\gamma))\,
\frac{{t^\star}^{-\gamma}}
{\sum_{j=1}^{t^\star}j^{-\gamma}} \\
&=
\frac{{t^\star}^{-\gamma}}
{{t^\star}^{-\gamma}(T-t^\star)
+\sum_{j=1}^{t^\star}j^{-\gamma}} \\
&=
\frac{\varepsilon(\gamma)}{T-t^\star}.
\end{align*}
Therefore, for all $t\in\mathcal{H}$,
\begin{align*}
p_t \ge p_{t^\star}=\frac{\varepsilon(\gamma)}{T-t^\star}.
\end{align*}
\end{proof}

\begin{proof}[Proof of Lemma \ref{lem:tail_smaller_higher_variance}]
Under homogeneous control variate residuals, $$e(\boldsymbol{\phi})=\sum_{t=1}^Te_t(\boldsymbol{\phi})=T\overline{e}(\boldsymbol{\phi}).$$
Then, by Lemma \ref{lem:expectation_variance_WDE}(ii), and focusing only on the tail contributions,
\begin{align}\label{eq:variance_from_tail}
    \mathbb{V}_\mathcal{T}(\gamma)= \frac{\overline{e}(\boldsymbol{\phi})^2}{m}\sum_{t\in \mathcal{T}} \frac{(1-Tp_t(\gamma))^2}{p_t(\gamma)},
\end{align}
where $p_t(\gamma)$ denotes the truncated sampling probabilities in \eqref{eq:trunc_probs_TPD}. By Lemma \ref{lem:tail_smaller}(ii), $$p_t(\gamma)\leq \frac{1}{T}, \,\, t\in \mathcal{T}.$$
Each term in the sum in \eqref{eq:variance_from_tail} has functional form
$$f(p) = \frac{(1 - Tp)^2}{p}, \,\, 0 < p < \frac{1}{T},$$
which is strictly decreasing on $(0, 1/T)$ since $$f'(p) = -1/p^2 + T^2 < 0.$$ 
Moreover, Lemma \ref{lem:tail_smaller}(i) shows that $\varepsilon(\gamma)$ is non-increasing in $\gamma \geq 0$. Therefore, for any $0\leq \gamma_1 < \gamma_2$,
$$p_t(\gamma_1) = \frac{\varepsilon(\gamma_1)}{T-t^\star} \geq  \frac{\varepsilon(\gamma_2)}{T-t^\star}= p_t(\gamma_2),\,\, t\in  \mathcal{T}.$$
Since $f$ is strictly decreasing on $(0,1/T)$, it follows that $$f(p_t(\gamma_1)) < f(p_t(\gamma_2)),\,\, t \in \mathcal{T},$$ 
and summing over $t\in\mathcal{T}$ yields
$$\mathbb{V}_\mathcal{T}(\gamma_1) <     \mathbb{V}_\mathcal{T}(\gamma_2).$$
\end{proof}

\begin{proof}[Proof of Lemma \ref{lem:upper_bound_variance}]
   From Lemma \ref{lem:expectation_variance_WDE}(ii), $$\mathbb{V}\left(\widehat{\ell}_{\mathrm{WDE}}(\boldsymbol{\phi})\right) =  \frac{1}{m}\sum_{t=1}^T\left(\frac{e_t}{p_t} - e\right)^2 p_t = \frac{1}{m} \left(\sum_{t=1}^T\frac{e^2_t}{p_t} - e^2\right).$$
   For uniform weights $p_t=1/T$,
$$\mathbb{V}\left(\widehat{\ell}_{\mathrm{DE}}(\boldsymbol{\phi})\right) = \frac{1}{m} \left(\sum_{t=1}^TTe^2_t - e^2\right).$$
   Since $p_t \geq \frac{c}{T}$ for all $t=1,\dots, T$,
   $$\frac{1}{p_t} \leq \frac{T}{c} 
\;\;\implies\;\; \frac{e_t^2}{p_t} \leq \frac{T}{c}e_t^2
\;\;\implies\;\; \sum_{t=1}^T \frac{e_t^2}{p_t} \leq \frac{T}{c}\sum_{t=1}^T e_t^2.$$
   Thus,  
$$\mathbb{V}\left(\widehat{\ell}_{\mathrm{WDE}}(\boldsymbol{\phi})\right) = \frac{1}{m} \left(\sum_{t=1}^T\frac{e^2_t}{p_t} - e^2\right) \leq \frac{1}{m}\left( \frac{T}{c}\sum_{t=1}^T e^2_t - e^2\right).$$
Taking the ratio with the difference estimator (uniform weights) gives
   \begin{align*}
       \frac{\mathbb{V}(\widehat{\ell}_\mathrm{WDE})}{\mathbb{V}(\widehat{\ell}_\mathrm{DE})}  & \leq \frac{\frac{T}{c}\sum_{t=1}^T e^2_t - e^2}{T\sum_{t=1}^T e^2_t - e^2} =  \frac{\frac{1}{c} - \frac{e^2}{T\sum_{t=1}^T e^2_t}}{1 - \frac{e^2}{T\sum_{t=1}^T e^2_t} }.
   \end{align*}
   Finally, by the Cauchy-Schwarz inequality
   $$e^2=\left(\sum_{t=1}^T e_t1\right)^2 \leq \left(\sum_{t=1}^T e^2_t \right) \left(\sum_{t=1}^T 1^2 \right)= T\left(\sum_{t=1}^T e^2_t \right),$$
  which implies $\rho \in (0,1)$.  
\end{proof}

\begin{proof}[Proof of Lemma \ref{lem:asymptotic_variance_estimators_truncated}]
First note that
$$\mathbb{V}\left(\widehat{\ell}_{\mathrm{WDE}}(\boldsymbol{\phi})\right) = \frac{1}{m} \left(\sum_{t=1}^T\frac{e^2_t}{p_t} - e^2\right) \leq \frac{1}{m} \sum_{t=1}^T\frac{e^2_t}{p_t}.$$
By the safeguard in \eqref{eq:safe_guard_floor} together with Lemma \ref{lem:tail_smaller}(iii), the sampling probabilities satisfy $p_t\geq c/T$ for $t=1,\dots, T$, and hence $1/p_t \leq T/c$. Therefore
$$\sum_{t=1}^T\frac{e^2_t}{p_t}\leq \frac{T}{c}\sum_{t=1}^Te^2_t.$$
Since $e_t = \mathcal{O}(T^{-3/2})$, it follows that $e^2_t= \mathcal{O}(T^{-3})$ and thus $\sum_{t=1}^Te^2_t=\mathcal{O}(T^{-2})$. This gives
$$\mathbb{V}\left(\widehat{\ell}_{\mathrm{WDE}}(\boldsymbol{\phi})\right) \leq \frac{1}{m}\frac{T}{c} \mathcal{O}(T^{-2})= \frac{1}{m}\mathcal{O}(T^{-1}).$$
\end{proof}

\begin{proof}[Proof of Lemma \ref{lem:average_cost}]
For part (i), note that $u_{\max}$ is a non-negative discrete random variable with sample space $\{1,\dots,T\}$. Using the tail sum formula,
$$
\mathbb{E}(u_{\max})
=
\sum_{k=1}^T \Pr(u_{\max}\geq k),
$$
where
\begin{align*}
    \Pr(u_{\max}\geq k)
    &=
    1-\Pr(u_{\max}\leq k-1) \\
    &=
    1-\left(\sum_{t=1}^{k-1}p_t\right)^m.
\end{align*}
Summing over $k$ gives the stated expression. The upper bound
$\mathbb{E}(u_{\max})\leq T$ follows immediately since the term subtracted from $T$ is non-negative.

For part (ii), note that
$$
u_{\max}
\leq
t^\star+(T-t^\star)\mathbbm{1}(u_{\max}>t^\star),
$$
where $\mathbbm{1}(A)$ is one if the event $A$ occurs and zero otherwise. Taking expectations gives
\begin{align*}
    \mathbb{E}(u_{\max})
    &\leq
    t^\star+(T-t^\star)\Pr(u_{\max}>t^\star) \\
    &=
    t^\star+(T-t^\star)
    \left(1-\Pr(u_{\max}\leq t^\star)\right).
\end{align*}
Since $u_{\max}=\max\{u_1,\dots,u_m\}$ and the $u_i$ are independent with
$\Pr(u_i\leq t^\star)=1-\varepsilon$,
$$
\Pr(u_{\max}\leq t^\star)
=
\Pr(u_1\leq t^\star,\dots,u_m\leq t^\star)
=
(1-\varepsilon)^m,
$$
which gives \eqref{eq:exact_compute_bound_TPD}. To establish monotonicity, define
$$
B(\varepsilon)
=
t^\star+(T-t^\star)\left(1-(1-\varepsilon)^m\right).
$$
Then
$$
B'(\varepsilon)
=
(T-t^\star)m(1-\varepsilon)^{m-1}>0
$$
for all $\varepsilon\in(0,1)$ and $m\geq1$, since $T>t^\star$.

For part (iii), define
$$
M_m=\max\{u_1,\dots,u_m\},
\qquad
M_{m+1}=\max\{M_m,u_{m+1}\}.
$$
Since $M_{m+1}\geq M_m$, it is sufficient to show that
$$
\Pr(M_{m+1}>M_m)>0.
$$
Now,
$$
\{M_{m+1}>M_m\}
\supseteq
\{M_m\leq t^\star\}
\cap
\{u_{m+1}\in\mathcal{T}\}.
$$
Using independence,
$$
\Pr(M_{m+1}>M_m)
\geq
(1-\varepsilon(\gamma))^m\varepsilon(\gamma)>0,
$$
since $\varepsilon(\gamma)\in(0,1)$. Hence,
$\mathbb{E}(M_{m+1})>\mathbb{E}(M_m)$.
\end{proof}

\begin{proof}[Proof of Lemma \ref{lem:expectation_variance_WDE_grad}]
The proof of part (i) follows from Lemma \ref{lem:expectation_variance_WDE}(i) and the linearity of the expectation and the gradient. Part (ii) is a straightforward extension of Lemma \ref{lem:expectation_variance_WDE}(ii), replacing scalar residuals with vector-valued residuals.
\end{proof}

\begin{proof}[Proof of Lemma \ref{lem:reparam_grad_Hess_variance}]
First, to express the differential of $\boldsymbol{\theta}_v$ in terms of the differential of $\boldsymbol{\phi}_v$, note that
\begin{align}\label{eq:diff_theta_as_diff_phi}
d\boldsymbol{\theta}_v & = J^{\boldsymbol{\phi}_v} d\boldsymbol{\phi}_v,
\end{align}
where 
$$J^{\boldsymbol{\phi}_v} = \nabla_{\boldsymbol{\phi}_v} \boldsymbol{\theta}_v = \frac{\partial \boldsymbol{\theta}_v}{\partial\boldsymbol{\phi}_v^\top} =\frac{\partial h^{-1}(\boldsymbol{\phi}_v)}{\partial\boldsymbol{\phi}_v^\top}.$$
Combining this with the first identification theorem in \eqref{eq:first_identification_theorem_grad},
$$d\sigma^2_t(\boldsymbol{\theta}_v)= \nabla^\top_{\boldsymbol{\theta}_v}\sigma^2_t(\boldsymbol{\theta}_v)d\boldsymbol{\theta}_v = \nabla^\top_{\boldsymbol{\theta}_v}\sigma^2_t(\boldsymbol{\theta}_v)J^{\boldsymbol{\phi}_v} d\boldsymbol{\phi}_v.$$
We can now use the first identification theorem again to identify that the gradient with respect to $\boldsymbol{\phi}_v$ is
\begin{align}\label{eq:general_gradient_phi}
\nabla_{\boldsymbol{\phi}_v}\sigma^2_t(\boldsymbol{\theta}_v)  & = \nabla_{\boldsymbol{\phi}_v}\sigma^2_t(h^{-1}(\boldsymbol{\phi}_v))  = \left(\nabla^\top_{\boldsymbol{\theta}_v}\sigma^2_t(\boldsymbol{\theta}_v)J^{\boldsymbol{\phi}_v}\right)^\top=J^{\boldsymbol{\phi}_v\top} \nabla_{\boldsymbol{\theta}_v}\sigma^2_t(\boldsymbol{\theta}_v).    
\end{align}
Note that since the transformations are component-wise, i.e.\ $h^{-1}_i$ depend only on $\phi_{vi}$, the Jacobian is a diagonal matrix
$$J^{\boldsymbol{\phi}_v} = \mathrm{diag}\left(h^{-1\, '}_1(\phi_{v1}), \dots, h^{-1\, '}_{d_v}(\phi_{vd_v})\right).$$
To get the Hessian, we use the result that the differential of the gradient is the Hessian and the product rule, together with \eqref{eq:diff_theta_as_diff_phi},
\begin{align*}
d\nabla_{\boldsymbol{\phi}_v}\sigma^2_t(\boldsymbol{\theta}_v) & = d(J^{\boldsymbol{\phi}_v\top}) \nabla_{\boldsymbol{\theta}_v}\sigma^2_t(\boldsymbol{\theta}_v) + J^{\boldsymbol{\phi}_v\top}d\nabla_{\boldsymbol{\theta}_v}\sigma^2_t(\boldsymbol{\theta}_v) \\
& = d(J^{\boldsymbol{\phi}_v\top}) \nabla_{\boldsymbol{\theta}_v}\sigma^2_t(\boldsymbol{\theta}_v) + J^{\boldsymbol{\phi}_v\top}\nabla^2_{\boldsymbol{\theta}_v}\sigma^2_t(\boldsymbol{\theta}_v)d\boldsymbol{\theta}_v \\
& = d(J^{\boldsymbol{\phi}_v\top}) \nabla_{\boldsymbol{\theta}_v}\sigma^2_t(\boldsymbol{\theta}_v) + J^{\boldsymbol{\phi}_v\top}\nabla^2_{\boldsymbol{\theta}_v}\sigma^2_t(\boldsymbol{\theta}_v)J^{\boldsymbol{\phi}_v}d\boldsymbol{\phi}_v.
\end{align*}
To understand the differential $d(J^{\boldsymbol{\phi}_v\top})$, first note that $$J^{\boldsymbol{\phi}_v}_{ij} = \frac{\partial}{\partial\phi_{vj}}h_i^{-1}(\phi_{vi}),$$
and thus in the transposed matrix $J^{\boldsymbol{\phi}_v\top}$ the $i$th column contains the gradient $\nabla_{\boldsymbol{\phi}_v}h_i^{-1}(\phi_{vi})$, i.e.\
$$J^{\boldsymbol{\phi}_v\top} = \begin{pmatrix}
    \nabla_{\boldsymbol{\phi}_v}h_1^{-1}(\phi_{v1}) & \cdots & \nabla_{\boldsymbol{\phi}_v}h_{d_v}^{-1}(\phi_{vd_v}),
\end{pmatrix}$$
stacked as columns. The resulting $d(J^{\boldsymbol{\phi}_v\top})$ is a $d_v\times d_v \times d_v$ tensor, where each $$d \nabla_{\boldsymbol{\phi}_v} h_i^{-1}(\phi_{vi}) = \nabla^2_{\boldsymbol{\phi}_v} h_i^{-1}(\phi_{vi})d\boldsymbol{\phi}_v,$$
which gives the tensor-vector contraction
$$d(J^{\boldsymbol{\phi}_v\top}) \nabla_{\boldsymbol{\theta}_v}\sigma^2_t(\boldsymbol{\theta}_v) = \sum_{i=1}^{d_v} \left(\nabla_{\boldsymbol{\theta}_v}\sigma^2_t(\boldsymbol{\theta}_v)\right)_i d \nabla_{\boldsymbol{\phi}_v} h_i^{-1}(\phi_{vi}) = \sum_{i=1}^{d_v} \left(\nabla_{\boldsymbol{\theta}_v}\sigma^2_t(\boldsymbol{\theta}_v)\right)_i \nabla^2_{\boldsymbol{\phi}_v} h_i^{-1}(\phi_{vi})d\boldsymbol{\phi}_v.$$
The final expression becomes
$$d\nabla_{\boldsymbol{\phi}_v}\sigma^2_t(\boldsymbol{\theta}_v) = \sum_{i=1}^{d_v}\left(\nabla_{\boldsymbol{\theta}_v}\sigma^2_t(\boldsymbol{\theta}_v)\right)_i \nabla^2_{\boldsymbol{\phi}_v} h_i^{-1}(\phi_{vi})d\boldsymbol{\phi}_v + J^{\boldsymbol{\phi}_v\top}\nabla^2_{\boldsymbol{\theta}_v}\sigma^2_t(\boldsymbol{\theta}_v)J^{\boldsymbol{\phi}_v}d\boldsymbol{\phi}_v,$$
from which we get the Hessian matrix as
\begin{align}\label{eq:general_Hess_phi}
\nabla^2_{\boldsymbol{\phi}_v}\sigma^2_t(\boldsymbol{\theta}_v) & = J^{\boldsymbol{\phi}_v\top}\nabla^2_{\boldsymbol{\theta}_v}\sigma^2_t(\boldsymbol{\theta}_v)J^{\boldsymbol{\phi}_v} + \sum_{i=1}^{d_v} \left(\nabla_{\boldsymbol{\theta}_v}\sigma^2_t(\boldsymbol{\theta}_v)\right)_i \nabla^2_{\boldsymbol{\phi}_v} h_i^{-1}(\phi_{vi}) \nonumber \\
 & = J^{\boldsymbol{\phi}_v\top}\nabla^2_{\boldsymbol{\theta}_v}\sigma^2_t(\boldsymbol{\theta}_v)J^{\boldsymbol{\phi}_v} + \sum_{i=1}^{d_v} \left(\nabla_{\boldsymbol{\theta}_v}\sigma^2_t(\boldsymbol{\theta}_v)\right)_i  h^{-1\, ''}_i(\phi_{vi})\mathbf{e}_i\mathbf{e}^\top_i.
\end{align}
\end{proof}

\begin{proof}[Proof of Lemma \ref{lem:reparam_grad_Hess_log_dens}]
Suppose that $$\boldsymbol{\phi}=h(\boldsymbol{\theta}), \quad \boldsymbol{\theta} = (\boldsymbol{\theta}_v^\top, \theta_\varepsilon)^\top, \,\,
\boldsymbol{\phi} = (\boldsymbol{\phi}_v^\top, \phi_\varepsilon)^\top,$$ 
is a one-to-one, elementwise transformation to an unrestricted parameterisation, with inverse
$$\boldsymbol{\theta}=h^{-1}(\boldsymbol{\phi}),$$
and in particular $\mu=h^{-1}_1(\phi_1)$. Then, by the chain rule similarly to \eqref{eq:gradient_theta},
\begin{align}\label{eq:gradient_phi}
\nabla_{\boldsymbol{\phi}_v}\ell_t(\sigma^2_t(\boldsymbol{\theta}_v), \mu, \theta_\varepsilon) & = \frac{\partial \ell_t}{\partial\sigma^2_t}\nabla_{\boldsymbol{\phi}_v}\sigma^2_t(\boldsymbol{\theta}_v) + \frac{\partial \ell_t}{\partial\mu} \nabla_{\boldsymbol{\phi}_v}\mu \nonumber \\
 & = \frac{\partial \ell_t}{\partial\sigma^2_t}\nabla_{\boldsymbol{\phi}_v}\sigma^2_t(\boldsymbol{\theta}_v) + \frac{\partial \ell_t}{\partial\mu} \mathbf{e}_1 h^{-1\, '}_1(\phi_1), 
\end{align}
where $\nabla_{\boldsymbol{\phi}_v}\sigma^2_t(\boldsymbol{\theta}_v)$ is given in \eqref{eq:general_gradient_phi}. For the Student-$t$ density, similarly to \eqref{eq:gradient_theta}, the last component of the gradient with respect to the transformed parameters is
$$\frac{\partial \ell_t}{\phi_\varepsilon}=\frac{\partial \ell_t}{\partial\theta_\varepsilon}\frac{\partial \theta_\varepsilon}{\partial \phi_\varepsilon}, $$
where $\theta_\varepsilon = h^{-1}(\phi_\varepsilon)$. To obtain the Hessian, we again rely on the chain rule to account for the parameter transformation. In particular, the top-left component is given by
\begin{align*}
\nabla^2_{\boldsymbol{\phi}_v}\ell_t(\sigma^2_t(\boldsymbol{\theta}_v), \mu, \theta_\varepsilon) 
&= \nabla_{\boldsymbol{\phi}_v} \left(\frac{\partial \ell_t}{\partial\sigma^2_t}\nabla_{\boldsymbol{\phi}_v}\sigma^2_t(\boldsymbol{\theta}_v) + \frac{\partial \ell_t}{\partial\mu} \mathbf{e}_1 h^{-1\, '}_1(\phi_1)\right)\\
&= \left(\nabla_{\boldsymbol{\phi}_v}\frac{\partial \ell_t}{\partial \sigma^2_t}\right)\nabla_{\boldsymbol{\phi}_v}\sigma^2_t(\boldsymbol{\theta}_v)^\top + \frac{\partial \ell_t}{\partial \sigma^2_t}\nabla^2_{\boldsymbol{\phi}_v}\sigma^2_t(\boldsymbol{\theta}_v) + \left(\nabla_{\boldsymbol{\phi}_v}\frac{\partial \ell_t}{\partial \mu}\right)\mathbf{e}_1^\top h^{-1\, '}_1(\phi_1)\\
&\quad + \frac{\partial \ell_t}{\partial \mu}\left(\nabla_{\boldsymbol{\phi}_v}\mathbf{e}_1h^{-1\, '}_1(\phi_1)\right)\\
&= \frac{\partial^2\ell_t}{\partial(\sigma^2_t)^2}\nabla_{\boldsymbol{\phi}_v}\sigma^2_t(\boldsymbol{\theta}_v)\nabla_{\boldsymbol{\phi}_v}\sigma^2_t(\boldsymbol{\theta}_v)^\top + \frac{\partial \ell_t}{\partial \sigma^2_t}\nabla^2_{\boldsymbol{\phi}_v}\sigma^2_t(\boldsymbol{\theta}_v)\\
&\quad + \frac{\partial^2 \ell_t}{\partial \sigma^2_t\partial \mu}\left(\mathbf{e}_1 \nabla_{\boldsymbol{\phi}_v}\sigma^2_t(\boldsymbol{\theta}_v) + \nabla_{\boldsymbol{\phi}_v}\sigma^2_t(\boldsymbol{\theta}_v) \mathbf{e}_1 ^\top \right) h^{-1\, '}_1(\phi_1) \\
&\quad + \frac{\partial^2 \ell_t}{\partial \mu^2} \mathbf{e}_1 \mathbf{e}_1 ^\top\left(h^{-1\, '}_1(\phi_1)\right)^2 + \frac{\partial \ell_t}{\partial \mu}\mathbf{e}_1\mathbf{e}_1^\top h^{-1\, ''}_1(\phi_1),
\end{align*}
where $\nabla_{\boldsymbol{\phi}_v}\sigma^2_t(\boldsymbol{\theta}_v)$ and $\nabla^2_{\boldsymbol{\phi}_v}\sigma^2_t(\boldsymbol{\theta}_v)$ are given as in \eqref{eq:general_gradient_phi} and \eqref{eq:general_Hess_phi} respectively. The cross-terms of the Hessian can be computed as
\begin{align*}  \nabla_{\boldsymbol{\phi}_v}\frac{\partial \ell_t}{\partial\phi_\varepsilon} 
&= \nabla_{\boldsymbol{\phi}_v}\left(\frac{\partial \ell_t}{\partial\theta_\varepsilon}\frac{\partial \theta_\varepsilon}{\partial \phi_\varepsilon}\right)\\
&= \frac{\partial^2\ell_t}{\partial \sigma^2_t \partial \theta_\varepsilon}\frac{\partial\theta_\varepsilon}{\partial \phi_\varepsilon}(\nabla_{\boldsymbol{\phi}_v} \sigma^2_t\left(\boldsymbol{\theta}_v)\right)^\top + \frac{\partial^2\ell_t}{\partial \mu \partial \theta_\varepsilon}\frac{\partial\theta_\varepsilon}{\partial \phi_\varepsilon}\mathbf{e}_1^\top h^{-1\, '}_1(\phi_1),
\end{align*}
and the final bottom-right component of the Hessian is
\begin{align*}
    \frac{\partial^2 \ell_t}{\partial\phi_\varepsilon^2} = \frac{\partial^2\ell_t}{\partial \theta_\varepsilon^2}\frac{\partial\theta_\varepsilon}{\partial \phi_\varepsilon}\frac{\partial\theta_\varepsilon}{\partial \phi_\varepsilon}^\top + \frac{\partial \ell_t}{\partial \theta_\varepsilon}\frac{\partial^2\theta_\varepsilon}{\partial \phi_\varepsilon^2},
\end{align*}
where $\theta_\varepsilon = h^{-1}(\phi_\varepsilon)$.
\end{proof}

\begin{proof}[Proof of Lemma \ref{lem:asymtotic_variance_estimators}] In what follows, expressions involving subtraction of order terms are interpreted as upper bounds, and the dominant order in $T$ is retained.
    For part (i), rewrite $$\frac{1}{m}\sum_{t=1}^T\left(\frac{e_t}{p_t} - e\right)^2 p_t = \frac{1}{m}\sum_{t=1}^T\left(\frac{e_t}{p_t} - T\overline{e}\right)^2 p_t,$$
    with $\overline{e} = e/T$ of the same order as $e_t$, i.e.\ $\mathcal{O}(T^{-3/2})$. Since 
    $$p_t = \frac{1/T}{\sum_{j=1}^T 1/T}, $$
    $p_t = \mathcal{O}(T^{-1})$, which gives
    \begin{align*}
        \frac{1}{m}\sum_{t=1}^T\left(\frac{e_t}{p_t} - T\overline{e}\right)^2 p_t & = \frac{1}{m}\sum_{t=1}^T\left(\mathcal{O}(T)\mathcal{O}(T^{-3/2}) - T\mathcal{O}(T^{-3/2})\right)^2 \mathcal{O}(T^{-1}) \\
        & = \frac{1}{m}T\left(\mathcal{O}(T^{-1/2})\right)^2 \mathcal{O}(T^{-1}) \\
        & = \frac{1}{m}\mathcal{O}(T^{-1}).
    \end{align*}
    For part (ii), under the power-law decaying probabilities, first note that the normalising constant is $\mathcal{O}(1)$ in $T$ since the denominator in
    $$p_t = \frac{1/t^\lambda}{\sum_{j=1}^T 1/j^\lambda} $$
    converges to Riemann's zeta function $\zeta(\lambda)$ as $T \rightarrow \infty$ for $\lambda>1$. Thus $p_t = \mathcal{O}(1)1/t^{\lambda}$, which gives
    \begin{align*}
        \frac{1}{m}\sum_{t=1}^T\left(\frac{e_t}{p_t} - T\overline{e}\right)^2 p_t & = \frac{1}{m}\sum_{t=1}^T\left(\mathcal{O}(1)t^{\lambda}\mathcal{O}(T^{-3/2}) - T\mathcal{O}(T^{-3/2})\right)^2 \mathcal{O}(1)t^{-\lambda} \\
        & = \frac{1}{m} \sum_{t=1}^T\left(t^{2\lambda}\mathcal{O}(T^{-3}) - t^\lambda \mathcal{O}(T^{-2}) + \mathcal{O}(T^{-1})\right)  t^{-\lambda} \\
        & = \frac{1}{m} \left( \mathcal{O}(T^{-3})\sum_{t=1}^T t^{\lambda} - \mathcal{O}(T^{-2})\sum_{t=1}^T 1 + \mathcal{O}(T^{-1})\sum_{t=1}^T t^{-\lambda}\right).
    \end{align*}
    From above $\sum_{t=1}^T t^{-\lambda}=\mathcal{O}(1)$, and by Faulhaber's formula (see, for example \citealp{knuth1993johann}) 
$$\sum_{t=1}^T t^\lambda = \mathcal{O}(T^{\lambda + 1}).$$
Thus
\begin{align}\label{eq:order_variance_PD}
        \frac{1}{m}\sum_{t=1}^T\left(\frac{e_t}{p_t} - T\overline{e}\right)^2 p_t & = \frac{1}{m} \left( \mathcal{O}(T^{\lambda - 2}) - \mathcal{O}(T^{-1}) + \mathcal{O}(T^{-1})\right) \nonumber \\
        & = \frac{1}{m} \mathcal{O}(T^{\lambda - 2}),
\end{align}
for $\lambda > 1$. Note that in \eqref{eq:order_variance_PD} the variance goes to zero when $1 < \lambda < 2$, is bounded when $\lambda = 2$, and grows when $\lambda > 2$.

For the exponentially decaying probabilities, note that the normalising constant is $\mathcal{O}(1)$ in $T$ since the denominator in
$$p_t = \frac{\exp(-\kappa t)}{\sum_{j=1}^T \exp(-\kappa j)},$$
is a geometric sum $$\sum_{j=1}^T \left(\exp(-\kappa)\right)^j, \quad  \text{ with } |\exp(-\kappa)| < 1,\,\,\kappa > 0,$$
and thus converges as $T \rightarrow \infty$. Hence $p_t = \mathcal{O}(1)\exp(-\kappa t)$, which gives
    \begin{align*}
        \frac{1}{m}\sum_{t=1}^T\left(\frac{e_t}{p_t} - T\overline{e}\right)^2 p_t & = \frac{1}{m}\sum_{t=1}^T\left(\mathcal{O}(1)\exp(\kappa t)\mathcal{O}(T^{-3/2}) - T\mathcal{O}(T^{-3/2})\right)^2 \mathcal{O}(1)\exp(-\kappa t) \\
        & = \frac{1}{m} \sum_{t=1}^T\left(\exp(2\kappa t)\mathcal{O}(T^{-3}) - \exp(\kappa t) \mathcal{O}(T^{-2}) + \mathcal{O}(T^{-1})\right) \exp(-\kappa t) \\
        & = \frac{1}{m} \left( \mathcal{O}(T^{-3})\sum_{t=1}^T \exp(\kappa t) - \mathcal{O}(T^{-2})\sum_{t=1}^T 1 + \mathcal{O}(T^{-1})\sum_{t=1}^T \exp(-\kappa t)\right).
    \end{align*}
Now, the geometric sum (for a finite $T$)
$$\sum_{t=1}^T \exp(\kappa t) = \exp(\kappa 1) + \exp(\kappa 2) + \cdots + \exp(\kappa T) = \frac{\exp(\kappa)(\exp(\kappa T) - 1)}{\exp(\kappa) - 1} = \mathcal{O}(\exp(\kappa T)),$$ thus
\begin{align*}
        \frac{1}{m}\sum_{t=1}^T\left(\frac{e_t}{p_t} - T\overline{e}\right)^2 p_t & = \frac{1}{m} \left( \mathcal{O}(T^{- 3}\exp(\kappa T)) - \mathcal{O}(T^{-1}) + \mathcal{O}(T^{-1})\right) \nonumber \\
        & = \frac{1}{m} \mathcal{O}(T^{- 3}\exp(\kappa T)).
\end{align*}
\end{proof}

\begin{proof}[Proof of Lemma \ref{lem:TGARCH_expectation_indicator}]
We begin by simplifying the expectation as follows
\begin{align*}
    \mathbb{E}(\gamma_i (y_{t-i} - \mu)^2\mathbbm{1}(y_{t-i} < \mu)) 
    &= \gamma_i \mathbb{E}( (y_{t-i} - \mu)^2\mathbbm{1}(y_{t-i} - \mu < 0))\\
    &= \gamma_i \mathbb{E}( \sigma_{t-i}^2 \varepsilon_{t-i}^2\mathbbm{1}(\sigma_{t-i} \varepsilon_{t-i}< 0))\\
    &= \gamma_i \mathbb{E}(\mathbb{E}( \sigma_{t-i}^2 \varepsilon_{t-i}^2\mathbbm{1}(\sigma_{t-i} \varepsilon_{t-i}< 0) | \mathcal{F}_{t-i-1}))\\
    &= \gamma_i \mathbb{E}( \sigma_{t-i}^2 \mathbb{E}( \varepsilon_{t-i}^2\mathbbm{1}(\sigma_{t-i} \varepsilon_{t-i}< 0) | \mathcal{F}_{t-i-1})).
\end{align*}

Since $\mathcal{D}$ is symmetric and $\sigma_{t-i}$ is known given $\mathcal{F}_{t-i-1}$, the inner expectation simplifies to
$$\mathbb{E}( \varepsilon_{t-i}^2\mathbbm{1}(\sigma_{t-i} \varepsilon_{t-i}< 0) | \mathcal{F}_{t-i-1})
= \frac{1}{2}\mathbb{E}(\varepsilon_{t-i}^2 | \mathcal{F}_{t-i-1}).
$$
Therefore, 
\begin{align*}
    \mathbb{E}(\gamma_i (y_{t-i} - \mu)^2\mathbbm{1}(y_{t-i} < \mu)) 
    &= \frac{1}{2} \gamma_i \mathbb{E}(\sigma_{t-i}^2 \mathbb{E}(\varepsilon_{t-i}^2 | \mathcal{F}_{t-i-1}))\\
    &= \frac{1}{2} \gamma_i \mathbb{E}(\sigma_{t-i}^2) \mathbb{E}(\varepsilon_{t-i}^2)\\
    &= \frac{1}{2}\gamma_i\mathbb{E}(\sigma^2_t(\boldsymbol{\theta}_v)).
\end{align*}

\end{proof}

\clearpage
\section{Additional material supporting the main paper}
\label{supp:moved_main_material}

\subsection{Additional results for untruncated decaying sampling probabilities}
\label{supp:decaying_probabilities}

To further motivate the sampling scheme in the main paper, we first examine two decaying probability schemes, whose shortcomings will guide our final design. Note that under uniform sampling, the probability of sampling the most computationally expensive observation behaves as $p_T=\mathcal{O}(T^{-1})$.
\begin{itemize}
    \item Power-law decaying (PD) sampling probabilities: 
    \begin{align}\label{eq:PD_prob}
        p^{\mathrm{PD}}_t \propto t^{-\lambda}, \,\, \lambda > 1,\,\, t = 1, \dots, T,
    \end{align}
    with $p_T=\mathcal{O}(T^{-\lambda}) \ll \mathcal{O}(T^{-1})$ under this scheme (decaying much faster than uniform sampling).
    \item Exponentially decaying (ED) sampling probabilities:
    \begin{align}\label{eq:ED_prob}
     p^{\mathrm{ED}}_t \propto \exp(-\kappa t), \,\, \kappa > 0, \,\, t = 1, \dots, T,
    \end{align}
    with $p_T=\mathcal{O}(\exp(-\kappa T))\ll \mathcal{O}(T^{-1})$ under this scheme.
\end{itemize}
Note that we require $\kappa>0$ or $\lambda > 1$ to ensure that corresponding $p_T$ probabilities decay faster to zero than under the simple random sampling scheme (which is $\mathcal{O}(T^{-1})$). Moreover, for any $\kappa > 0$ and $\lambda > 1$, $\mathcal{O}(\exp(-\kappa T)) \ll \mathcal{O}(T^{-\lambda})$. The next lemma shows the asymptotic behaviour for large $T$ of the variance of the estimators under the different sampling schemes.
\begin{lemma}\label{lem:asymtotic_variance_estimators}
Suppose that $e_t = \mathcal{O}(T^{-3/2})$. Then
\begin{enumerate}
    \item[(i)] \citep{quiroz2019speeding}. For the difference estimator in \eqref{eq:diff_estimator}, with uniform sampling,
    $$\mathbb{V}\left(\widehat{\ell}_{\mathrm{DE}}(\boldsymbol{\phi})\right) = \frac{1}{m}\mathcal{O}(T^{-1}).$$
    \item[(ii)]
    For the weighted difference estimator in \eqref{eq:weighted_diff_estimator},  
    $$\mathbb{V}\left(\widehat{\ell}_{\mathrm{WDE}}(\boldsymbol{\phi})\right) = \begin{cases}
        \frac{1}{m}\mathcal{O}\left(T^{\lambda - 2}\right), & \text{ with } p_t \text{ in } \eqref{eq:PD_prob} \\
        \frac{1}{m}\mathcal{O}\left(\exp(\kappa T) T^{-3}\right),& \text{ with } p_t \text{ in } \eqref{eq:ED_prob}.
    \end{cases}$$ 
\end{enumerate}
\end{lemma}

Lemma \ref{lem:asymtotic_variance_estimators}(i) shows that without weighting, i.e.\ $p_t = 1/T$ the variance decays to zero at the rate $1/T$ for a fixed subsample size $m$. Note that the control variates are crucial for this behaviour; without control variates, $e_t=\mathcal{O}(1)$, and the variance becomes $\mathcal{O}(T^2)$ for fixed $m$. Lemma \ref{lem:asymtotic_variance_estimators}(ii) shows that, with power-law decaying weights, the variance goes to zero only when $1 < \lambda < 2$, is bounded when $\lambda =2$, and becomes unbounded when $\lambda > 2$. With exponentially decaying weights, the variance grows exponentially for any $\kappa > 0$.

\subsection{Further properties of the stabilised sampling scheme}
\label{supp:stabilised_sampling_properties}

This section provides additional discussion of the stabilised truncated power-law decaying sampling scheme introduced in Section \ref{subsec:TPD_probs}. In particular, we discuss the approximation of the variance inflation bound by $1/c$, its empirical stability across parameter values, and the estimation of the variance from a single subsample.

To motivate the approximation by $1/c$ of the upper bound in Lemma \ref{lem:upper_bound_variance}, recall that
$$
\rho
=
\frac{e^2}{T\sum_{t=1}^T e_t^2}
=
\frac{T\overline{e}^2}{\sum_{t=1}^T e_t^2},
$$
where $\overline{e}=T^{-1}\sum_{t=1}^T e_t$. Expanding the denominator gives
$$
\sum_{t=1}^T e_t^2
=
\sum_{t=1}^T(e_t-\overline{e}+\overline{e})^2
=
\sum_{t=1}^T(e_t-\overline{e})^2
+
T\overline{e}^2.
$$
Hence,
$$
\rho
=
\frac{\overline{e}^2}{s_e^2+\overline{e}^2},
\qquad
s_e^2
=
\frac{1}{T}\sum_{t=1}^T(e_t-\overline{e})^2.
$$
When the control variate approximation errors are approximately zero-centred, $\overline{e}$ is small relative to their dispersion and hence $\rho$ is small. Consequently,
$$
\frac{\frac{1}{c}-\rho}{1-\rho}
\approx
\frac{1}{c},
$$
so that the variance inflation bound in Lemma \ref{lem:upper_bound_variance} is essentially governed by $1/c$. This motivates the interpretation of $R_{\max}=1/c_{\min}$ as an approximate upper bound on the allowable variance inflation relative to uniform sampling.

Empirically, we also find that the variance ratio is largely insensitive to the values of $\boldsymbol{\phi}$ explored by the inference algorithm. Both the weighted and uniform estimators depend on the same control variate residuals $e_t(\boldsymbol{\phi})$ and differ only through the sampling weights $p_t$, which do not depend on $\boldsymbol{\phi}$. This common dependence leads to substantial cancellation in the variance ratio and explains its approximate stability across $\boldsymbol{\phi}$. The empirical behaviour of the variance ratio is examined further in the applications.

Some applications additionally require the variance of the log-likelihood estimator to be estimated from the subsample itself, as in the subsampling Markov chain Monte Carlo approach in Section \ref{sec:MCMC_application}. An unbiased estimator of $\mathbb{V}(\widehat{\ell}_{\mathrm{WDE}})$ is obtained by using the TPD sampling probabilities in the sample-variance analogue of Lemma \ref{lem:expectation_variance_WDE}(ii). However, the TPD scheme is designed to reduce computational cost by concentrating sampling probability on the head rather than to ensure frequent sampling of observations in the tail. Infrequently sampled tail observations may contribute disproportionately to the estimator variance. Consequently, a variance estimate based on a single subsample $u_1,\dots,u_m$ may substantially underestimate the true variance, despite being unbiased over repeated subsampling; see Figure \ref{fig:bulk_empirical1}.

\subsection{Asymptotic interpretation and heuristic justification for the scaling $\varepsilon=\mathcal{O}(T^{-1})$}
\label{supp:heuristic_scaling}

We first clarify the asymptotic interpretation of the computational-cost bounds in Section \ref{subsec:computational_cost}. When the tail probability floor is attained, as in the tuning procedure of Section \ref{subsec:balancing_variance_and_computational_cost}, we have
$\varepsilon_T=c(T-t^\star)/T$. For fixed $c$, $m$, and $t^\star$, Lemma \ref{lem:average_cost}(ii) therefore gives
\begin{align*}
    \mathbb{E}(u_{\max})
    &\leq
    t^\star+(T-t^\star)
    \left(1-\left(1-c\frac{T-t^\star}{T}\right)^m\right)\\
    &=
    \left(1-(1-c)^m\right)T+\mathcal{O}(1).
\end{align*}
Thus, both truncated decaying and uniform sampling admit $\mathcal{O}(T)$ upper bounds on expected computational cost. Uniform sampling, corresponding to $c=1$, has the largest leading constant in this bound, whereas any $c<1$ yields a strictly smaller constant.

These fixed-$c$ asymptotics should be interpreted as a sequence of sampling schemes $p_{t,T}(c)$ indexed by $T$. The induced tail mass $\varepsilon_T$ and decay rate $\gamma_T$ may therefore vary with $T$, while the sampling probabilities remain normalised for each $T$. The variance asymptotics in Section \ref{subsec:TPD_probs} are interpreted analogously, with $c$ held fixed while the induced sampling scheme varies with $T$.

We next consider a different asymptotic regime, in which the objective is to keep the expected computational cost bounded as $T$ grows. To motivate the required scaling of $\varepsilon$, recall that $u_{\max}$ takes values in $\mathcal{H}\cup\mathcal{T}$, where $\mathcal{H}$ and $\mathcal{T}$ are the head and tail defined in Section \ref{subsec:TPD_probs}. By the law of total expectation,
\begin{align}\label{eq:Eumax_total_expectation}
    \mathbb{E}(u_{\max})
    &=
    \mathbb{E}(u_{\max}|u_{\max}\in\mathcal{H})
    \Pr(u_{\max}\in\mathcal{H}) \nonumber\\
    &\quad+
    \mathbb{E}(u_{\max}|u_{\max}\in\mathcal{T})
    \Pr(u_{\max}\in\mathcal{T}).
\end{align}
Since the subsampling indices are independent and each has probability $\varepsilon$ of belonging to the tail,
\[
    \Pr(u_{\max}\in\mathcal{T})
    =
    1-(1-\varepsilon)^m.
\]
For small $\varepsilon$ and fixed $m$, a Maclaurin expansion gives
\[
    \Pr(u_{\max}\in\mathcal{T})
    =
    m\varepsilon+\mathcal{O}(\varepsilon^2).
\]
Hence, the tail contribution to \eqref{eq:Eumax_total_expectation} is heuristically of order $\mathcal{O}(T)m\varepsilon$. More precisely, since $u_{\max}\leq T$,
\begin{align}\label{eq:heuristic_conditional_expectation}
    &\mathbb{E}(u_{\max}|u_{\max}\in\mathcal{T})
    \Pr(u_{\max}\in\mathcal{T}) \nonumber\\
    &\qquad\leq
    T\left(1-(1-\varepsilon)^m\right)
    \leq Tm\varepsilon.
\end{align}
The head contribution is bounded by $t^\star$, since $u_{\max}\leq t^\star$ whenever $u_{\max}\in\mathcal{H}$. Consequently,
\[
    \mathbb{E}(u_{\max})
    \leq t^\star+Tm\varepsilon.
\]
For fixed $m$ and $t^\star$, the scaling
$\varepsilon=\mathcal{O}(T^{-1})$ is therefore sufficient to keep the expected computational cost bounded as $T$ grows. This regime differs from the fixed-$c$ asymptotics considered above. When the tail floor is attained, achieving $\varepsilon_T=\mathcal{O}(T^{-1})$ requires $c$ to decrease with $T$, weakening the variance safeguard. This highlights the caveat discussed in Section \ref{subsec:computational_cost}.

Figure \ref{fig:E_umax_TPD} illustrates the expected computational cost for different values of $m$, $\varepsilon$, and $T$. The $\varepsilon$ values are obtained from different tail floor fractions $c$ by determining $\gamma_{\max}(c)$ as described in Section \ref{subsec:balancing_variance_and_computational_cost}. The figure shows that truncated decaying sampling can substantially reduce expected computational cost relative to uniform sampling, particularly for small $\varepsilon$. It also illustrates the importance of adjusting the tail mass as the sample size increases. In particular, curves with comparable normalised computational costs in the two panels use smaller values of $\varepsilon$ for the larger dataset. For the examples shown in the figure, a tenfold increase in $T$ is accompanied by an approximately tenfold reduction in $\varepsilon$. This is consistent with the scaling $\varepsilon=\mathcal{O}(T^{-1})$ discussed above, which keeps the expected computational cost bounded for fixed $m$ and $t^\star$.

\begin{figure}[ht]
    \centering
    \includegraphics[width=0.48\textwidth]{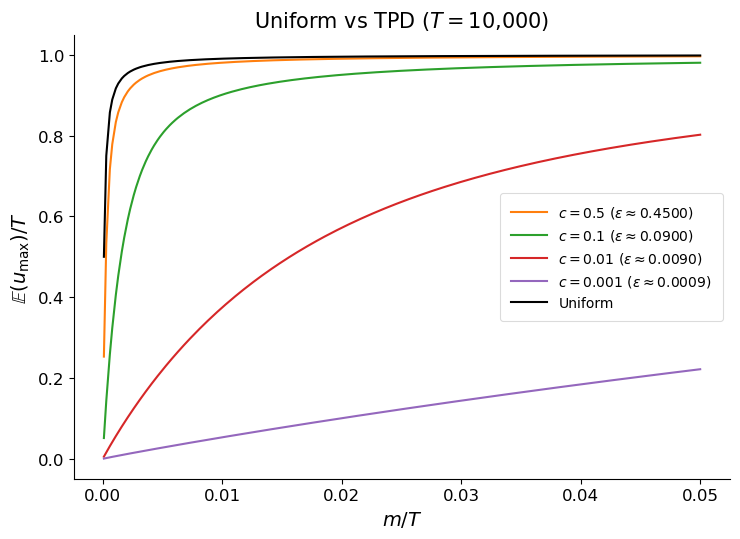}
    \hfill
    \includegraphics[width=0.48\textwidth]{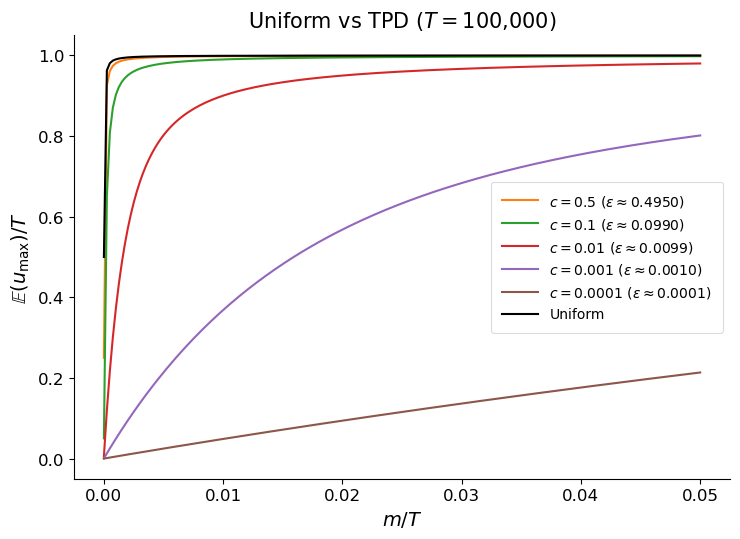}
    \caption{$\mathbb{E}(u_{\max})/T$ as a function of $m/T$, computed using Lemma \ref{lem:average_cost}(i), for $T=10{,}000$ (left panel) and $T=100{,}000$ (right panel), with $t^\star=1{,}000$ as in the applications. The curves correspond to uniform sampling and truncated decaying sampling probabilities with different tail masses $\varepsilon$; see the legend.}
    \label{fig:E_umax_TPD}
\end{figure}

\newpage

\section{Derivations, assumption verification, and additional results}\label{supp:derivations}

\subsection{Gradients and Hessians for the GARCH process}

We now derive the gradient and Hessian for the recursive process in \eqref{eq:GARCH} using matrix differential calculus \citep{magnus2019matrix}. The gradient $\nabla_{\boldsymbol{\theta}_v}\sigma^2_t(\boldsymbol{\theta}_v)$ relates to the differential $d\sigma_t^2(\boldsymbol{\theta}_v)$ via the so-called first identification theorem,
\begin{align}\label{eq:first_identification_theorem_grad}
    d\sigma^2_t(\boldsymbol{\theta}_v) & = \nabla_{\boldsymbol{\theta}_v}^\top \sigma_t^2(\boldsymbol{\theta}_v) d\boldsymbol{\theta}_v,
\end{align}
where $d\boldsymbol{\theta}_v = (d\theta_{v1}, \dots d \theta_{vd})^\top$ is the differential vector. Moreover, it can be shown that the differential of the gradient $\nabla_{\boldsymbol{\theta}_v}\sigma^2_t(\boldsymbol{\theta}_v)$ is the Hessian, i.e.\
\begin{align}\label{eq:first_identification_theorem_Hess}    d\nabla_{\boldsymbol{\theta}_v}\sigma^2_t(\boldsymbol{\theta}_v) & = \nabla_{\boldsymbol{\theta}_v}^2 \sigma_t^2(\boldsymbol{\theta}_v) d\boldsymbol{\theta}_v.
\end{align}
Rewrite \eqref{eq:GARCH} in vector form,
\begin{align}\label{eq:GARCH_vectorised}
    \sigma^2_t(\boldsymbol{\theta}_v) & = \omega + {\mathbf{z}^{2^\top}_{t-1:t-p}}\boldsymbol{\alpha} +  {\boldsymbol{\sigma}_{t-1:t-q}^{2^\top}} \boldsymbol{\beta},
\end{align}
where
$$\mathbf{z}_{t-1:t-p} = \mathbf{y}_{t-1:t-p} -\mathbf{1}_p \mu,$$
where $\mathbf{1}_p$ denotes a unit column vector of length $p$ and for a vector $\mathbf{x}_{a:b}$, with $a>b$, $(x_{a}, x_{a-1}, \dots, x_{b})^\top$. Then,
\begin{align}\label{eq:GARCH_differential}
    d\sigma^2_t(\boldsymbol{\theta}_v) & = d\omega + d\left({\mathbf{z}^{2^\top}_{t-1:t-p}}\boldsymbol{\alpha}\right) +  d\left({\boldsymbol{\sigma}_{t-1:t-q}^{2^\top}} \boldsymbol{\beta}\right) \nonumber \\
    & = d\omega + (d\mathbf{z}^{2}_{t-1:t-p})^\top\boldsymbol{\alpha} + {\mathbf{z}^{2^\top}_{t-1:t-p}}d\boldsymbol{\alpha} + \left(d{\boldsymbol{\sigma}_{t-1:t-q}^{2}}\right)^\top \boldsymbol{\beta}  + {\boldsymbol{\sigma}_{t-1:t-q}^{2^\top}} d\boldsymbol{\beta} \nonumber \\
    & = d\omega + \boldsymbol{\alpha}^\top d\mathbf{z}^{2}_{t-1:t-p} + {\mathbf{z}^{2^\top}_{t-1:t-p}}d\boldsymbol{\alpha} + \boldsymbol{\beta}^\top d{\boldsymbol{\sigma}_{t-1:t-q}^{2}}  + {\boldsymbol{\sigma}_{t-1:t-q}^{2^\top}} d\boldsymbol{\beta} \nonumber \\
    & = d\omega - \boldsymbol{\alpha}^\top 2\mathbf{z}_{t-1:t-p}d\mu
     + {\mathbf{z}^{2^\top}_{t-1:t-p}}d\boldsymbol{\alpha} + \boldsymbol{\beta}^\top d{\boldsymbol{\sigma}_{t-1:t-q}^{2}}  + {\boldsymbol{\sigma}_{t-1:t-q}^{2^\top}} d\boldsymbol{\beta} \nonumber \\
     & =\begin{pmatrix} -\boldsymbol{\alpha}^\top 2\mathbf{z}_{t-1:t-p} & 1 & \mathbf{z}^{2^\top}_{t-1:t-p} & {\boldsymbol{\sigma}_{t-1:t-q}^{2^\top}} \end{pmatrix} d\boldsymbol{\theta}_v +    \boldsymbol{\beta}^\top d{\boldsymbol{\sigma}_{t-1:t-q}^{2}} \nonumber\\
     & =\begin{pmatrix} -\boldsymbol{\alpha}^\top 2\mathbf{z}_{t-1:t-p} & 1 & \mathbf{z}^{2^\top}_{t-1:t-p} & {\boldsymbol{\sigma}_{t-1:t-q}^{2^\top}} \end{pmatrix} d\boldsymbol{\theta}_v +    \sum_{j=1}^q  \beta_j d\sigma^2_{t-j}(\boldsymbol{\theta}_v).
\end{align}
Using the first identification theorem, the recursive components $d \sigma^2_{t-j}(\boldsymbol{\theta}_v)$, $j=1,\dots,q$, can be written as 
$$d \sigma^2_{t-j}(\boldsymbol{\theta}_v) = \nabla^\top_{\boldsymbol{\theta}_v}\sigma^2_{t-j}(\boldsymbol{\theta}_v)d\boldsymbol{\theta}_v.$$
It then follows from \eqref{eq:GARCH_differential} and the first identification theorem that the gradient of \eqref{eq:GARCH} can be recursively evaluated as
\begin{align}\label{eq:GARCH_gradient}
\nabla_{\boldsymbol{\theta}_v} \sigma_t^2(\boldsymbol{\theta}_v) = \begin{pmatrix} -2\boldsymbol{\alpha}^\top\mathbf{z}_{t-1:t-p} \\ 1 \\ \mathbf{z}^{2}_{t-1:t-p} \\ {\boldsymbol{\sigma}_{t-1:t-q}^{2}} \end{pmatrix} +  \sum_{j=1}^q  \beta_j \nabla_{\boldsymbol{\theta}_v} \sigma_{t-j }^2(\boldsymbol{\theta}_v),
\end{align}
where, when $t=1$, $\nabla_{\boldsymbol{\theta}_v}\sigma_{1-j}^2(\boldsymbol{\theta}_v)=\mathbf{0}_d$ for $j=1,\dots,q$.

For the Hessian, from \eqref{eq:first_identification_theorem_Hess} we need to compute the differential of \eqref{eq:GARCH_gradient}. The differential of the first three terms of the first term in \eqref{eq:GARCH_gradient} can be written as a linear function $d\boldsymbol{\theta}_v$, whereas the last one is recursively defined. We include the terms independent of $d\boldsymbol{\theta}_v$ in a matrix $\mathbf{A} \in \mathbb{R}^{d\times d}$, and those that depend on $d\boldsymbol{\theta}_v$ in a matrix $\mathbf{B} \in \mathbb{R}^{d\times d}$. The differential of the gradient, i.e.\ the Hessian, has the structure
\begin{align}\label{eq:Hess_GARCH_struct}
    d \nabla_{\boldsymbol{\theta}_v} \sigma_t^2(\boldsymbol{\theta}_v) = \mathbf{A} d\boldsymbol{\theta}_v + \mathbf{B}d\boldsymbol{\theta}_v + \sum_{j = 1}^{q} d\left(\beta_j \nabla_{\boldsymbol{\theta}_v}\sigma^2_{t-j}(\boldsymbol{\theta}_v)\right).
\end{align}
Now,
\begin{align*}
    d(-2\boldsymbol{\alpha}^\top\mathbf{z}_{t-1:t-p}) & = -2 \mathbf{z}^\top_{t-1:t-p}d\boldsymbol{\alpha} + 2\boldsymbol{\alpha}^\top\mathbf{1}_p d\mu \\
    & =  \begin{pmatrix} 2\boldsymbol{\alpha}^\top\mathbf{1}_p & 0 & -2\mathbf{z}^{\top}_{t-1:t-p} & \mathbf{0}_q^\top\end{pmatrix} d\boldsymbol{\theta}_v,
\end{align*}
\begin{align*}
    d1  & =  \begin{pmatrix} 0 & 0 & \mathbf{0}_p^\top & \mathbf{0}_q^\top\end{pmatrix} d\boldsymbol{\theta}_v,
\end{align*}
\begin{align*}
    d\mathbf{z}^{2}_{t-1:t-p}  & =  \begin{pmatrix} -2 \mathbf{z}_{t-1:t-p} & \mathbf{0}_p & \mathbf{0}_{p\times p} & \mathbf{0}_{p \times q} \end{pmatrix} d\boldsymbol{\theta}_v,
\end{align*}
and thus
\begin{align}
    \mathbf{A} =  \begin{pmatrix} 2\boldsymbol{\alpha}^\top\mathbf{1}_p & 0 & -2\mathbf{z}^{\top}_{t-1:t-p} & \mathbf{0}_q^\top \\
    0 & 0 & \mathbf{0}_p^\top & \mathbf{0}_q^\top \\
    -2 \mathbf{z}_{t-1:t-p} & \mathbf{0}_p & \mathbf{0}_{p\times p} & \mathbf{0}_{p \times q} \\
    \mathbf{0}_{q} & \mathbf{0}_{q} & \mathbf{0}_{q\times p} & \mathbf{0}_{q \times q}\end{pmatrix}.
\end{align}
The terms that are recursively defined correspond to the differential of the vector $\boldsymbol{\sigma}^2_{t-1:t-q}$. Since $d\sigma^2_{t-1}(\boldsymbol{\theta}_v) = \nabla^\top_{\boldsymbol{\theta}_v} \sigma^2 (\boldsymbol{\theta}_v)d\boldsymbol{\theta}_v$, 
\begin{align}
    \mathbf{B} =  \begin{pmatrix} \mathbf{0}^\top_d\\\mathbf{0}^\top_d \\ \mathbf{0}_{p\times d} \\ \nabla^\top_{\boldsymbol{\theta}_v}\sigma^2_{t-1}(\boldsymbol{\theta}_v) \\
    \nabla^\top_{\boldsymbol{\theta}_v}\sigma^2_{t-2}(\boldsymbol{\theta}_v) \\
    \vdots\\
    \nabla^\top_{\boldsymbol{\theta}_v}\sigma^2_{t-q}(\boldsymbol{\theta}_v) 
    \end{pmatrix}.
\end{align}
Finally, to compute the last term in \eqref{eq:Hess_GARCH_struct}, we use that $\beta_j = \mathbf{e}^\top_{k}\boldsymbol{\theta}_v$, so that $d\beta_j = \mathbf{e}^\top_{k} d\boldsymbol{\theta}_v$, where $\mathbf{e}_k$ is the $k=(2+p+j)$th basis vector and that $$(\mathbf{b}^\top\mathbf{x})\mathbf{a} = (\mathbf{a} \otimes \mathbf{b}^\top)\mathbf{x}.$$
With this,
\begin{align*}
d(\beta_j \nabla_{\boldsymbol{\theta}_v}\sigma^2_{t-j}(\boldsymbol{\theta}_v)) & = (d\beta_j)\nabla_{\boldsymbol{\theta}_v}\sigma^2_{t-j}(\boldsymbol{\theta}_v) + \beta_j d \nabla_{\boldsymbol{\theta}_v}\sigma^2_{t-j}(\boldsymbol{\theta}_v) \\
& =  (\mathbf{e}^\top_{k} d\boldsymbol{\theta}_v)\nabla_{\boldsymbol{\theta}_v}\sigma^2_{t-j}(\boldsymbol{\theta}_v) + \beta_j \nabla^2_{\boldsymbol{\theta}_v}\sigma^2_{t-j}(\boldsymbol{\theta}_v)d\boldsymbol{\theta}_v \\
& = \left(\nabla_{\boldsymbol{\theta}_v}\sigma^2_{t-j}(\boldsymbol{\theta}_v) \otimes \mathbf{e}^\top_k\right)d\boldsymbol{\theta}_v + \beta_j \nabla^2_{\boldsymbol{\theta}_v}\sigma^2_{t-j}(\boldsymbol{\theta}_v)d\boldsymbol{\theta}_v .
\end{align*}
The final recursive expression for the Hessian of \eqref{eq:GARCH} is
\begin{align}\label{eq:Hess_GARCH}
     \nabla^2_{\boldsymbol{\theta}_v} \sigma_t^2(\boldsymbol{\theta}_v) = \left(\mathbf{A} + \sum_{j=1}^q \nabla_{\boldsymbol{\theta}_v}\sigma^2_{t-j}(\boldsymbol{\theta}_v) \otimes \mathbf{e}^\top_k \right) + \mathbf{B} + \sum_{j = 1}^q \beta_j \nabla^2_{\boldsymbol{\theta}_v}\sigma^2_{t-j}(\boldsymbol{\theta}_v),
\end{align}
where, when $t=1$, $\nabla^2_{\boldsymbol{\theta}_v}\sigma_{1-j}^2(\boldsymbol{\theta}_v)=\mathbf{0}_{d\times d}$ for $j=1,\dots,q$.

The gradient in \eqref{eq:GARCH_gradient} and Hessian in \eqref{eq:Hess_GARCH} appear in the gradient and Hessian of the log-likelihood, see Appendix \ref{subsec:grad_Hess_log_dens} and Section \ref{subsec:likelihood} for details.

\subsection{Gradients and Hessians for the TGARCH process}\label{supp:TGARCH_gradient}

To derive the gradient and Hessian for the TGARCH model, we replace the indicator $\mathbbm{1}(y_{t-i} < \mu)$ in \eqref{eq:TGARCH} with the Heaviside step function. Specifically, let $H(x)$ denote the Heaviside function, defined by 
\begin{equation*}
H(x) =
\begin{cases}
1, & x > 0,\\
0, & \text{otherwise}.
\end{cases}
\end{equation*}
The indicator can reformulated as $\mathbbm{1}(\mu - y_{t-i} > 0)$ and subsequently as $H(\mu - y_{t-i})$. Then, we rewrite \eqref{eq:TGARCH} in vector form as,
\begin{align}\label{eq:TGARCH_vectorised}
    \sigma^2_t(\boldsymbol{\theta}_v) & = \omega + {\mathbf{z}^{2^\top}_{t-1:t-p}}\boldsymbol{\alpha} + \boldsymbol{\gamma}^{\top} \mathrm{diag}({\mathbf{z}^{2}_{t-1:t-p}})\mathbf{G}(\mu) +  {\boldsymbol{\sigma}_{t-1:t-q}^{2^\top}} \boldsymbol{\beta},
\end{align}
where
\begin{align*}
    \mathbf{G}(\mu) = (H(\mu - y_{t-1}), \dots, H(\mu - y_{t-p}))^{\top}.
\end{align*}
At this point, to compute the gradient and Hessian, it is also necessary to define
\begin{align*}
    d H(\mu - y_{t-i}) = \delta(\mu - y_{t-i})d\mu \quad \text{and} \quad d \delta(\mu - y_{t-i}) = \delta'(\mu - y_{t-i})d\mu,
\end{align*}
where $\delta(\cdot)$ denotes the Dirac delta distribution and $\delta'(\cdot)$ represents its first derivative with respect to $\mu$. Hence, we can define,
\begin{align*}
    d\mathbf{G}(\mu) = \mathbf{G}'(\mu)d\mu \quad \text{and} \quad d\mathbf{G}'(\mu) = \mathbf{G}''(\mu)d\mu,
\end{align*}
where
\begin{align*}
    \mathbf{G}'(\mu) &= (\delta(\mu - y_{t-1}), \dots, \delta(\mu - y_{t-p}))^{\top},\\
    \mathbf{G}''(\mu) &= (\delta'(\mu - y_{t-1}), \dots, \delta'(\mu - y_{t-p}))^{\top}.
\end{align*}

We note that the only difference from the conditional variance formulation in the GARCH model arises from the term involving $\boldsymbol{\gamma}$ in \eqref{eq:TGARCH_vectorised}. Hence, when computing the differential $d\sigma^2_t(\boldsymbol{\theta}_v)$, it is sufficient to focus on
\begin{align*}
    d(\boldsymbol{\gamma}^{\top} \mathrm{diag}({\mathbf{z}^{2}_{t-1:t-p}})\mathbf{G}(\mu))
\end{align*}
and add it to the differential already obtained in the GARCH case. Proceeding with this step, we obtain,
\begin{align}\label{eq:TGARCH_differential_gamma_component}
   &d(\boldsymbol{\gamma}^{\top} \mathrm{diag}({\mathbf{z}^{2}_{t-1:t-p}})\mathbf{G}(\mu)) \nonumber \\
   & =  d(\boldsymbol{\gamma}^{\top}) \mathrm{diag}({\mathbf{z}^{2}_{t-1:t-p}})\mathbf{G}(\mu) + \boldsymbol{\gamma}^{\top} d(\mathrm{diag}({\mathbf{z}^{2}_{t-1:t-p}}))\mathbf{G}(\mu) +  \boldsymbol{\gamma}^{\top} \mathrm{diag}({\mathbf{z}^{2}_{t-1:t-p}})d\mathbf{G}(\mu) \nonumber \\
   &= (\mathrm{diag}(\mathbf{z}^{2}_{t-1:t-p})\mathbf{G}(\mu))^{\top}d \boldsymbol{\gamma} -2 \boldsymbol{\gamma}^{\top} d(\mathrm{diag}(\mathbf{z}_{t-1:t-p}d\mu))\mathbf{G}(\mu) +  \boldsymbol{\gamma}^{\top} \mathrm{diag}({\mathbf{z}^{2}_{t-1:t-p}})\mathbf{G}'(\mu)d\mu \nonumber\\
   &= \mathbf{G}(\mu)^{\top}\mathrm{diag}(\mathbf{z}^{2}_{t-1:t-p})d\boldsymbol{\gamma} - \boldsymbol{\gamma}^{\top}\{2\mathrm{diag}(\mathbf{z}_{t-1:t-p})\mathbf{G}(\mu) - \mathrm{diag}(\mathbf{z}^{2}_{t-1:t-p})\mathbf{G}'(\mu)\}d\mu.
\end{align}
Utilising \eqref{eq:GARCH_gradient}, \eqref{eq:TGARCH_differential_gamma_component} and the first identification theorem allows for the gradient of \eqref{eq:TGARCH} to be recursively computed as
\begin{align}\label{eq:TGARCH_gradient}
\nabla_{\boldsymbol{\theta}_v} \sigma_t^2(\boldsymbol{\theta}_v) &= \begin{pmatrix} -2\boldsymbol{\alpha}^\top\mathbf{z}_{t-1:t-p} - \boldsymbol{\gamma}^{\top}\{2\mathrm{diag}(\mathbf{z}_{t-1:t-p})\mathbf{G}(\mu) - \mathrm{diag}(\mathbf{z}^{2}_{t-1:t-p})\mathbf{G}'(\mu)\}\\ 1 \\ \mathbf{z}^{2}_{t-1:t-p} \\ 
\mathrm{diag}(\mathbf{z}^{2}_{t-1:t-p})\mathbf{G}(\mu)\\
{\boldsymbol{\sigma}_{t-1:t-q}^{2}} \end{pmatrix} \nonumber\\
&\quad+  \sum_{j=1}^q  \beta_j \nabla_{\boldsymbol{\theta}_v} \sigma_{t-j }^2(\boldsymbol{\theta}_v),
\end{align}
where, when $t=1$, $\nabla_{\boldsymbol{\theta}_v}\sigma_{1-j}^2(\boldsymbol{\theta}_v)=\mathbf{0}_d$ for $j=1,\dots,q$.

In a similar manner to the GARCH model, the Hessian retains the structure specified in \eqref{eq:Hess_GARCH_struct} and can be obtained by computing the differential of \eqref{eq:TGARCH_gradient}. Note 
\begin{align*}
    &d(-2\boldsymbol{\alpha}^\top\mathbf{z}_{t-1:t-p} - \boldsymbol{\gamma}^{\top}\{2\mathrm{diag}(\mathbf{z}_{t-1:t-p})\mathbf{G}(\mu) - \mathrm{diag}(\mathbf{z}^{2}_{t-1:t-p})\mathbf{G}'(\mu)\})\\
    & = -2 \mathbf{z}^\top_{t-1:t-p}d\boldsymbol{\alpha} + 2\boldsymbol{\alpha}^\top\mathbf{1}_p d\mu - \{2\mathrm{diag}(\mathbf{z}_{t-1:t-p})\mathbf{G}(\mu) + \mathrm{diag}(\mathbf{z}^{2}_{t-1:t-p})\mathbf{G}'(\mu)\}^\top d\boldsymbol{\gamma} \\
    &\quad + \boldsymbol{\gamma}^\top \{2\mathbf{G}(\mu) - 4\mathrm{diag}(\mathbf{z}_{t-1:t-p})\mathbf{G}'(\mu) + \mathrm{diag}(\mathbf{z}^{2}_{t-1:t-p})\mathbf{G}''(\mu)\} d\mu \\
    & =  \begin{pmatrix} 2\mathbf{1}_p^\top\boldsymbol{\alpha} +  \{2\mathbf{G}(\mu) - 4\mathrm{diag}(\mathbf{z}_{t-1:t-p})\mathbf{G}'(\mu) + \mathrm{diag}(\mathbf{z}^{2}_{t-1:t-p})\mathbf{G}''(\mu)\}^\top \boldsymbol{\gamma}\\ 0 \\ -2\mathbf{z}_{t-1:t-p} \\ -\{2\mathrm{diag}(\mathbf{z}_{t-1:t-p})\mathbf{G}(\mu) + \mathrm{diag}(\mathbf{z}^{2}_{t-1:t-p})\mathbf{G}'(\mu)\} \\
    \mathbf{0}_q\end{pmatrix}^\top d\boldsymbol{\theta}_v,
\end{align*}
\begin{align*}
    d1  & =  \begin{pmatrix} 0 & 0 & \mathbf{0}_p^\top & \mathbf{0}_p^\top &\mathbf{0}_q^\top\end{pmatrix} d\boldsymbol{\theta}_v,
\end{align*}
\begin{align*}
    d\mathbf{z}^{2}_{t-1:t-p}  & =  \begin{pmatrix} -2 \mathbf{z}_{t-1:t-p} & \mathbf{0}_p & \mathbf{0}_{p\times p} & \mathbf{0}_{p\times p} & \mathbf{0}_{p \times q} \end{pmatrix} d\boldsymbol{\theta}_v,
\end{align*}
and
\begin{align*}
&d(\mathrm{diag}(\mathbf{z}^{2}_{t-1:t-p})\mathbf{G}(\mu))\\ &= -2\mathrm{diag}(\mathbf{z}_{t-1:t-p})\mathbf{G}(\mu)d\mu + \mathrm{diag}(\mathbf{z}^{2}_{t-1:t-p})\mathbf{G}'(\mu)d\mu\\
& =  \begin{pmatrix}  -2\mathrm{diag}(\mathbf{z}_{t-1:t-p})\mathbf{G}(\mu) + \mathrm{diag}(\mathbf{z}^{2}_{t-1:t-p})\mathbf{G}'(\mu)& \mathbf{0}_p & \mathbf{0}_{p\times p} & \mathbf{0}_{p\times p} & \mathbf{0}_{p \times q} \end{pmatrix} d\boldsymbol{\theta}_v.
\end{align*}
Thus for the TGARCH model, the matrix $\mathbf{A}$, which collects the terms independent of $d\boldsymbol{\theta}_v$, is
\begin{align}\label{eq:TGARCH_A}
    \mathbf{A} =  \begin{pmatrix} \mathbf{A}_{\mu\mu} & 0 & -2\mathbf{z}^{\top}_{t-1:t-p} & \mathbf{A}_{\mu\boldsymbol{\gamma}}  & \mathbf{0}_q^\top \\
    0 & 0 & \mathbf{0}_p^\top & \mathbf{0}_p^\top & \mathbf{0}_q^\top \\
    -2 \mathbf{z}_{t-1:t-p} & \mathbf{0}_p & \mathbf{0}_{p\times p} & \mathbf{0}_{p\times p} &  \mathbf{0}_{p \times q} \\ \mathbf{A}_{\boldsymbol{\gamma}\mu}
    & \mathbf{0}_p & \mathbf{0}_{p\times p} & \mathbf{0}_{p\times p} &  \mathbf{0}_{p \times q}\\ \mathbf{0}_{q} &
    \mathbf{0}_{q} & \mathbf{0}_{q\times p} & \mathbf{0}_{q\times p} & \mathbf{0}_{q \times q}\end{pmatrix},
\end{align}
where
\begin{align*}
    \mathbf{A}_{\mu\mu} = 2\boldsymbol{\alpha}^\top\mathbf{1}_p + \boldsymbol{\gamma}^\top \{2\mathbf{G}(\mu) - 4\mathrm{diag}(\mathbf{z}_{t-1:t-p})\mathbf{G}'(\mu) + \mathrm{diag}(\mathbf{z}^{2}_{t-1:t-p})\mathbf{G}''(\mu)\},
\end{align*}
\begin{align*}
     \mathbf{A}_{\mu\boldsymbol{\gamma}} = (\mathbf{A}_{\boldsymbol{\gamma}\mu})^\top= -\{2\mathrm{diag}(\mathbf{z}_{t-1:t-p})\mathbf{G}(\mu) + \mathrm{diag}(\mathbf{z}^{2}_{t-1:t-p})\mathbf{G}'(\mu)\}^\top.
\end{align*}
The terms involving $d\boldsymbol{\theta}_v$ are collected in the matrix $\mathbf{B}$ as
\begin{align}\label{eq:TGARCH_B}
    \mathbf{B} =  \begin{pmatrix} \mathbf{0}^\top_d\\\mathbf{0}^\top_d \\ \mathbf{0}_{p\times d} \\ 
    \mathbf{0}_{p\times d} \\\nabla^\top_{\boldsymbol{\theta}_v}\sigma^2_{t-1}(\boldsymbol{\theta}_v) \\
    \nabla^\top_{\boldsymbol{\theta}_v}\sigma^2_{t-2}(\boldsymbol{\theta}_v) \\
    \vdots\\
    \nabla^\top_{\boldsymbol{\theta}_v}\sigma^2_{t-q}(\boldsymbol{\theta}_v) 
    \end{pmatrix}.
\end{align}

Then, following steps analogous to those in the GARCH model case, the final recursive expression for the Hessian of \eqref{eq:TGARCH} is
\begin{align}\label{eq:Hess_TGARCH}
     \nabla^2_{\boldsymbol{\theta}_v} \sigma_t^2(\boldsymbol{\theta}_v) = \left(\mathbf{A} + \sum_{j=1}^q \nabla_{\boldsymbol{\theta}_v}\sigma^2_{t-j}(\boldsymbol{\theta}_v) \otimes \mathbf{e}^\top_k \right) + \mathbf{B} + \sum_{j = 1}^q \beta_j \nabla^2_{\boldsymbol{\theta}_v}\sigma^2_{t-j}(\boldsymbol{\theta}_v),
\end{align}
where, when $t=1$, $\nabla^2_{\boldsymbol{\theta}_v}\sigma_{1-j}^2(\boldsymbol{\theta}_v)=\mathbf{0}_{d\times d}$ for $j=1,\dots,q$ and $\mathbf{e}_k$ is the $k=(2+2p+j)$th basis vector.

\begin{remark}\label{rem:computing_dirac}
The Dirac delta is a distribution (generalised function) rather than an ordinary function and is therefore not directly computationally feasible. For optimisation purposes, we replace the Heaviside function with a smooth logistic approximation $\Lambda_k(x) = \{1 + \exp(-kx)\}^{-1}$, whose steepness is controlled by $k$. Note that as $k \to \infty$, $\Lambda_k(x) \to H(x)$ pointwise, while for finite $k$ we obtain stable gradients for optimisation. Hence, to use the gradient and Hessian presented in \eqref{eq:TGARCH_gradient} and \eqref{eq:Hess_TGARCH}, we simply replace $\mathbf{G}(\mu)$, $\mathbf{G}'(\mu)$ and $\mathbf{G}''(\mu)$ with 

\begin{align}
    \mathbf{G}(\mu) &= (\Lambda(\mu - y_{t-1}), \dots, \Lambda(\mu - y_{t-p}))^{\top}, \nonumber\\
    \mathbf{G}'(\mu) &= (\Lambda'(\mu - y_{t-1}), \dots, \Lambda'(\mu - y_{t-p}))^{\top}, \nonumber\\
    \mathbf{G}''(\mu) &= (\Lambda''(\mu - y_{t-1}), \dots, \Lambda''(\mu - y_{t-p}))^{\top}, \nonumber
\end{align}
where $\Lambda_k'(x) = k\Lambda_k(x)(1-\Lambda_k(x))$ and $\Lambda_k''(x) = k^2\Lambda_k(x)(1-\Lambda_k(x))(1-2\Lambda_k(x))$.
\end{remark}

\subsection{Unconditional variance of the TGARCH model under symmetry}\label{supp:TGARCH_unconditional_var}
\cite{zakoian1994threshold} derives the unconditional variance of the TGARCH for a general case. The lemma that follows simplifies the expectation component involving the indicator function when $\mathcal{D}(0, 1)$ from \eqref{eq:measurement_plus_deterministic_state_restated} is a symmetric distribution.
\begin{lemma}\label{lem:TGARCH_expectation_indicator}
    Suppose that $\mathbb{E}(\sigma^2_t(\boldsymbol{\theta}_v)) = \mathbb{E}(\sigma^2_{t-1}(\boldsymbol{\theta}_v)) = \dots =\mathbb{E}(\sigma^2_{1}(\boldsymbol{\theta}_v))$ and $\varepsilon_t \overset{\mathrm{ind}}{\sim} \mathcal{D}(0, 1)$, with $\mathcal{D}$ symmetric. Then, 

    $$
    \mathbb{E}(\gamma_i (y_{t-i} - \mu)^2\mathbbm{1}(y_{t-i} < \mu)) = \frac{1}{2}\gamma_i\mathbb{E}(\sigma^2_t(\boldsymbol{\theta}_v)).
    $$
\end{lemma}
Using Lemma \ref{lem:TGARCH_expectation_indicator}, it follows that the unconditional variance of the TGARCH process is
\begin{align}
\mathbb{E}(\sigma^2_t(\boldsymbol{\theta}_v))
= \frac{\omega}{1 - \sum_{i = 1}^p \alpha_i - \frac{1}{2} \sum_{i = 1}^p \gamma_i- \sum_{j = 1}^q \beta_j }.    
\end{align}

\subsection{Reparameterisation example GARCH}\label{supp:reparam_ex_GARCH}
This section shows how to apply Lemma \ref{lem:reparam_grad_Hess_variance} to the GARCH model under specific parameter transformations. The natural transformation for all parameters in the GARCH model except $\mu$ is $\phi=\log(\theta)\in \mathbb{R}$ so that $\theta = \exp\left(\phi\right) \in (0, \infty)$, and hence
$$\theta_i=h_i^{-1}(\phi_i) = \exp(\phi_i) \implies h^{-1\,'}_i(\phi_i) = \exp(\phi_i),\,\, h^{-1\,''}_i(\phi_i) = \exp(\phi_i),\,\, i=2,\dots, d_v.$$
Since $\mu$ is already unrestricted, $\phi_1 = \mu$, and $h_1$ and its inverse are the identity function, giving $h^{-1\,'}_i(\phi_i)=1$ and $h^{-1\,''}_i(\phi_i)=0$. Equations \eqref{eq:general_gradient_phi} and \eqref{eq:GARCH_gradient} give the gradient
\begin{align}\label{eq:GARCH_gradient_unconstrained}
\nabla_{\boldsymbol{\phi}} \sigma_t^2(\boldsymbol{\theta}) & = \mathrm{diag}\left(1,\exp(\boldsymbol{\phi}_{2:p})\right)\begin{pmatrix} -2\exp\left(\boldsymbol{\phi_{\boldsymbol{\alpha}}}\right)^\top\mathbf{z}_{t-1:t-p} \\ 1 \\ \mathbf{z}^{2}_{t-1:t-p} \\ {\boldsymbol{\sigma}_{t-1:t-q}^{2}} \end{pmatrix}  +  \sum_{j=1}^q \exp\left(\phi_{{\beta}_j} \right)  \nabla_{\boldsymbol{\phi}} \sigma_{t-j }^2(\boldsymbol{\phi}),
\end{align}
where $\boldsymbol{\phi}_{\boldsymbol{\alpha}}=(\log(\alpha_1), \dots ,\log(\alpha_p))^\top$ and $\phi_{\beta_j}=\log(\beta_j)$. The Hessian expression $\nabla_{\boldsymbol{\phi}} \sigma^2_t(\boldsymbol{\theta})$ is given by \eqref{eq:general_Hess_phi} together with \eqref{eq:Hess_GARCH}.

\subsection{Verifying the assumption in Lemma \ref{lem:taylor_remainder} for a GARCH($1,1$)}\label{supp:verify_Lemma2_ass}
We verify the assumption in Lemma \ref{lem:taylor_remainder} for a standard GARCH($1,1$) model with Gaussian measurement errors, under a stationary variance process and an unrestricted parameterisation, conditioning on a bounded realised sample $\mathbf{y}_{0:T}$. Although the GARCH($1,1$) specification is structurally simple, it remains one of the most empirically successful volatility models \citep{hansen2005GARCH}.

Specifically, we consider
\begin{align}\label{eq:model_to_verify}
    y_t & = \phi_{\mu} + z_t, \,\, z_t = \sigma_t \varepsilon_t, \,\, \varepsilon_t \overset{\mathrm{ind}}{\sim} \mathcal{N}(0, 1),  \,\, \sigma_t > 0, \nonumber \\
    \sigma^2_t(\boldsymbol{\phi}_v) & = \exp(\phi_\omega) + \exp(\phi_{\alpha_1}) (y_{t-1} - \phi_{\mu})^2 +  \exp(\phi_{\beta_1}) \sigma^2_{t-1}(\boldsymbol{\phi}_v), 
\end{align}
where $$\boldsymbol{\phi}_v = (\phi_{\mu}, \phi_\omega, \phi_{\alpha_1}, \phi_{\beta_1})^\top \in \mathbb{R}^4,$$
with the reparameterisation $$\phi_\mu=\mu, \,\, \phi_{\omega}=\log(\omega), \,\, \phi_{\alpha_1}=\log(\alpha_1), \,\, \phi_{\beta_1}=\log(\beta_1),$$
which map the original (constrained) parameters  $\mu\in \mathbb{R},\omega > 0, \alpha_1 \geq 0, \beta_1 \geq 0$ to an unrestricted parameter space. We assume that the implied GARCH parameters satisfy the stationary condition $\alpha_1  + \beta_1 < 1$.

The assumption states that there exist constants $r >0$ and $K>0$ such that $$\mathop{\max\vphantom{\sup}}\limits_{1 \le t \le T}\;
\mathop{\sup\vphantom{\max}}\limits_{\|\boldsymbol{\phi} -\boldsymbol{\phi}^\star\|\leq r} \| \nabla^3_{\boldsymbol{\phi}}\ell_t(\boldsymbol{\phi})\| \leq K.$$
This is verified in three steps: 
\begin{enumerate}
    \item Establishing the existence of a radius $r>0$ such that the ball $\{\boldsymbol{\phi}: \|\boldsymbol{\phi} - \boldsymbol{\phi}^\star \|\leq r \}$ lies strictly within the stationarity region of the GARCH parameters implied by the unrestricted parameterisation. 
    \item On the ball in Step 1, we show that the conditional variance process $\sigma^2_t(\boldsymbol{\phi}_v)$ is uniformly positive and bounded for all $t=1,\dots,T$ and all $\boldsymbol{\phi}_v$ in the ball.
    \item Using that the Gaussian log-density is three times continuously differentiable with respect to $(\phi_\mu,\sigma_t^2)$ for $\sigma_t^2>0$, together with the boundedness of $\sigma^2_t(\boldsymbol{\phi}_v)$ and its derivatives up to order three on the ball in Step 1, we conclude that there exists a constant $K$ such that $\| \nabla^3_{\boldsymbol{\phi}}\ell_t(\boldsymbol{\phi})\|$ is uniformly bounded.
\end{enumerate}

\paragraph{Step 1: Establishing the existence of $r$ such that the ball lies strictly within the stationarity region.}
Assume that $\boldsymbol{\phi}^{\star}$ satisfies the stationarity inequality, i.e.\
$$\alpha^\star_1 = \exp(\phi^\star_{\alpha_1}),\,\, \beta^\star_1 = \exp(\phi^\star_{\beta_1}),$$
such that $\alpha_1^{\star} + \beta_1^{\star} < 1$. Define the continuous mapping $$g(\boldsymbol{\phi}) = \exp(\phi_{\alpha_1}) + \exp(\phi_{\beta_1}),$$
and let
$$\epsilon = (1 - g(\boldsymbol{\phi}^\star))/2,$$
with $\epsilon > 0 $ since $g(\boldsymbol{\phi}^\star) <1$. Because $g$ is continuous at $\boldsymbol{\phi}^\star$, there exists an $r>0$ such that
$$\| \boldsymbol{\phi} - \boldsymbol{\phi}^\star\| \leq r \implies |g(\boldsymbol{\phi}) - g(\boldsymbol{\phi}^\star)| < \epsilon.$$
Thus, for all $\boldsymbol{\phi}$ satisfying $\| \boldsymbol{\phi} - \boldsymbol{\phi}^{\star}\|\leq r$, $$g(\boldsymbol{\phi}) < g(\boldsymbol{\phi}^\star) + \epsilon = 1 - \epsilon,$$  establishing the existence of a ball on which the stationarity condition holds uniformly.

\paragraph{Step 2: Showing $\sigma^2_t(\boldsymbol{\phi}_v)$ is uniformly positive and bounded on the ball.}
For uniform positivity, note that the function $\exp(\phi_{\omega})$ is continuous and strictly positive on the compact ball
$$\{\boldsymbol{\phi}: \| \boldsymbol{\phi}-\boldsymbol{\phi}^\star\| \leq r \}.$$
Hence it attains a strictly positive minimum
$$\underline{\omega} = \inf_{\| \boldsymbol{\phi} - \boldsymbol{\phi}^\star\| \leq r} \exp(\phi_\omega) > 0.$$
Since all remaining terms in the recursion are non-negative,
$$ \sigma^2_t(\boldsymbol{\phi}_v) \geq \exp(\phi_{\omega})\geq \underline{\omega}>0, \,\, t=1, \dots, T,$$
establishing uniform positivity over $\| \boldsymbol{\phi} - \boldsymbol{\phi}^\star\| \leq r$.

For uniform boundedness, because the ball is compact and the exponential function is continuous, the following upper bounds exist
\begin{align}\label{eq:upper_bounds_GARCH_params}
  \exp(\phi_\omega)  \leq \overline{\omega} &  = \sup_{\|\boldsymbol{\phi}-\boldsymbol{\phi}^\star\| \leq r}\exp(\phi_\omega)  \nonumber\\
  \exp(\phi_{\alpha_1}) \leq \overline{\alpha}_1 & =\sup_{\|\boldsymbol{\phi}-\boldsymbol{\phi}^\star\|\leq r}\exp(\phi_{\alpha_1}) \nonumber \\
  \exp(\phi_{\beta_1}) \leq \overline{\beta}_1 & =\sup_{\|\boldsymbol{\phi}-\boldsymbol{\phi}^\star\|\leq r}\exp(\phi_{\beta_1})  < 1, 
\end{align}
where the strict inequality follows from Step 1, which establishes that the stationarity condition
\begin{align}\label{eq:stationarity_cond}
    \exp(\phi_{\alpha_1}) + \exp(\phi_{\beta_1})\leq 1 - \epsilon
\end{align}
holds uniformly on the ball. Furthermore, on the compact ball, there exists
\begin{align}\label{eq:bound_mu}
  \overline{\mu} & = \sup_{\|\boldsymbol{\phi}-\boldsymbol{\phi}^\star\|\leq r} |\phi_\mu|, 
\end{align}
hence, for $t=1,\dots, T$,
$$(y_{t-1} - \phi_\mu)^2 \leq (|y_{t-1}| + |\phi_\mu|)^2 \leq (|y_{t-1}| + \overline{\mu})^2.$$
Moreover, since we condition on a bounded realised sample $\mathbf{y}_{0:T}$, there exists an $M<\infty$ such that 
\begin{align}\label{eq:bound_data}
  \max_{0\leq t \leq T}|y_{t}| \leq M  
\end{align}
and thus $(|y_{t-1}| + \overline{\mu})^2 \leq (M+\overline{\mu})^2$. Using this in the GARCH recursion in \eqref{eq:model_to_verify},
\begin{align*}
    \sigma^2_t(\boldsymbol{\phi}_v)  & = \exp(\phi_\omega) + \exp(\phi_{\alpha_1}) (y_{t-1} - \phi_{\mu})^2 +  \exp(\phi_{\beta_1}) \sigma^2_{t-1}(\boldsymbol{\phi}_v) \\
     & \leq \overline{\omega} + \overline{\alpha}_1(M + \overline{\mu})^2 + \overline{\beta}_1\sigma_{t-1}^2(\boldsymbol{\phi}_v).
\end{align*}
Iterating the recursion from any finite start value $\sigma^2_0(\boldsymbol{\phi}_v)<\infty$ yields
\begin{align}\label{eq:bound1_sigma_t}
\sigma^2_t(\boldsymbol{\phi}_v) & \leq \overline{\beta}_1^{\,t}\sigma^2_0(\boldsymbol{\phi}_v) +  (\overline{\omega} + \overline{\alpha}_1(M + \overline{\mu})^2)\sum_{k=0}^{t-1} \overline{\beta}_1^{\,k}.
\end{align}
The geometric sum 
$$\sum_{k=0}^{t-1} \overline{\beta}_1^{\,k} = \frac{1 - \overline{\beta}_1^{\,t}}{1-\overline{\beta}_1}<\frac{1}{1-\overline{\beta}_1},$$
where the inequality follows from $\overline{\beta}_1 < 1$ in \eqref{eq:upper_bounds_GARCH_params}. In particular, by \eqref{eq:stationarity_cond}, $1-\overline{\beta}_1 \geq \epsilon>0$, and thus
\begin{align}\label{eq:bound_geometric_sum}
    \sum_{k=0}^{t-1} \overline{\beta}_1^{\,k} & \leq \frac{1}{\epsilon}.
\end{align}
Then, from \eqref{eq:bound1_sigma_t}, it follows that 
\[
\mathop{\max\vphantom{\sup}}\limits_{1 \le t \le T}\;
\mathop{\sup\vphantom{\max}}\limits_{\|\boldsymbol{\phi} -\boldsymbol{\phi}^\star\|\leq r}
\sigma^2_t(\boldsymbol{\phi}_v) \leq \mathop{\vphantom{\max}\sup}\limits_{\|\boldsymbol{\phi} -\boldsymbol{\phi}^\star\|\leq r}
\sigma^2_0(\boldsymbol{\phi}_v) + (\overline{\omega} +\overline{\alpha}_1(M + \overline{\mu})^2)\frac{1}{\epsilon} =\overline{\sigma}^{\,2} < \infty.
\]
Together with the lower bound $\underline{\omega}$, this shows that $$0 < \underline{\omega} \leq \sigma^2_t(\boldsymbol{\phi}_v) \leq \overline{\sigma}^{\,2},$$
uniformly for all $t=1,\dots, T$ and all $\boldsymbol{\phi}$ such that $\|\boldsymbol{\phi} - \boldsymbol{\phi}^\star\|\leq r$.

\paragraph{Step 3: Establishing the existence of $K$ such that $\| \nabla^3_{\boldsymbol{\phi}}\ell_t(\boldsymbol{\phi})\|$ is uniformly bounded on the ball.}
For the Gaussian measurement density in \eqref{eq:model_to_verify}, the log-likelihood contribution at time $t$ is
\begin{align*}
  \ell_t(\boldsymbol{\phi}) & = -\frac{1}{2}\log(2\pi) - \frac{1}{2}\log(\sigma^2_t(\boldsymbol{\phi}_v)) - \frac{1}{2\sigma^2_t(\boldsymbol{\phi}_v)}(y_t - \phi_\mu)^2.
\end{align*}
From Step 2, on the ball $$\{\boldsymbol{\phi}: \|\boldsymbol{\phi} - \boldsymbol{\phi}^\star \|\leq r \},$$
there exist constants $\underline{\omega}$ and $\overline{\sigma}^{\,2}$ such that
$$\underline{\omega} \leq \sigma^2_t(\boldsymbol{\phi}_v)\leq \overline{\sigma}^{\,2}, \,\, t= 1, \dots, T.$$
Moreover, by \eqref{eq:bound_mu} and \eqref{eq:bound_data}, 
$$|y_t - \phi_\mu| \leq M + \overline{\mu}, \,\, t= 1, \dots, T.$$
These bounds imply that, uniformly for $t=1,\dots,T$ and $\boldsymbol{\phi}$ such that $\|\boldsymbol{\phi}-\boldsymbol{\phi}^\star \|\leq r$,
\begin{align}\label{eq:compact_domain}
(u, v) = (y_t -\phi_\mu, \sigma^2_t(\boldsymbol{\phi}_v)) \in [-(M + \overline{\mu}), M + \overline{\mu}] \times [\underline{\omega},\overline{\sigma}^{\,2}],    
\end{align}
which is a compact subset of $\mathbb{R}\times (0, \infty)$. The function 
\begin{align}\label{eq:function_u_v}
  (u,v) \mapsto -\frac{1}{2}\log v - \frac{1}{2}u^2/v  
\end{align}
is three times continuously differentiable for $v>0$, with bounded partial derivatives up to order three on the compact subset of $\mathbb{R}\times (0,\infty)$ defined in \eqref{eq:compact_domain}. By the chain rule, boundedness of third-order derivatives of $\ell_t(\boldsymbol{\phi})$ follows once the derivatives up to order three of $u=y_t-\phi_{\mu}$ and $v=\sigma_t^2(\boldsymbol{\phi}_v)$ are bounded on the ball $\{\boldsymbol{\phi}: \|\boldsymbol{\phi}-\boldsymbol{\phi}^\star\|\leq r\}$. This is immediate for $u$. For $v=\sigma^2(\boldsymbol{\phi}_v)$, let $\partial \sigma^2(\boldsymbol{\phi}_v)$ denote an arbitrary partial derivative of $\sigma^2(\boldsymbol{\phi}_v)$, of total order at most three (including mixed partial derivatives), with respect to the components of $\boldsymbol{\phi}_v$, for example
$$\frac{\partial\sigma^2(\boldsymbol{\phi}_v)}{\partial \phi_{\beta_1}}, \,\, \frac{\partial^2\sigma^2(\boldsymbol{\phi}_v)}{\partial \phi_{\beta_1}\partial\phi_{\alpha_1}},\,\, \frac{\partial^3\sigma^2(\boldsymbol{\phi}_v)}{\partial^2 \phi_{\beta_1}\partial\phi_{\omega}}.$$
Recall the GARCH($1,1$) recursion,
\begin{align*}
    \sigma^2_t(\boldsymbol{\phi}_v)  & = \exp(\phi_\omega) + \exp(\phi_{\alpha_1}) (y_{t-1} - \phi_{\mu})^2 +  \exp(\phi_{\beta_1}) \sigma^2_{t-1}(\boldsymbol{\phi}_v).
\end{align*}
Any partial derivative can be written as a recursion
\begin{align}\label{eq:partial_sigma2_recursion}
    \partial\sigma^2_t(\boldsymbol{\phi}_v)  & = \exp(\phi_{\beta_1}) \partial\sigma^2_{t-1}(\boldsymbol{\phi}_v) + F_t(\boldsymbol{\phi}),
\end{align}
where $F_t(\boldsymbol{\phi})$ collects all remaining terms arising from differentiation of the exponential coefficients and the quadratic term $(y_{t-1}-\phi_\mu)^2$, with the data assumed bounded. On the compact ball, all derivatives up to order three of the exponential coefficients are bounded, and any terms arising from repeated application of the product rule are finite sums of products of such bounded quantities. Hence, there exists a constant $D$ such that $$\sup_{\|\boldsymbol{\phi} - \boldsymbol{\phi}^\star \|\leq r} |F_t(\boldsymbol{\phi})|\leq D, \,\, t= 1, \dots, T.$$
Iterating the recursion in \eqref{eq:partial_sigma2_recursion} from any finite start value $\partial\sigma^2_0(\boldsymbol{\phi}_v)<\infty$ and taking the supremum yields
\begin{align}\label{eq:bound1_partial_sigma_t}
\sup_{\|\boldsymbol{\phi} - \boldsymbol{\phi}^\star \|\leq r} |\partial \sigma^2_t(\boldsymbol{\phi}_v)| & \leq \overline{\beta}_1^{\,t}\sup_{\|\boldsymbol{\phi} - \boldsymbol{\phi}^\star \|\leq r} |\partial\sigma^2_0(\boldsymbol{\phi}_v)| +  D\sum_{k=0}^{t-1} \overline{\beta}_1^{\,k}.
\end{align}
The bound of the geometric term is given in \eqref{eq:bound_geometric_sum}, and since from Step 2, $\overline{\beta}_1<1$, we conclude that
\[
\mathop{\max\vphantom{\sup}}\limits_{1 \le t \le T}\;
\mathop{\sup\vphantom{\max}}\limits_{\|\boldsymbol{\phi} -\boldsymbol{\phi}^\star\|\leq r}| \partial
\sigma^2_t(\boldsymbol{\phi}_v)| \leq \mathop{\vphantom{\max}\sup}\limits_{\|\boldsymbol{\phi} -\boldsymbol{\phi}^\star\|\leq r}
|\partial \sigma^2_0(\boldsymbol{\phi}_v)| + D\frac{1}{\epsilon}  < \infty.
\]
Thus, combining with the bounded partial derivatives of \eqref{eq:function_u_v}, we conclude that there exists a $K$ such that,
$$\mathop{\max\vphantom{\sup}}\limits_{1 \le t \le T}\;
\mathop{\sup\vphantom{\max}}\limits_{\|\boldsymbol{\phi} -\boldsymbol{\phi}^\star\|\leq r}\| \nabla^3_{\boldsymbol{\phi}}\ell_t(\boldsymbol{\phi})\| \leq K,$$
which verifies that the assumption in Lemma \ref{lem:taylor_remainder} holds for the GARCH($1,1$) model with Gaussian measurement errors under the unrestricted parameterisation.

\subsection{Gradient and Hessian of the log-density measurements}
In deriving the general structure of the gradient and Hessian of $\ell_t$, we use the more explicit notation $$\ell_t(\boldsymbol{\theta}) = \ell_t(\sigma^2_t(\boldsymbol{\theta}_v), \mu, \eta_\varepsilon)$$ for clarity,  with
\begin{align*}
    \sigma^2_t : & \, \mathbb{R}^{\dim_{\boldsymbol{\theta}_v}} \rightarrow (0,\infty).
\end{align*}
With this notation, the gradient with respect to the parameter vector $\boldsymbol{\theta}=(\boldsymbol{\theta}_v^\top, \theta_\varepsilon)^\top$ decomposes as
\begin{align}\label{eq:gradient_log_dens_vartheta}
\nabla_{\boldsymbol{\theta}}\ell_t(\boldsymbol{\theta}) & = \left(\nabla_{\boldsymbol{\theta}_v}\ell_t(\sigma^2_t(\boldsymbol{\theta}_v), \mu, \theta_\varepsilon)^\top, \frac{\partial \ell_t(\sigma^2_t(\boldsymbol{\theta}_v), \mu, \theta_\varepsilon)}{\partial\theta_\varepsilon}\right)^\top,  
\end{align}
whereby the chain rule, the gradient with respect to $\boldsymbol{\theta}_v$ ($\theta_\varepsilon$ does not depend on $\boldsymbol{\theta}_v$)
\begin{align}\label{eq:gradient_theta}
\nabla_{\boldsymbol{\theta}_v}\ell_t(\sigma^2_t(\boldsymbol{\theta}_v), \mu, \theta_\varepsilon) & = \frac{\partial \ell_t}{\partial\sigma^2_t}\nabla_{\boldsymbol{\theta}_v}\sigma^2_t(\boldsymbol{\theta}_v) + \frac{\partial \ell_t}{\partial\mu} \nabla_{\boldsymbol{\theta}_v}\mu \nonumber \\
 & = \frac{\partial \ell_t}{\partial\sigma^2_t}\nabla_{\boldsymbol{\theta}_v}\sigma^2_t(\boldsymbol{\theta}_v) + \frac{\partial \ell_t}{\partial\mu} \mathbf{e}_1, 
\end{align}
where $\mathbf{e}_1$ is the first canonical basis vector (assuming $\mu$ appears as the first element in $\boldsymbol{\theta}_v$), and using the shorthands
$$\frac{\partial \ell_t}{\partial\sigma^2_t} = \frac{\partial \ell_t(\sigma^2_t(\boldsymbol{\theta}_v), \mu, \theta_\varepsilon)}{\partial\sigma^2_t(\boldsymbol{\theta}_v)} \,\, \text{ and } \,\, \frac{\partial \ell_t}{\partial\mu} = \frac{\partial \ell_t(\sigma^2_t(\boldsymbol{\theta}_v), \mu, \theta_\varepsilon)}{\partial\mu}.$$

The Hessian with respect to the parameter vector has the following structure
\begin{align}
\nabla^2_{\boldsymbol{\theta}}\ell_t(\boldsymbol{\theta}) = & 
\begin{bmatrix}
\nabla^2_{\boldsymbol{\theta}_v}\ell_t(\sigma^2_t(\boldsymbol{\theta}_v), \mu, \theta_\varepsilon) & \nabla_{\boldsymbol{\theta}_v}\frac{\partial \ell_t(\sigma^2_t(\boldsymbol{\theta}_v), \mu, \theta_\varepsilon)}{\partial\theta_\varepsilon} \\
\left(\nabla_{\boldsymbol{\theta}_v}\frac{\partial \ell_t(\sigma^2_t(\boldsymbol{\theta}_v), \mu, \theta_\varepsilon)}{\partial\theta_\varepsilon}\right)^\top & \frac{\partial^2 \ell_t(\sigma^2_t(\boldsymbol{\theta}), \mu, \theta_\varepsilon)}{\partial\theta_\varepsilon^2}
\end{bmatrix},
\end{align}
which simplifies to the top-left element for the Gaussian density (recall $\theta_\varepsilon$ is void for the Gaussian case).

For the Hessian with respect to $\boldsymbol{\theta}_v$, we apply the product rule for a scalar times a vector,
\begin{align}\label{eq:Hess_term1}
\nabla_{\boldsymbol{\theta}_v}\left(\frac{\partial \ell_t}{\partial\sigma^2_t} \nabla_{\boldsymbol{\theta}}\sigma^2_t(\boldsymbol{\theta}_v)\right) & = \left(\nabla_{\boldsymbol{\theta}_v} \frac{\partial \ell_t}{\partial\sigma^2_t} \right) \nabla_{\boldsymbol{\theta}_v}\sigma^2_t(\boldsymbol{\theta}_v)^\top  +  \frac{\partial \ell_t}{\partial\sigma^2_t} \nabla^2_{\boldsymbol{\theta}_v}\sigma^2_t(\boldsymbol{\theta}_v),
\end{align}
for the first term in \eqref{eq:gradient_theta},
and,
\begin{align}\label{eq:Hess_term2}
\nabla_{\boldsymbol{\theta}_v}\left(\frac{\partial \ell_t}{\partial\mu} \mathbf{e}_1\right) & = \left(\nabla_{\boldsymbol{\theta}_v}\frac{\partial \ell_t}{\partial\mu} \right)\mathbf{e}_1^\top.
\end{align}
for the second term in \eqref{eq:gradient_theta}. The Hessian $\nabla^2_{\boldsymbol{\theta}_v} \sigma_t^2(\boldsymbol{\theta}_v)$ of the variance model is derived in the corresponding subsection.
Now, in \eqref{eq:Hess_term1},
\begin{align}\label{eq:Hess_term1_term}
\nabla_{\boldsymbol{\theta}_v} \frac{\partial \ell_t}{\partial\sigma^2_t} & = \frac{\partial^2 \ell_t}{\partial(\sigma^2_t)^2}\nabla_{\boldsymbol{\theta}_v}\sigma^2_t(\boldsymbol{\theta}_v) + \frac{\partial^2 \ell_t}{\partial\mu\partial\sigma^2_t} \nabla_{\boldsymbol{\theta}_v}\mu \nonumber \\    & = \frac{\partial^2 \ell_t}{\partial(\sigma^2_t)^2}\nabla_{\boldsymbol{\theta}_v}\sigma^2_t(\boldsymbol{\theta}_v) + \frac{\partial^2 \ell_t}{\partial\mu\partial\sigma^2_t}\mathbf{e}_1,
\end{align}
and in \eqref{eq:Hess_term2},
\begin{align}\label{eq:Hess_term2_term}
\nabla_{\boldsymbol{\theta}_v}\frac{\partial \ell_t}{\partial\mu}  & =\frac{\partial^2 \ell_t}{\partial\sigma^2_t\partial\mu}\nabla_{\boldsymbol{\theta}_v}\sigma^2_t(\boldsymbol{\theta}_v) + \frac{\partial^2 \ell_t}{\partial\mu^2} \nabla_{\boldsymbol{\theta}_v}\mu \nonumber \\    & = \frac{\partial^2 \ell_t}{\partial\sigma^2_t\partial\mu} \nabla_{\boldsymbol{\theta}_v}\sigma^2_t(\boldsymbol{\theta}_v) + \frac{\partial^2 \ell_t}{\partial\mu^2} \mathbf{e}_1.
\end{align}
Collecting all terms gives
\begin{align*} %\label{eq:Hess_theta}
\nabla_{\boldsymbol{\theta}_v}^2\ell_t(\sigma^2_t(\boldsymbol{\theta}_v), \mu, \theta_\varepsilon) & = \frac{\partial^2 \ell_t}{\partial(\sigma^2_t)^2}\nabla_{\boldsymbol{\theta}_v}\sigma^2_t(\boldsymbol{\theta}_v)\nabla_{\boldsymbol{\theta}_v}\sigma^2_t(\boldsymbol{\theta}_v)^\top + \frac{\partial^2 \ell_t}{\partial\mu\partial\sigma^2_t} \mathbf{e}_1\nabla_{\boldsymbol{\theta}_v}\sigma^2_t(\boldsymbol{\theta}_v)^\top +  \frac{\partial \ell_t}{\partial\sigma^2_t} \nabla^2_{\boldsymbol{\theta}_v}\sigma^2_t(\boldsymbol{\theta}_v) \\
 & + \frac{\partial^2 \ell_t}{\partial\sigma^2_t\partial\mu} \nabla_{\boldsymbol{\theta}_v}\sigma^2_t(\boldsymbol{\theta}_v)\mathbf{e}_1^\top + \frac{\partial^2 \ell_t}{\partial\mu^2} \mathbf{e}_1\mathbf{e}_1^\top,
\end{align*}
simplifying to
\begin{align} \label{eq:Hess_theta_log_density}
\nabla_{\boldsymbol{\theta}_v}^2\ell_t(\sigma^2_t(\boldsymbol{\theta}_v), \mu, \theta_\varepsilon) & = \frac{\partial^2 \ell_t}{\partial(\sigma^2_t)^2}\nabla_{\boldsymbol{\theta}_v}\sigma^2_t(\boldsymbol{\theta}_v)\nabla_{\boldsymbol{\theta}_v}\sigma^2_t(\boldsymbol{\theta}_v)^\top + \frac{\partial^2 \ell_t}{\partial\sigma^2_t\partial\mu}\left( \nabla_{\boldsymbol{\theta}_v}\sigma^2_t(\boldsymbol{\theta}_v)\mathbf{e}_1^\top + \mathbf{e}_1\nabla_{\boldsymbol{\theta}_v}\sigma^2_t(\boldsymbol{\theta}_v)^\top \right) \nonumber   \\
 & +  \frac{\partial \ell_t}{\partial\sigma^2_t} \nabla^2_{\boldsymbol{\theta}_v}\sigma^2_t(\boldsymbol{\theta}_v) + \frac{\partial^2 \ell_t}{\partial\mu^2} \mathbf{e}_1\mathbf{e}_1^\top.
\end{align}

For the Student-$t$ measurement error, we also need the off-diagonal elements \begin{align}\label{eq:hess_of_diag}
\nabla_{\boldsymbol{\theta}_v}\frac{\partial\ell_t(\sigma^2_t(\boldsymbol{\theta}), \mu, \theta_\varepsilon)}{\partial \theta_\varepsilon} & = \frac{\partial^2 \ell_t}{\partial\sigma^2_t\partial\theta_\varepsilon}\nabla_{\boldsymbol{\theta}_v}\sigma^2_t(\boldsymbol{\theta}_v) + \frac{\partial^2 \ell_t }{\partial\mu\partial \theta_\varepsilon}  \mathbf{e}_1.
\end{align}

The partial derivatives of $\ell_t$ are measurement model-specific and given in Propositions \ref{supp:prop_partial_Gauss} and \ref{supp:prop_partial_Studentt}, whereas the gradient $\nabla_{\boldsymbol{\theta}}\sigma^2_t(\boldsymbol{\theta})$ and Hessian $\nabla^2_{\boldsymbol{\theta}}\sigma^2_t(\boldsymbol{\theta})$ are conditional variance process specific and given in Propositions \ref{prop:grad_Hess_GARCH} and \ref{prop:grad_Hess_TGARCH}.

\subsection{Additional results and model extension for the variational Bayes application to the Dow Jones GARCH($1,1$) model}{\label{subsec:additional_results_App1}}
\subsubsection{Additional figures}
Figure \ref{fig:App1_VB_normal_ELBO} shows the 
convergence of the upper median run (third out of $n_{\mathrm{rep}}=6$ runs, ranked by ELBO) for each method. The first 500 iterations use the Euclidean gradient, after which the natural gradient is employed. The ELBO converges for all methods, although occasional large spikes are visible due to the stochastic nature of the algorithms.
\begin{figure}[H]
    \centering    \includegraphics[width=0.95\textwidth]{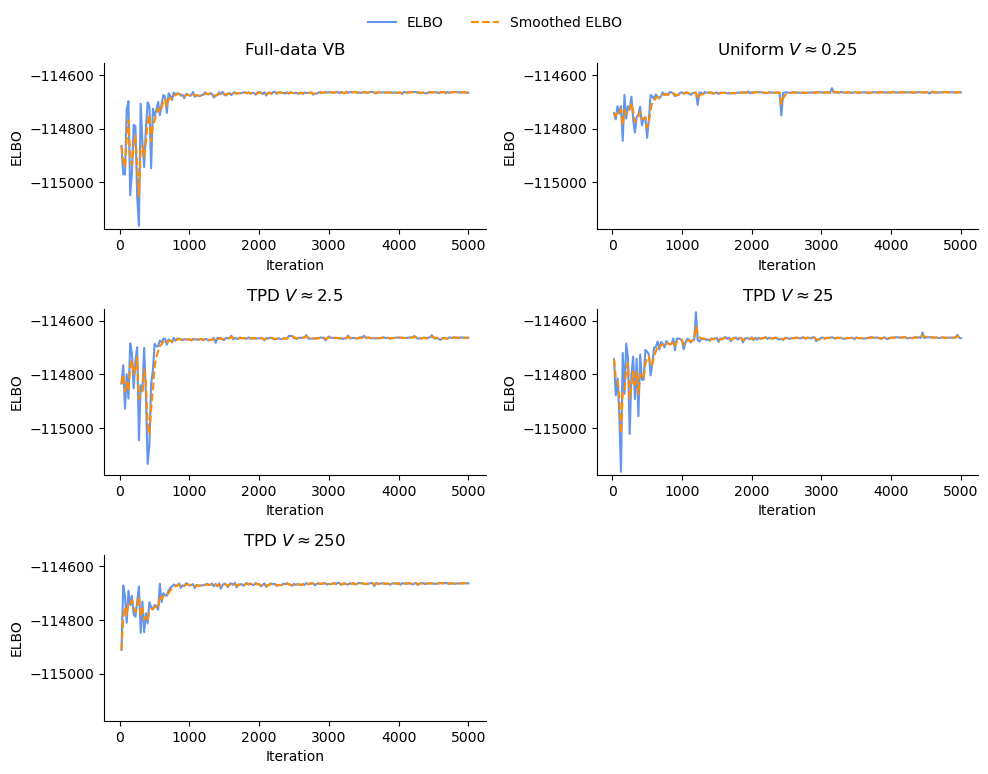}
    \caption{Additional results for the GARCH($1,1$) with normal errors in Section \ref{sec:VB_application}. Evidence lower bound (ELBO) (blue line) and smoothed ELBO (orange dashed line) versus iteration for the upper median run (third out of $n_{\mathrm{rep}}=6$ runs, ranked by ELBO) for each method in the application in Section \ref{sec:VB_application}. The smoothed ELBO is obtained using an exponential moving average of the ELBO estimates across iterations with smoothing constant 0.5.}
    \label{fig:App1_VB_normal_ELBO}
\end{figure}

Figure \ref{fig:u_max} shows the largest subsample index $u_{\max}$ across optimisation iterations together with the full-data VB cost as a benchmark. The figure illustrates that subsampling substantially reduces the computational burden, particularly for large variance tolerances, which allow more aggressive decay rates, i.e.\ lower $c$ (equivalently, larger $\gamma$).

\begin{figure}[H]
    \centering    \includegraphics[width=0.95\textwidth]{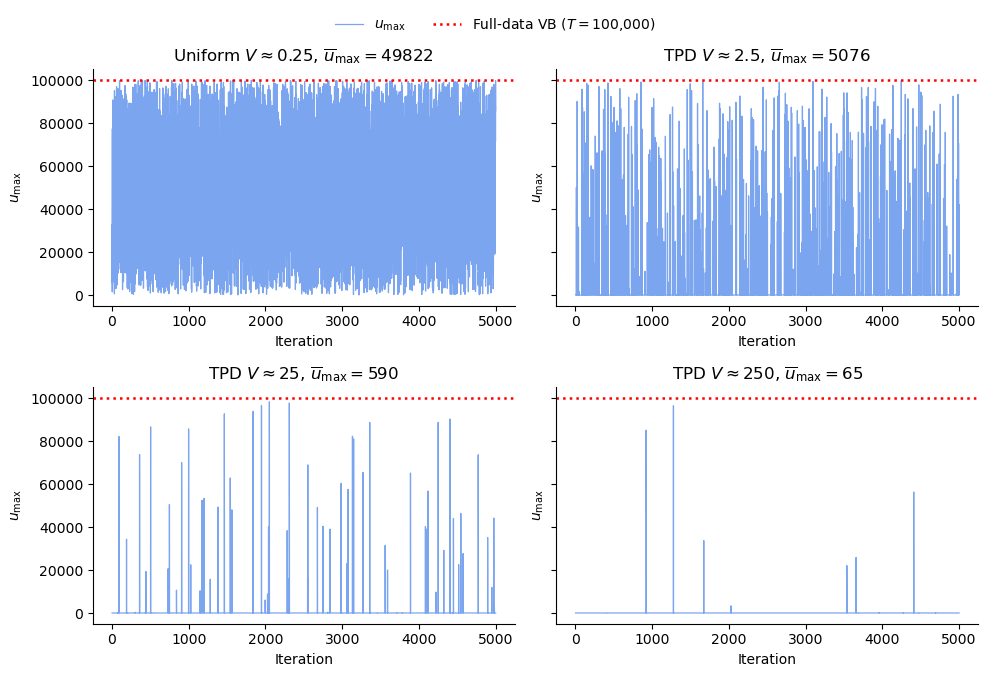}
    \caption{Additional results for the GARCH($1,1$) with normal errors in Section \ref{sec:VB_application}. Largest subsample index $u_{\max}$ (blue line) versus iterations for the subsampling methods in the application in Section \ref{sec:VB_application}. The subsampling variational Bayes (VB) method with truncated power-law decaying (TPD) sampling probabilities is shown for $R_{\max}=10, 100, 1000$, corresponding to variance tolerances $V\approx 2.5, 25, 250$. The uniform sampling case corresponds to $R_{\max}=1$ with $V \approx 0.25$. The figure also shows the full-data VB cost a as benchmark (red dotted line). The quantity $\overline{u}_{\max}$ denotes the average of $u_{\max}$ over the optimisation iterations.}
    \label{fig:u_max}
\end{figure}

\subsubsection{Accommodating outliers with Student-$t$ errors}\label{subsub:settings_Student-t}
To account for heavy-tails and extreme observations, we also consider a Student-$t$ error distribution. We use the unconstrained parameterisation
$$\boldsymbol{\phi}=(\phi_\mu,\phi_\omega,\phi_{\alpha_1} ,\phi_{\beta_1},\phi_{\varepsilon})^\top,$$
where
\begin{equation}\label{eq:transformations}
\phi_\mu=\mu, \,\, \phi_{\omega}=\log(\omega), \,\, \phi_{\alpha_1}=\log(\alpha_1), \,\, \phi_{\beta_1}=\log(\beta_1), \,\, \phi_\varepsilon=\log(\nu -2),
\end{equation}
which map the original parameters $\mu\in \mathbb{R},\omega > 0, \alpha_1 \geq 0, \beta_1 \geq 0$ and $\nu>2$ to an unrestricted space. 

We use the same priors as in Section \ref{subsec:spec_settings_App1} for $\mu,\omega,\alpha_1$, and $\beta_1$, with the joint prior supported on the stationary region
$\alpha_1 + \beta_1 < 1$. The degrees-of-freedom parameter $\nu$ is assigned the shifted Gamma prior $\nu - 2 \sim \text{Gamma}(2,1),$ which implies $\nu>2$.

We follow the settings in Section \ref{subsec:spec_settings_App1}. However, for the optimisation we do not use basin-hopping, as we initialise the parameters using estimates from the corresponding normal error model, which are assumed to be close to a good mode for the GARCH variance parameters. For the truncated decaying sampling probability weights, the constrained objective function for $R_{\max}=1000,100,10$ is shown in Figure \ref{fig:expected_cost_var_constraint_Student_t}. The figure is similar to Figure \ref{fig:expected_cost_var_constraint}, but for the Student-$t$ error model. This setting exhibits lower variability, which is likely due to more accurate control variates, resulting in less extreme control variate residuals.

\begin{figure}[h]
    \centering
    \includegraphics[width=1.0\textwidth]{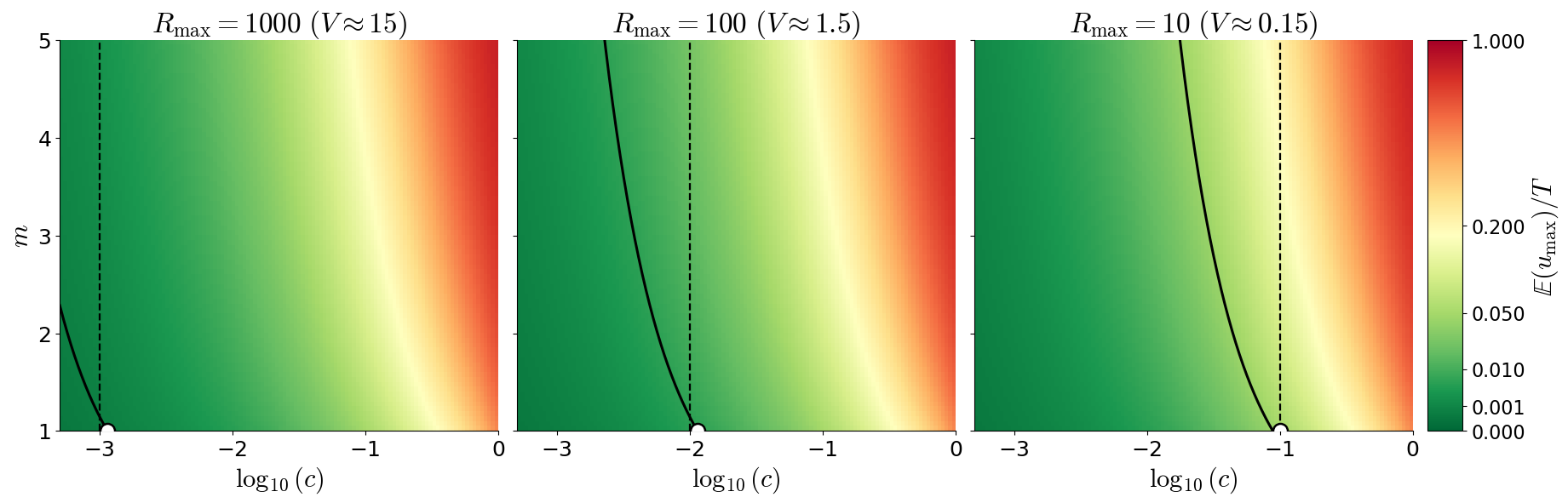}\caption{Results for the Student-$t$ error extension. Constrained objective function $\mathbb{E}(u_{\max})$ defined in \eqref{eq:miminises_compute_subject_variance} (normalised by $T$, with $T=100{,}000$) and described in Section \ref{subsec:balancing_variance_and_computational_cost}, shown as a function of $c$ and $m$ (where $m$ is treated as continuous for visualisation). The black curve marks the variance constraint, and the dashed vertical line marks the safeguard lower bound constraint. Results are shown for three choices of $R_{\max}=1/c_{\min}$ (see panel titles) and $t^\star=1{,}000$. The white circles indicate the optima for each setting. The colour scale represents the expected computational cost ratio, where red indicates no compute savings relative to the full dataset of size $T$, and green indicates increasing compute savings.}\label{fig:expected_cost_var_constraint_Student_t}
\end{figure}

We consider subsampling variational Bayes with $R_{\max}=100$, the recommended default value in Section \ref{subsec:results_VB_app1}. Figure \ref{fig:student_t_panels}(a) shows the convergence of the resulting ELBO and the smoothed ELBO as a function of the iterations. The ELBO increases rapidly and then stabilises, apart from occasional stochastic spikes. The subsampling VB optimisation, including its tuning, takes approximately 533 seconds to run, compared to approximately 31,589 seconds for the full-data Markov chain Monte Carlo (MCMC), corresponding to an approximately 59-fold speed-up. Figure \ref{fig:student_t_panels}(b) confirms the small realised values of $u_{\max}$ (on average 581) over the optimisation iterations, explaining the substantial computational speed-up.

\begin{figure}[H]
\centering

\begin{subfigure}{0.48\linewidth}
\centering
\includegraphics[width=\linewidth]{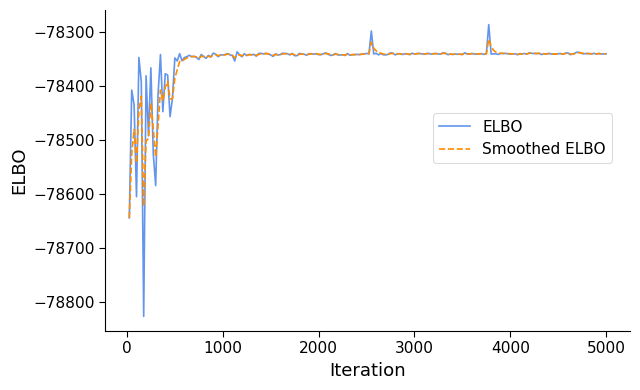}
\caption{ELBO and smoothed-ELBO.}
\end{subfigure}
\hfill
\begin{subfigure}{0.48\linewidth}
\centering
\includegraphics[width=\linewidth]{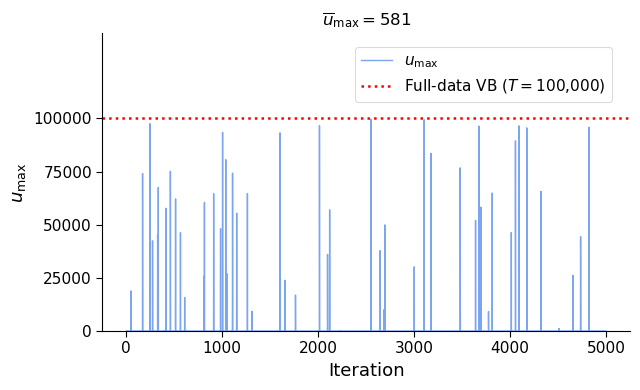}
\caption{$u_{\max}$.}
\end{subfigure}

\caption{(a) Results for the Student-$t$ error extension. Evidence lower bound (ELBO) and smoothed ELBO versus optimisation iteration for the Student-$t$. The smoothed ELBO is obtained using an exponential moving average of the ELBO estimates across iterations with smoothing constant 0.5. (b) Largest subsample index $u_{\max}$ across optimisation iterations. The red dotted line shows the full-data VB benchmark $T=100{,}000$. The quantity $\overline{u}_{\max}$ denotes the average of $u_{\max}$ over the optimisation iterations. Both panels correspond to subsampling variational Bayes using truncated power-law weights tuned with $R_{\max}=100$.}
\label{fig:student_t_panels}
\end{figure}

Finally, Figure \ref{fig:App1_Student_t_VB_posterior_kde_comparison} compares the marginal posterior estimates produced by the variational approximation with the full-data MCMC. As noted in Section \ref{sec:VB_application}, discrepancies are likely due to the restricted variational family. The next section considers subsampling Markov chain Monte Carlo, which achieves substantially improved posterior accuracy.

\begin{figure}[h]
\centering
\includegraphics[width=\linewidth]{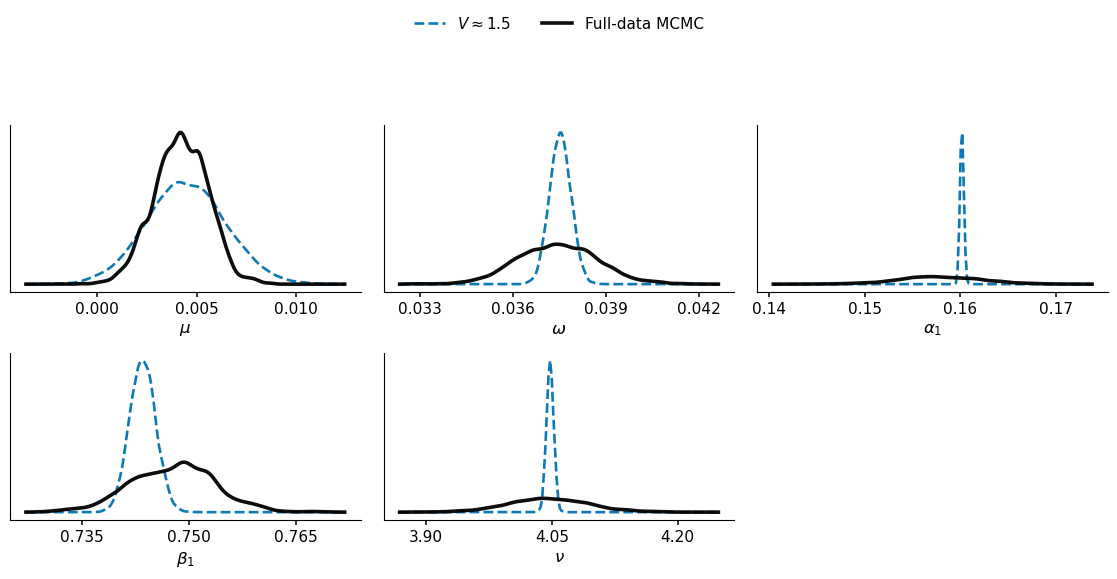}
\caption{Results for the Student-$t$ error extension. Posterior marginal distributions under the original parameterisation $\boldsymbol{\theta}$. The dashed blue curves correspond to subsampling variational Bayes with truncated decaying sampling probabilities tuned using $R_{\max}=100$, while the solid black curves show the full-data Markov chain Monte Carlo posterior for reference.}
\label{fig:App1_Student_t_VB_posterior_kde_comparison}
\end{figure}

\section{Additional empirical example: Subsampling Markov chain Monte Carlo}\label{sec:MCMC_application}

\subsection{Background, model, and data}
Our second and more voluminous application considers estimating, via a subsampling-based Markov chain Monte Carlo method, GARCH models and a threshold extension thereof. There are many subsampling-based Markov chain Monte Carlo approaches, for example, \cite{korattikara2014austerity,bardenet2014towards, bardenet2017markov, chen2014stochastic, quiroz2019speeding, prado2026metropolis}. \cite{johndrow2020no} show that subsampling-based Markov chain Monte Carlo methods face inherent limitations in their ability to scale to large datasets. However, they also identify efficient control variates as an important exception, under which meaningful practical computational speed-ups may still be achievable (see also \citealp{rudolf2024perturbations}). We consider the so-called subsampling Markov chain Monte Carlo (MCMC) approach developed in \cite{quiroz2019speeding}. Several extensions of this approach have been proposed, including \cite{quiroz2018delay, dang2019hamiltonian,salomone2020spectral, Villani2024spectral, goodwin2026spectral}. See \cite{quiroz2018subsampling, quiroz2023wiley} for introductory surveys and overviews to this class of subsampling-based Markov chain Monte Carlo methods.

We consider the following four candidate volatility models for Apple Inc log-returns described below: GARCH($1, 1$) with normal errors, GARCH($1, 1$) with Student-$t$ errors, TGARCH($1, 1$) with Student-$t$ errors, and TGARCH($1, 2$) with Student-$t$ errors. The goal is to select the model that provides the best fit based on out-of-sample predictive performance. Ideally, this would be conducted using a fully sequential out-of-sample evaluation based on a rolling window, which requires repeated reestimation of the model as new observations become available. However, Markov chain Monte Carlo methods, including subsampling variants, are not well suited for sequential estimation. For this reason, practitioners often rely on sequential Monte Carlo methods in online settings \citep{chopin2002sequential}. We leave incorporating our estimators into subsampling sequential Monte Carlo methods, such as \cite{gunawan2020subsampling}, for future work. 

In this application, we instead restrict attention to a single out-of-sample test set consisting of $T_{\mathrm{test}}=10{,}000$ observations, covering the period 2024-05-06 to 2024-06-28. As in the main paper, we use $T=100{,}000$ observations to estimate the posterior distribution, but here these observations span the period 2023-01-30 to 2024-05-06. Threshold effects can be difficult to identify at an aggregate level, such as in the Dow Jones Industrial Average data in Section \ref{sec:VB_application}. We therefore focus on a single stock in this application, namely Apple Inc. (AAPL), shown in Figure \ref{fig:AppleInc}.

\begin{figure}[H]
    \centering    \includegraphics[width=0.65\textwidth]{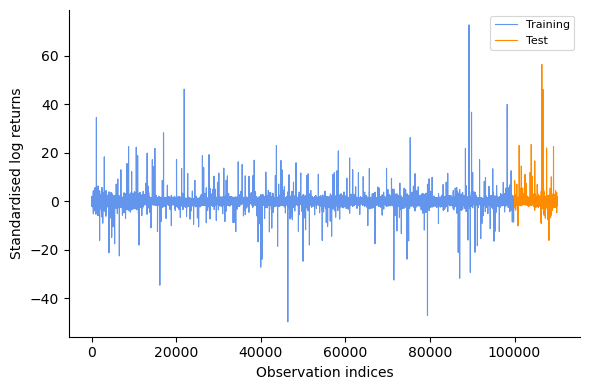}
    \caption{Observations $y_t$ for the Apple Inc stock (AAPL), constructed as one-minute log-returns over the period 2023-01-30 to 2024-06-28 and rescaled to have unit sample standard deviation. The blue line corresponds to the training data ($T=100{,}000$), while the orange line corresponds to the test data ($T_{\mathrm{test}} = 10{,}000$).}
    \label{fig:AppleInc}
\end{figure}

\subsection{Subsampling Markov chain Monte Carlo}
The subsampling MCMC method in 
\cite{quiroz2019speeding} constructs an unbiased estimator of the log-likelihood using uniform sampling weights and the difference estimator in \eqref{eq:diff_estimator}. We follow the same framework, but replace the uniform sampling weights with truncated power-law decaying sampling probabilities and use the weighted difference estimator in \eqref{eq:weighted_diff_estimator}. We note at the outset that, in the tuning regime of recursive likelihood models, uniform subsampling may perform poorly; Section \ref{subsec:limitations_uniform_recursive} examines this in detail.

To use the log-likelihood estimator within a pseudo-marginal MCMC algorithm, which requires an estimate of the likelihood in the ordinary scale, \cite{quiroz2019speeding} convert the log-likelihood estimator into a likelihood estimator by applying a bias-correction as in \cite{ceperley1999penalty}. In our setting, the likelihood estimator takes the form
\begin{align}\label{eq:bias-corrected-likelihood-estimator}
\exp\left(\widehat{\ell}_{\mathrm{WDE}}(\boldsymbol{\phi})-\frac{1}{2}\mathbb{V}(\widehat{\ell}_{\mathrm{WDE}}(\boldsymbol{\phi}))\right),
\end{align}
which is an unbiased estimator of $p(\mathbf{y}_{1:T}|\boldsymbol{\phi})$ if $\widehat{\ell}_{\mathrm{WDE}}$ is normally distributed \citep{ceperley1999penalty,quiroz2019speeding}. 

In practice, however, the variance term in \eqref{eq:bias-corrected-likelihood-estimator} must itself be estimated and $\widehat{\ell}_{\mathrm{WDE}}(\boldsymbol{\phi})$ may not be normally distributed, as is the case in our setting. Under these conditions, \cite{quiroz2019speeding} show that the resulting pseudo-marginal algorithm targets a perturbed posterior that is within $\mathcal{O}(m^{-2}T^{-1})$ of the true posterior in total variation norm, under some regularity conditions \cite[Assumption 2]{quiroz2019speeding}. These regularity conditions are those imposed in \cite{chen1985asymptotic} to obtain a Bernstein-von Mises result (asymptotic normality of the posterior). In Section \ref{subsec:WDE}, we argue that these conditions are unnecessarily strong for our theoretical analysis of the stabilised weighted estimator, which primarily requires establishing the asymptotic order of the control variate error. Accordingly, Lemma \ref{lem:taylor_remainder} is developed under substantially milder assumptions, which we verify in Section \ref{supp:verify_Lemma2_ass} for the GARCH model considered in the main paper. To utilise the perturbation result in \cite{quiroz2019speeding} for the present application, however, the full Bernstein-von Mises assumptions would need to be verified. Establishing these conditions for GARCH or TGARCH models is beyond the scope of our paper. Instead, for the model specifications where reliable posterior exploration is achieved, we verify empirically that the posterior approximation obtained via subsampling MCMC is close to that obtained from full-data MCMC.

\subsection{Specifications and settings}\label{subsec:spec_settings_App2}
We use the same transformations to an unrestricted space and prior specifications as in Sections \ref{subsec:spec_settings_App1} and \ref{subsub:settings_Student-t} for $\mu$, $\omega$, $\phi_{\alpha_1}, \dots, \phi_{\alpha_p}$, $\phi_{\beta_1} \dots \phi_{\beta_q}$, and $\phi_{\varepsilon}$. For the additional threshold parameters we use the transformation $\phi_{\gamma_i}=\log(\gamma_i)$ for $i=1,\dots,p$ and assign independent marginal priors
$$\gamma_i \sim \text{Half-Normal}(0.2),\,\, i=1,\dots p.$$
Similar to the GARCH model, the joint prior is restricted to the stationary region for threshold GARCH models. In this case, the stationarity condition is given by
\begin{align}\label{eq:stationary_cond_TGARCH}    
\sum_{i=1}^p\alpha_i + \frac{1}{2}\sum_{i=1}^p\gamma_i + \sum_{j=1}^q\beta_j & < 1.
\end{align}
Our reparameterisation $\boldsymbol{\phi}$ ensures that the positivity constraints are respected, but does not enforce the stationarity condition. In the application considered in the main paper, this did not pose a problem, as the posterior of $\alpha_1 + \beta_1$ lay far from the non-stationary boundary when estimated using the Dow Jones data (posterior mean from full-data MCMC: 0.9474). The AAPL data, however, exhibits substantially higher persistence (corresponding posterior mean: 0.9998). In this setting, constructing a random walk proposal in the $\boldsymbol{\phi}$-space leads to a high rejection rate, since proposed values frequently violate the stationarity constraint. See Section \ref{subsec:stationary_constrained_proposal} for an efficient stationary-constrained random-walk proposal.

For our stabilised weighted subsampling methodology, we follow the tuning strategy described in Section \ref{subsec:spec_settings_App1}, with two modifications. First, we impose $m^{\star}\geq 2$, which is required for the unbiased estimator of the variance term in \eqref{eq:bias-corrected-likelihood-estimator} to be well-defined, as it has $m^{\star}-1$ in the denominator. Second, in the pilot-based calibration of the variance tolerance $V$, we replace the $m=1$ reference with $m=2$, ensuring that the reference variance remains well-defined. All other aspects of the tuning procedure remain unchanged. Figure \ref{fig:expected_cost_var_constraint_subsampling_application} shows $\mathbb{E}(u_{\max})/T$ as a function of $c$ and $m$ for the four models, with the feasible region determined jointly by the variance tolerance obtained under the default recommendation $R_{\max}=100$ (implying $c_{\min}=0.01$)  and the safeguard lower bound $c\geq c_{\min}$; see the caption for details. The optimal tuning parameters (white circles) are indicated in each panel. In two cases, the variance constraint is binding at the optimum, while in the other two the safeguard lower bound on $c$ is binding. We observe substantial variation in the variance tolerance $V$ across models, while all settings admit significant potential computational savings, with the optima lying well within the low-cost (green) region.

\begin{figure}[H]
    \centering
    \includegraphics[width=1.0\textwidth]{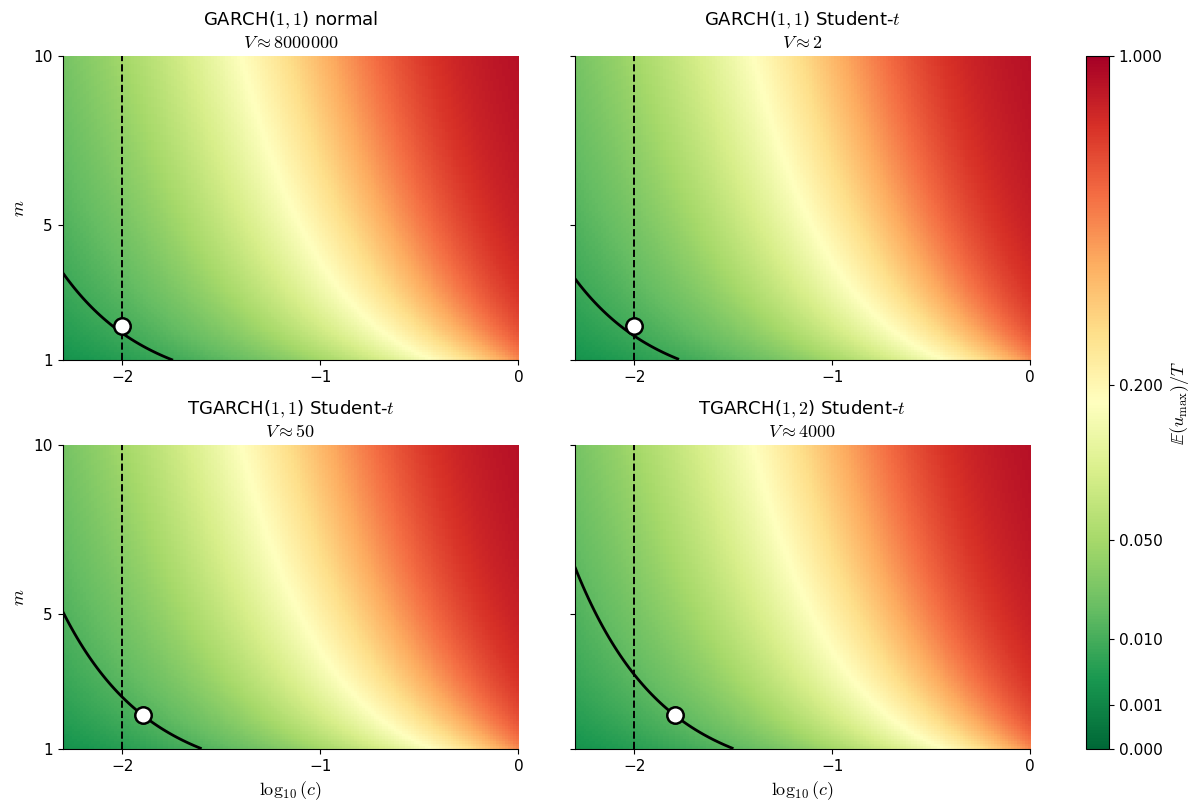}\caption{Constrained objective function $\mathbb{E}(u_{\max})$ defined in \eqref{eq:miminises_compute_subject_variance} (normalised by $T$, with $T=100{,}000$) and described in Section \ref{subsec:balancing_variance_and_computational_cost}, shown as a function of $c$ and $m$ (where $m$ is treated as continuous for visualisation). The black curve marks the variance constraint, and the dashed vertical line marks the safeguard lower bound constraint. Results are shown for four models (see panel titles) in the subsampling Markov chain Monte Carlo  application, all tuned with $R_{\max}=100$ and $t^\star=1{,}000$. The white circles indicate the optima for each setting. The colour scale represents the expected computational cost ratio, where red indicates no compute savings relative to the full dataset of size $T$, and green indicates increasing compute savings.}
\label{fig:expected_cost_var_constraint_subsampling_application}
\end{figure}

We implement subsampling MCMC with truncated sampling probabilities determined by the tuning procedure described above. Using a standard random-walk proposal preconditioned by a scaled covariance matrix from a Laplace approximation, we obtain $10{,}000$ samples after a burn-in period of $2{,}000$ iterations. To assess the accuracy of the perturbed posterior distribution arising from subsampling, we also run full-data Markov chain Monte Carlo algorithms for the model specifications where reliable posterior exploration is achieved, using the same settings as in Section \ref{subsec:spec_settings_App1}, including an adaptive proposal.  All chains start from the MAP estimate.

To evaluate the out-of-sample performance, we compute the posterior predictive distribution as follows. Let $\mathcal{F}_T$ denote the information set available at the forecasting origin $T$. Under GARCH-type models, the conditional variance process $\sigma^2_t$ evolves deterministically from $\mathcal{F}_T$ given $\boldsymbol{\theta}$. The posterior predictive distribution is therefore

\vspace{0.5em}
\resizebox{0.96\linewidth}{!}{$
\begin{aligned}
  p(y_{T+1}, \sigma^2_{T+1}, \dots, y_{T+T_{\mathrm{test}}}, \sigma^2_{T+T_{\mathrm{test}}}|\mathcal{F}_T)& = \int p(y_{T+1}, \sigma^2_{T+1}, \dots, y_{T+T_{\mathrm{test}}}, \sigma^2_{T+T_{\mathrm{test}}}|\boldsymbol{\theta},\mathcal{F}_T)\pi(\boldsymbol{\theta}|\mathcal{F}_T)d\boldsymbol{\theta}  \\
  & = \int \left( \prod_{j=1}^{T_{\mathrm{test}}} p(y_{T+j}|\boldsymbol{\theta},\sigma^2_{T+j}, \mathcal{F}_{T+j-1})\right)\pi(\boldsymbol{\theta}|\mathcal{F}_T)d\boldsymbol{\theta},
\end{aligned}
$}
\vspace{0.5em}

\noindent where $p(y_{T+j}|\boldsymbol{\theta},\sigma^2_{T+j}, \mathcal{F}_{T+j-1})$ denotes the one-step-ahead predictive density of the observed test observation $y_{T+j}$ conditional on $\boldsymbol{\theta}$, and $\pi(\boldsymbol{\theta}|\mathcal{F}_T)$ is the posterior distribution based on the training data. The training data are rescaled using their standard deviation, and the same scaling is applied to the test set. This ensures that predictive densities are evaluated on the same scale as that used for model estimation. To evaluate the out-of-sample performance, we use the logarithmic scoring rule for the observed test sample (expressed on the training data scale), also known as the log predictive density score (LPDS). For the test sequence $\mathbf{y}_{T+1:T+T_{\mathrm{test}}}$, this is defined as
$$\mathrm{LPDS} = \log p(y_{T+1}, \sigma^2_{T+1}, \dots, y_{T+T_{\mathrm{test}}}, \sigma^2_{T+T_{\mathrm{test}}}|\mathcal{F}_T).$$
This quantity is estimated by Monte Carlo using MCMC draws $\{\boldsymbol{\theta}^{(m)}\}_{m=1}^M\sim \pi(\boldsymbol{\theta}|\mathcal{F}_T)$, with $M=100$. The MCMC chain is obtained by thinning, for computational reasons, a post-burn-in chain of length $10{,}000$. The resulting estimate is
$$\widehat{\mathrm{LPDS}} \approx \log\left(\frac{1}{M}\sum_{m=1}^M \left( \prod_{j=1}^{T_{\mathrm{test}}} p(y_{T+j}|\boldsymbol{\theta}^{(m)},\sigma^{2\,(m)}_{T+j}, \mathcal{F}_{T+j-1})\right) \right),$$ 
where $\sigma^{2\,(m)}_{T+j}$ is obtained by running the GARCH-type recursion forward under $\boldsymbol{\theta}^{(m)}$. Finally, an approximate confidence interval for $\widehat{\mathrm{LPDS}}$ is obtained via a delta-method approximation applied to the logarithm of the sample mean, where the Monte Carlo variance of the latter is estimated using the integrated autocorrelation time of the transformed Markov chain.

\subsection{Stationary-constrained proposals}\label{subsec:stationary_constrained_proposal}
Recall that in the GARCH$(p,q)$ model, the parameter vector $\boldsymbol{\theta}$ is subject to positivity and stationarity constraints, in particular
$$
\omega > 0, \,\, \alpha_i \ge 0, \,\, \beta_j \ge 0, \,\, \sum_i \alpha_i + \sum_j \beta_j < 1.
$$
We now outline a proposal distribution that ensures positivity and stationarity constraints by reparameterising the model in an unconstrained space $\boldsymbol{\psi}$ such that, when mapped back, the resulting $\alpha_i\geq 0$ and $\beta_j \geq 0$ always lie within the stationary region. This eliminates unnecessary rejections due to constraint violations and improves sampling efficiency.

Let $\boldsymbol{\theta}$ denote the original parameter vector, and let $\Psi$ be a one-to-one transformation
$$
\boldsymbol{\psi} = \Psi(\boldsymbol{\theta}) \in \mathbb{R}^{d_{\boldsymbol{\theta}}},
$$
such that $\Psi$ is a bijection from the stationary parameter space to $\mathbb{R}^{d_{\boldsymbol{\theta}}}$. In particular, the transformation enforces positivity and stationarity by construction. We denote the inverse transformation by $\boldsymbol{\theta} = \Psi^{-1}(\boldsymbol{\psi})$. Specifically, let
$$
\boldsymbol{\psi} = \left(\psi_{\mu}, \psi_{\omega}, \psi_{\alpha_1}, \dots, \psi_{\alpha_p}, \psi_{\beta_1}, \dots, \psi_{\beta_q}\right) \in \mathbb{R}^{2+p+q}.
$$
We define the transformation $\boldsymbol{\theta} = \Psi^{-1}(\boldsymbol{\psi})$ as
$$
\mu = \psi_{\mu} ,\,\, \omega = \exp(\psi_{\omega}),
$$
and
$$
w_i = \exp(\psi_{\alpha_i}), \quad v_j = \exp(\psi_{\beta_j}).
$$
Moreover, 
$$
S = \sum_{i=1}^p w_i + \sum_{j=1}^q v_j,
$$
such that the GARCH coefficients are given by
$$
\alpha_i = \frac{w_i}{1 + S}, \quad \beta_j = \frac{v_j}{1 + S}.
$$
Then, by construction, $\alpha_i \ge 0$, $\beta_j \ge 0$, and 
$$\sum_{i=1}^p \alpha_i + \sum_{j=1}^q \beta_j = \frac{S}{1+S} < 1.$$
In addition, we consider the working parameterisation $\boldsymbol{\phi} = h(\boldsymbol{\theta})$  used in the main paper, so that the posterior density is expressed as a function of $\boldsymbol{\phi}$, while the proposals are constructed in the unconstrained $\boldsymbol{\psi}$. The mapping between spaces is therefore
$$
\boldsymbol{\psi} \;\longleftrightarrow\; \boldsymbol{\theta} \;\longleftrightarrow\; \boldsymbol{\phi}.
$$

For sampling, the Markov chain is defined on the $\boldsymbol{\phi}$-space, but proposals are generated in the unconstrained $\boldsymbol{\psi}$-space via a Gaussian random walk, i.e.\
$$
\boldsymbol{\psi}' = \boldsymbol{\psi} + \boldsymbol{\varepsilon}, \,\, \boldsymbol{\varepsilon} \sim \mathcal{N}(\mathbf{0}, \boldsymbol{\Sigma}_{\mathrm{prop}}).
$$
The proposal is then mapped back to $\boldsymbol{\phi}$ via
$$
\boldsymbol{\theta}' = \Psi^{-1}(\boldsymbol{\psi}'), \,\, \boldsymbol{\phi}' = h(\boldsymbol{\theta}'),
$$
and the Metropolis--Hastings ratio for the full-data case is (with a symmetric proposal in $\boldsymbol{\psi}$-space)
$$
\alpha_{\mathrm{MH}}(\boldsymbol{\phi},\boldsymbol{\phi}') = \min\!\left\{1,\;
\frac{\pi_T(\boldsymbol{\phi}')}{\pi_T(\boldsymbol{\phi})}
\cdot
\frac{\left| \det \frac{\partial \boldsymbol{\psi}}{\partial \boldsymbol{\phi}}(\boldsymbol{\phi}) \right|}
     {\left| \det \frac{\partial \boldsymbol{\psi}}{\partial \boldsymbol{\phi}}(\boldsymbol{\phi}') \right|}
\right\}.
$$
The Jacobian term arises from the change of variables between $\boldsymbol{\phi}$ and $\boldsymbol{\psi}$, and is computed via the chain rule as
$$
\left| \det \frac{\partial \boldsymbol{\psi}}{\partial \boldsymbol{\phi}} \right|
=
\left| \det \frac{\partial \boldsymbol{\psi}}{\partial \boldsymbol{\theta}} \right| \cdot
\left| \det \frac{\partial \boldsymbol{\theta}}{\partial \boldsymbol{\phi}} \right|.
$$

To summarise, this construction ensures that all proposed parameters automatically satisfy the GARCH constraints, in particular stationarity, while retaining the simplicity of a Gaussian random-walk proposals in an unconstrained space. 

\begin{remark}
    A similar proposal can be constructed for TGARCH models to satisfy the stationarity condition in \eqref{eq:stationary_cond_TGARCH}. The threshold coefficients $\gamma_i$ can be incorporated into the normalisation by defining $u_i = \exp(\psi_{\gamma_i})$ and setting
$$
\alpha_i = \frac{w_i}{1 + S}, \quad \gamma_i = \frac{2 u_i}{1 + S}, \quad \beta_j = \frac{v_j}{1 + S},
$$
with
$$
S = \sum_{i=1}^p w_i + \sum_{i=1}^p u_i + \sum_{j=1}^q v_j.
$$
Then $$\alpha_i,\gamma_i,\beta_j\geq 0,$$ and
$$
\sum_{i=1}^p \alpha_i + \frac{1}{2} \sum_{i=1}^p \gamma_i + \sum_{j=1}^q \beta_j < 1.
$$
The resulting mapping again defines a bijection from the stationary region to $\mathbb{R}^{2+2p+q}$, and the proposal mechanism proceeds identically.
\end{remark}

\begin{remark}
    The $\boldsymbol{\psi}$ parameterisation above can also be used as an alternative to the variational family $\mathcal{Q}_{\boldsymbol{\lambda}}(\boldsymbol{\phi})$ in Section \ref{subsec:GVA}. The reparameterisation trick can then be applied in $\boldsymbol{\psi}$-space, ensuring that all transformed samples satisfy the positivity and stationarity conditions.
\end{remark}

\subsection{Results}\label{subsec:results_subsampling}

Pseudo-marginal chains are well known to exhibit a tendency to become stuck. To assess the stability of the subsampling MCMC algorithm, we run $n_{\mathrm{rep}}=6$ independent chains using the initialisation described above (the MAP estimate). This allows us to detect potential issues such as poor mixing or chains becoming stuck. To this end, we define the statistic longest immobility streak (LIS), which is the maximum number of consecutive iterations (out of $12{,}000$) during which the chain does not move. Table \ref{tab:mcmc_behaviour_App2} shows generally short immobility streaks across the models, but one of the six TGARCH$(1,1)$ chains with Student-$t$ errors becomes stuck for an extended period, with an LIS of $10{,}147$. The remaining chains for this model exhibit much shorter streaks, as reflected by the median LIS of $16$. Table \ref{tab:mcmc_behaviour_App2} also reports measures of computational cost, which we return to below.

\begin{table}[H]
\centering
\small
\setlength{\tabcolsep}{4pt}
\caption{Results for models fitted to the AAPL data. The table shows results for subsampling Markov chain Monte Carlo (MCMC) tuned with $R_{\max}=100$ (equivalently $c_{\min}=0.01$) and run for $12{,}000$ iterations. For each model, $n_{\mathrm{rep}}=6$ independent chains are generated from a common initialisation. The table shows the median (over $n_{\mathrm{rep}}=6$) longest immobility streak (LIS), defined as the maximum number of consecutive iterations (out of $12{,}000$) during which the chain does not move, as well as wall-clock time (CPU, in seconds), compute fraction (CF), defined as the number of log-density evaluations relative to full-data MCMC, and relative CPU (RCPU), defined as the ratio of full-data  CPU time to subsampling CPU time. For each quantity, the minimum and maximum (over \(n_{\mathrm{rep}}=6\)) are reported in parentheses. For reference, the table also shows the CPU time for a single full-data MCMC chain.}
\label{tab:mcmc_behaviour_App2}
\begin{tabular}{lccccc}
\toprule
Model and error & \makecell{Median \\ LIS} & \makecell{Median \\ CPU} & \makecell{Median \\ CF} & \makecell{Full-data \\ CPU} & \makecell{Median \\ RCPU} \\
\midrule
GARCH(1,1), normal & \makecell{24.5 \\ {(22, 34)}} & \makecell{1052.5 \\ {(948.6, 3521.9)}} & \makecell{0.018 \\ {(0.017, 0.019)}} & 69039.4 & \makecell{65.6 \\ {(19.6, 72.8)}} \\
GARCH(1,1), Student-$t$ & \makecell{16.5 \\ {(13, 22)}} & \makecell{485.3 \\ {(470.3, 507.4)}} & \makecell{0.018 \\ {(0.017, 0.019)}} & 31792.4 & \makecell{65.5 \\ {(62.7, 67.6)}} \\
TGARCH(1,1), Student-$t$ & \makecell{16.0 \\ {(14, 10147)}} & \makecell{642.3 \\ {(592.9, 649.3)}} & \makecell{0.021 \\ {(0.021, 0.022)}} & 39697.2 & \makecell{61.8 \\ {(61.1, 67.0)}} \\
TGARCH(1,2), Student-$t$ & \makecell{23.5 \\ {(16, 46)}} & \makecell{727.9 \\ {(702.5, 776.0)}} & \makecell{0.024 \\ {(0.023, 0.025)}} & 34979.4 & \makecell{48.1 \\ {(45.1, 49.8)}} \\
\bottomrule
\end{tabular}
\end{table}

Further inspection of the traceplots reveals two distinct forms of sampling difficulty. For the TGARCH($1,1$) model with Student-$t$ errors, one of the six chains exhibits an extended period of immobility, consistent with the large LIS reported in Table \ref{tab:mcmc_behaviour_App2}, while the remaining five chains mix satisfactorily; see Figure \ref{fig:TGARCH11_traceplots_app2}. For the more complex TGARCH($1,2$) model, the chains do not exhibit a comparable immobility episode, but the Markov chains for $\beta_1$ and $\beta_2$ show noticeable between-chain variation; see Figure \ref{fig:TGARCH12_traceplots_app2}. This suggests the presence of a ridge in the posterior that is difficult to traverse using a random-walk proposal once entered. The two GARCH($1,1$) models exhibit good mixing behaviour (not shown).

\begin{figure}[H]
    \centering \includegraphics[width=0.95\textwidth]{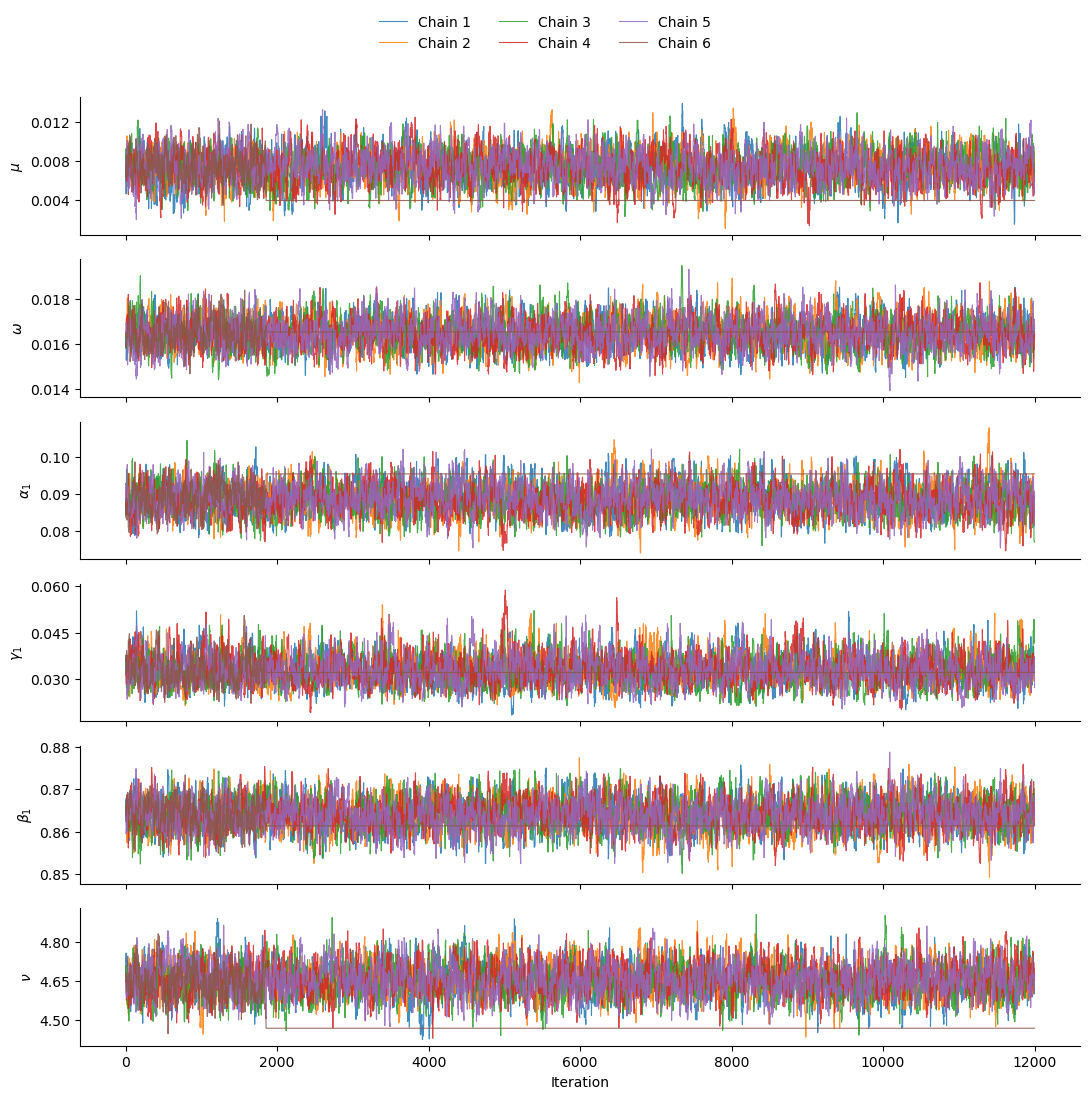}
    \caption{Results for the TGARCH($1,1$) model with Student-$t$ errors fitted to the AAPL data. Traceplots for the $n_{\mathrm{rep}}=6$ chains for each parameter.}
    \label{fig:TGARCH11_traceplots_app2}
\end{figure}

\begin{figure}[H]
    \centering \includegraphics[width=0.95\textwidth]{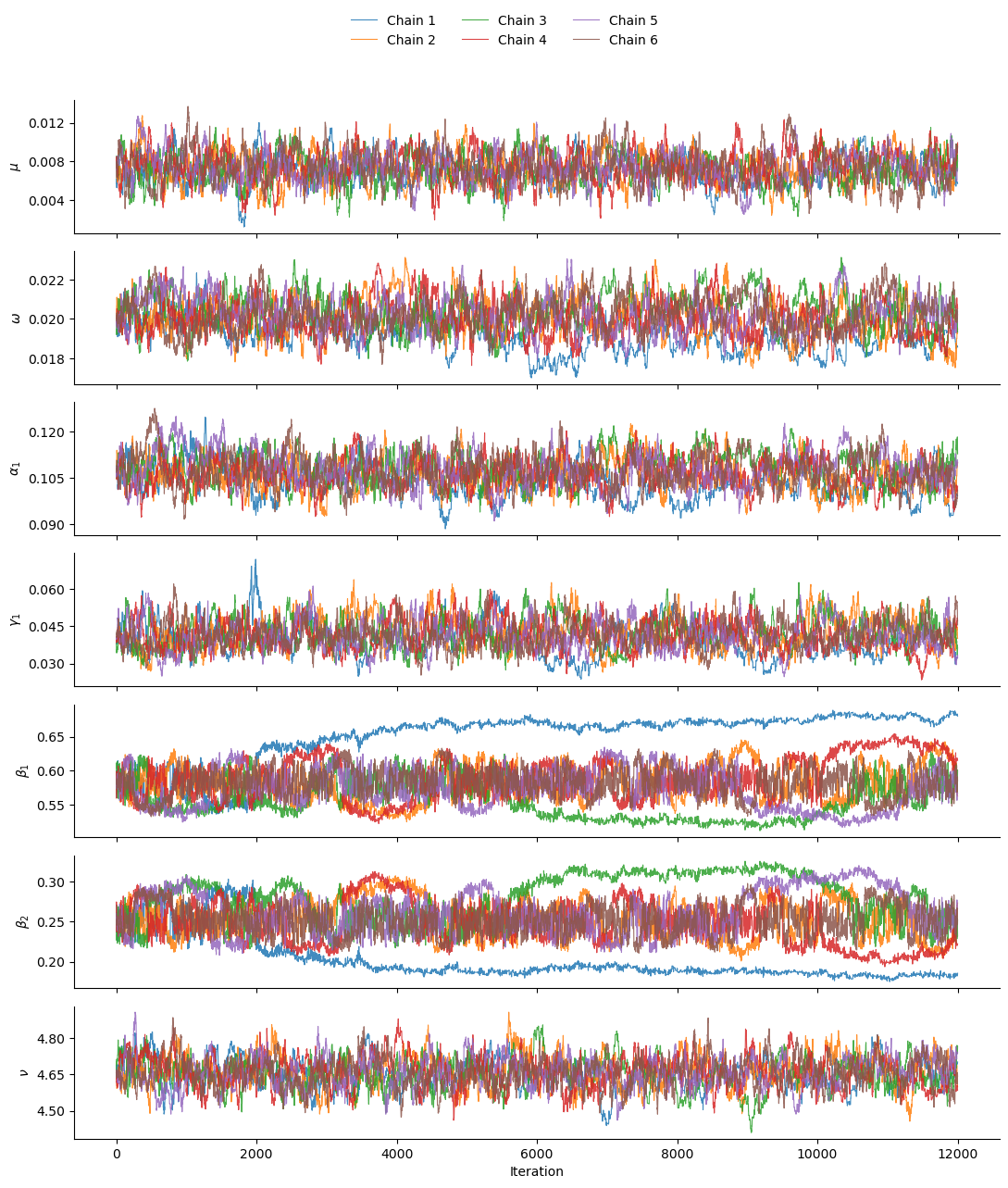}
    \caption{Results for the TGARCH($1,2$) model with Student-$t$ errors fitted to the AAPL data. Traceplots for the $n_{\mathrm{rep}}=6$ chains for each parameter.}
    \label{fig:TGARCH12_traceplots_app2}
\end{figure}

Returning to computational cost, Table \ref{tab:mcmc_behaviour_App2} shows that subsampling MCMC with truncated power-law decaying weights yields substantial computational gains, with median speed-ups of around $48$--$66\times$ in CPU time (including tuning and control variate construction) while using only $1.8$--$2.4\%$ of the log-density evaluations required by full-data MCMC. Figure \ref{fig:u_max_app2} further illustrates the behaviour of the largest subsample index $u_{\max}$ over the course of the Markov chain, confirming that the computational cost remains well below that of full-data MCMC for all models, in line with the theoretical analysis.

\begin{figure}[H]
    \centering \includegraphics[width=0.95\textwidth]{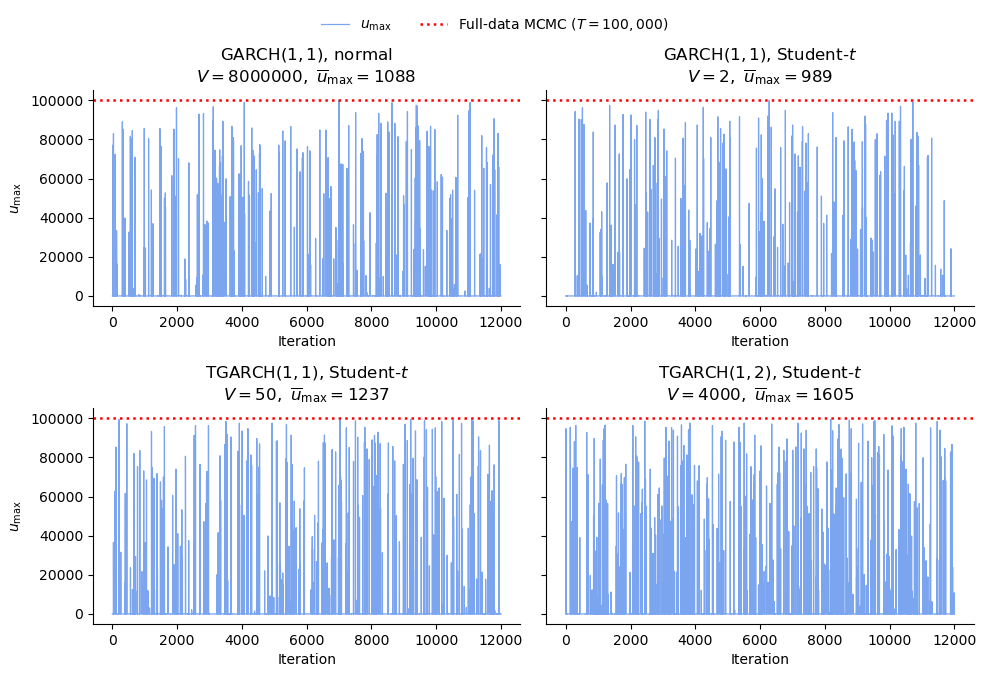}
    \caption{Results for models fitted to the AAPL data. Largest subsample index $u_{\max}$ (blue line) versus iterations for subsampling Markov chain Monte Carlo (MCMC) with truncated power-law decaying (TPD) sampling probabilities, tuned with $R_{\max}=100$, corresponding to the variance tolerances $V$ shown in the panel titles. The figure also shows the full-data MCMC cost as benchmark (red dotted line). The quantity $\overline{u}_{\max}$ denotes the average of $u_{\max}$ over the MCMC iterations.}
    \label{fig:u_max_app2}
\end{figure}

We next assess the accuracy of the subsampling-based posterior approximations by comparing them to the corresponding full-data posteriors. We focus on the three models that do not exhibit the pronounced ridge-like between-chain behaviour observed for the TGARCH($1,2$) model. For the TGARCH($1,1$) model with Student-$t$ errors, one of the six subsampling chains exhibits the prolonged immobility documented above. The representative run used for the posterior comparison does not exhibit this behaviour. For the TGARCH($1,2$) model, differences across posterior approximations may arise from chains exploring distinct ridge-like regions of the parameter space, rather than from the subsampling procedure itself. We therefore exclude this model from the accuracy comparison. Figures \ref{fig:App2_MCMC_GARCH_11_norm_kde}--\ref{fig:App2_MCMC_TGARCH_11_Student-t_kde} show kernel density estimates of the marginal posteriors obtained using subsampling MCMC and full-data MCMC.

\begin{figure}[H]
    \centering    \includegraphics[width=0.8\textwidth]{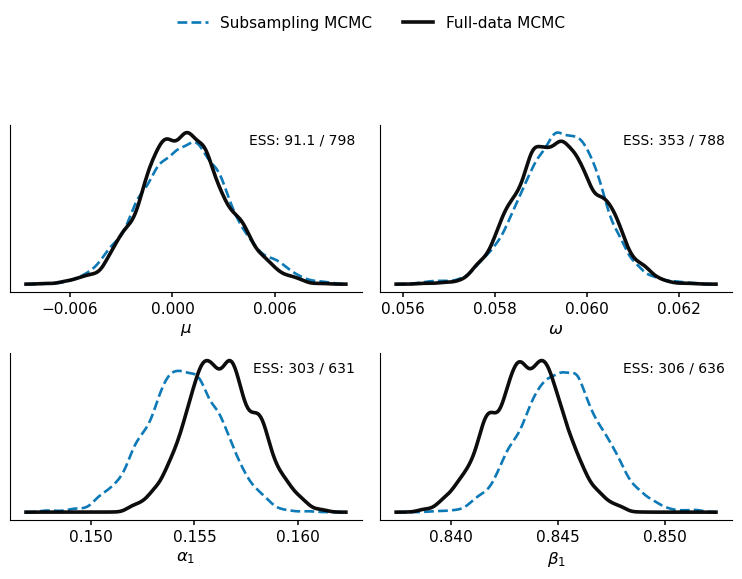}
    \caption{Posterior marginal distributions under the original parameterisation $\boldsymbol{\theta}$ for the GARCH($1,1$) model with normal errors fitted to the AAPL data. The subsampling Markov chain Monte Carlo (MCMC) method with truncated power-law decaying (TPD) sampling probabilities is tuned with $R_{\max}=100$, corresponding to a variance tolerance $V\approx 8{,}000{,}000$. The solid black curves represent the full-data MCMC posterior for reference. Each panel also reports the effective sample size (ESS), based on $10{,}000$ post-burn-in samples, with values shown for subsampling MCMC (first) and full-data MCMC (second).}
    \label{fig:App2_MCMC_GARCH_11_norm_kde}
\end{figure}

Across all three models, the two approaches are in close agreement, indicating that the subsampling-based method provides an accurate approximation to the full-data posterior. While the overall agreement is very strong, a small discrepancy (note the scale) can be observed for $\alpha_1$ and $\beta_1$ in Figure \ref{fig:App2_MCMC_GARCH_11_norm_kde}. This corresponds to the setting with the largest tuned value of the variance tolerance, $V \approx 8{,}000{,}000$, which is exceptionally large. In addition, subsampling MCMC yields a more faithful representation of posterior uncertainty than the variational approximations in Section \ref{sec:VB_application}.

\begin{figure}[H]
    \centering    \includegraphics[width=0.9\textwidth]{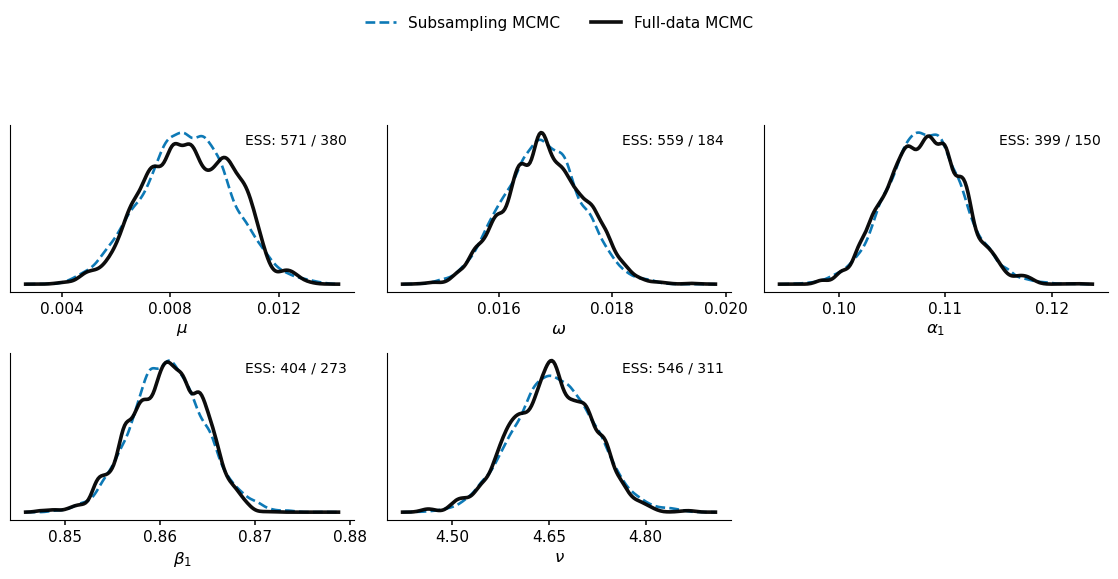}
    \caption{Posterior marginal distributions under the original parameterisation $\boldsymbol{\theta}$ for the GARCH($1,1$) model with Student-$t$ errors fitted to the AAPL data. The subsampling Markov chain Monte Carlo (MCMC) method with truncated power-law decaying (TPD) sampling probabilities is tuned with $R_{\max}=100$, corresponding to a variance tolerance $V\approx 2$. The solid black curves represent the full-data MCMC posterior for reference. Each panel also reports the effective sample size (ESS), based on $10{,}000$ post-burn-in samples, with values shown for subsampling MCMC (first) and full-data MCMC (second).}
    \label{fig:App2_MCMC_GARCH_11_Student-t_kde}
\end{figure}

\begin{figure}[h]
    \centering    \includegraphics[width=0.9\textwidth]{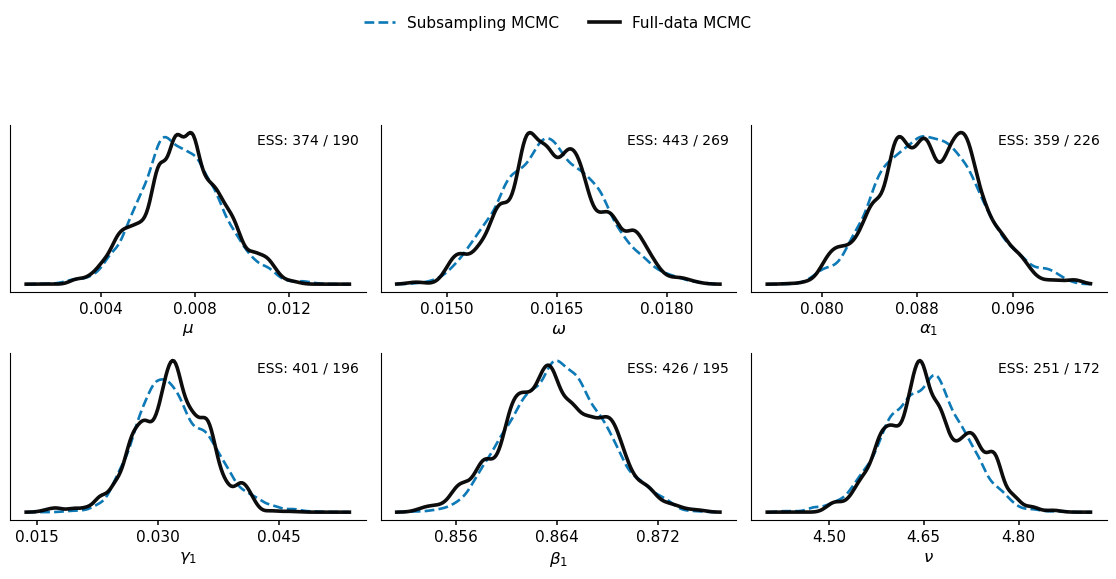}
    \caption{Posterior marginal distributions under the original parameterisation $\boldsymbol{\theta}$ for the TGARCH($1,1$) model with Student-$t$ errors fitted to the AAPL data. The subsampling Markov chain Monte Carlo (MCMC) method with truncated power-law decaying (TPD) sampling probabilities is tuned with $R_{\max}=100$, corresponding to a variance tolerance $V\approx 50$. The solid black curves represent the full-data MCMC posterior for reference. Each panel also reports the effective sample size (ESS), based on $10{,}000$ post-burn-in samples, with values shown for subsampling MCMC (first) and full-data MCMC (second).}
    \label{fig:App2_MCMC_TGARCH_11_Student-t_kde}
\end{figure}

Figures \ref{fig:App2_MCMC_GARCH_11_norm_kde}--\ref{fig:App2_MCMC_TGARCH_11_Student-t_kde} also report the effective sample sizes (ESS) for each approach. For the representative runs shown, subsampling MCMC yields higher ESS than full-data MCMC in two of the three models. The difference is likely attributable to the use of different proposal mechanisms. In particular, we observe that it is comparatively easier to construct an efficient random-walk proposal for the model parameters in the extended pseudo-marginal formulation than for the corresponding full-data formulation with an exact likelihood. This phenomenon has not been reported in previous subsampling MCMC studies, likely because the models considered there are substantially simpler than the GARCH and threshold variants analysed here. Such models can typically be sampled efficiently using simple preconditioned random-walk proposals. This explains our experimental choice of a simple preconditioned random-walk proposal for the subsampling MCMC, whereas a more sophisticated adaptive proposal is required for the full-data MCMC. We conjecture that the noise in the likelihood estimator may partially smooth ridges in the GARCH posterior, thereby making the resulting posterior easier to explore using a simple random-walk proposal. We emphasise that this effect is not uniform across models; for example, in the TGARCH($1,2$) case, the ridge structure remains difficult to traverse, as discussed above (see Figure \ref{fig:TGARCH12_traceplots_app2}).

 To address the empirical goal of this application, namely identifying the model with the best out-of-sample performance, Table \ref{tab:sub_sampling_mcmc_LPDS} reports the log predictive density scores and the corresponding confidence intervals for both subsampling and full-data MCMC. The two approaches yield very similar predictive scores across the first three models, leading to the same model ranking. The most substantial improvement is observed when a Student-$t$ error is used in place of a normal error, which is likely explained by the presence of pronounced outliers in Figure \ref{fig:AppleInc}. There is also clear evidence of a threshold effect, indicating an asymmetric impact of shocks on volatility. For the TGARCH($1,2$) model, full-data MCMC suggests an improvement in predictive performance relative to TGARCH($1,1$). However, as discussed above, posterior exploration for this model is more challenging due to the presence of ridge-like structures, which can affect the reliability of the resulting estimates. In contrast, subsampling MCMC yields predictive scores that are very close to those of TGARCH($1,1$). Taking these considerations into account, and in the interest of parsimony, we adopt the TGARCH($1,1$) model with Student-$t$ errors as our preferred specification for the AAPL data.

\begin{table}[H]
\centering
\small
\setlength{\tabcolsep}{4pt}
\caption{Results for models fitted to the AAPL data. The table shows the estimated log predictive density score (LPDS) for both subsampling Markov chain Monte Carlo (MCMC) and full-data MCMC, together with $95\%$ approximate confidence intervals computed via a delta-method approximation, where the variance in the approximation accounts for the integrated autocorrelation time. All results are based on $100$ posterior samples, with subsampling MCMC tuned using $R_{\max}=100$. The table also reports the effective sample size (ESS).
}
\label{tab:sub_sampling_mcmc_LPDS}
\begin{tabular}{llcc}
\toprule
Model and error & Method & $\widehat{\mathrm{LPDS}}$ & ESS \\
\midrule
GARCH($1,1$), normal & Subsampling MCMC & \makecell{-14630.21 \\ {(-14631.14, -14629.29)}} & 43 \\
 & Full-data MCMC & \makecell{-14631.27 \\ {(-14631.75, -14630.80)}} & 95 \\
\addlinespace
GARCH($1,1$), Student-$t$ & Subsampling MCMC & \makecell{-9226.06 \\ {(-9227.17, -9224.95)}} & 100 \\
 & Full-data MCMC & \makecell{-9225.90 \\ {(-9226.65, -9225.16)}} & 100 \\
\addlinespace
TGARCH($1,1$), Student-$t$ & Subsampling MCMC & \makecell{-9205.15 \\ {(-9206.56, -9203.73)}} & 100 \\
 & Full-data MCMC & \makecell{-9205.79 \\ {(-9206.77, -9204.81)}} & 100 \\
\addlinespace
TGARCH($1,2$), Student-$t$ & Subsampling MCMC & \makecell{-9203.31 \\ {(-9204.47, -9202.16)}} & 100 \\
 & Full-data MCMC & \makecell{-9201.88 \\ {(-9203.51, -9200.26)}} & 100 \\
\addlinespace
\bottomrule
\end{tabular}
\end{table}

\subsection{Limitations of uniform subsampling in the tuning regime for subsampling MCMC with recursive likelihood models}\label{subsec:limitations_uniform_recursive}

This section presents a striking empirical finding: in the tuning regime with small $m$ (to keep computational cost low), uniform subsampling, despite exhibiting provably lower variance in the log-likelihood estimator, can lead to substantially inferior subsampling MCMC performance compared to truncated power-law decaying (TPD) sampling in recursive likelihood settings. We illustrate this using the GARCH(1,1) model with Student-$t$ errors fitted to the AAPL data above. Increasing the subsample size $m$ stabilises the mixing of the subsampling MCMC chain (not shown); however, the expected computational cost quickly approaches that of full-data evaluation and is therefore not of practical interest.

We consider $R_{\max}=100$ for TPD and $R_{\max}=1$ for uniform with $m=2$, resulting in variance tolerances $V\approx 2$ and $V \approx 0.02$, respectively. The failure of uniform is counterintuitive from a variance-based perspective; however, examining both the distributional properties of the estimator and the structure of its bias-corrected form provides further insight. Figure \ref{fig:bulk_empirical1} shows the bulk of the distribution of the log-likelihood error and related quantities evaluated at the tuning reference value $\boldsymbol{\phi}^{\dagger}$ in \eqref{eq:phi_tune}, for the bias-corrected estimator in \eqref{eq:bias-corrected-likelihood-estimator} under the different sampling schemes. The behaviour in the tails is complemented by Table \ref{tab:bulk_empirical1}, which reports additional distributional summaries. 

\begin{figure}[H]
    \centering
    \includegraphics[width=1\textwidth]{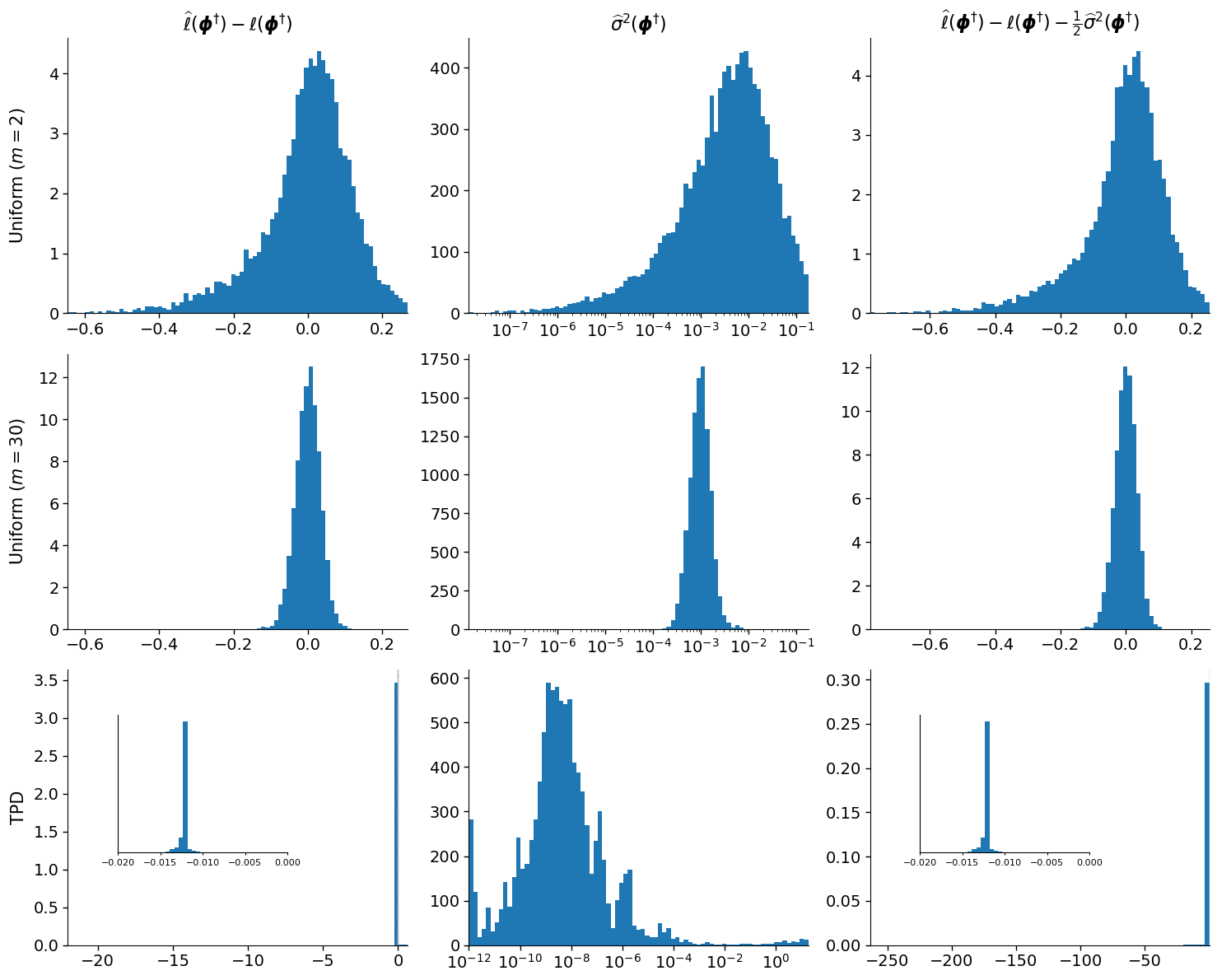}
    \caption{Results for the GARCH(1, 1) model with Student-$t$ errors fitted to the AAPL data. The figure displays the bulk of the empirical distributions of the log-likelihood error (left), the variance estimator (middle), and the bias-corrected error in \eqref{eq:bias-corrected-likelihood-estimator} (right), evaluated at $\boldsymbol{\phi}^{\dagger}$ in \eqref{eq:phi_tune}, under uniform sampling ($m=2$ top row, $m=30$ middle row) and truncated power-law decaying (TPD) sampling (bottom row). The inset panels provide a zoom of the central region.}
    \label{fig:bulk_empirical1}
\end{figure}

Figure \ref{fig:bulk_empirical1} shows that uniform sampling with $m=2$ produces a more dispersed distribution of the log-likelihood error $\mathrm{err}(\boldsymbol{\phi}^{\dagger})=\widehat{\ell}(\boldsymbol{\phi}^{\dagger})-\ell(\boldsymbol{\phi}^{\dagger})$, whereas the TPD scheme yields errors that are tightly concentrated within the bulk. This is also reflected in Table \ref{tab:bulk_empirical1}, where the interquantile range under uniform sampling with $m=2$ is substantially wider than under TPD, despite the latter having a larger variance. A more consequential difference concerns the sign and asymmetry of the error distribution. By unbiasedness, the expectation of the error is zero. Under uniform sampling, the empirical error distribution is approximately centred around zero, implying substantial probability mass on both sides of zero, including on the positive side; see Table \ref{tab:bulk_empirical1}. 

\begin{table}[H]
\centering
\small
\setlength{\tabcolsep}{5pt}
\caption{Results for the GARCH(1, 1) model with Student-$t$ errors fitted to the AAPL data. Summary statistics of the empirical distributions (complementing Figure \ref{fig:bulk_empirical1}) evaluated at $\boldsymbol{\phi}^{\dagger}$ in \eqref{eq:phi_tune}.}
\begin{tabular}{lccccccccc}
\toprule
 & \multicolumn{7}{c}{$\mathrm{err}(\boldsymbol{\phi}^{\dagger})=\widehat{\ell}(\boldsymbol{\phi}^{\dagger})-\ell(\boldsymbol{\phi}^{\dagger})$} & \multicolumn{2}{c}{$\widehat{\sigma}^2(\boldsymbol{\phi}^{\dagger})$}
 \\
\cmidrule(lr){2-8} \cmidrule(lr){9-10}
 & min & 0.01 & 0.05 & 0.95 & 0.99 & max & $\Pr(\mathrm{err}>0)$ & median & max \\
\midrule
Uniform &  &  &  &  &  &  &  &  &  \\
$\,\,m=2$ & -0.793 & -0.416 & -0.245 & 0.174 & 0.267 & 1.03 & 0.557 & 0.004 & 1.17 \\
$\,\,m=30$ & -0.135 & -0.079 & -0.055 & 0.052 & 0.075 & 0.121 & 0.514 & 0.001 & 0.008 \\
TPD & -38.2 & -0.030 & -0.013 & -0.012 & 0.608 & 19.7 & 0.014 & 0.000 & 1460 \\
\bottomrule
\end{tabular}
\label{tab:bulk_empirical1}
\end{table}

In contrast, the TPD scheme produces predominantly negative errors, with a distribution that is concentrated near zero but can generate occasional large tail events in both directions. When $m=30$ (which is not optimal in a recursive likelihood setting) for uniform sampling, the behaviour changes markedly. As expected from the central limit theorem, the sampling distribution becomes more symmetric around zero and substantially more concentrated; see Figure \ref{fig:bulk_empirical1} and Table \ref{tab:bulk_empirical1}. Consequently, the interquantile range is much closer to that of TPD, and large positive errors are effectively absent in this regime. 

Positive errors are particularly problematic in the pseudo-marginal setting, since they correspond to regions with artificially inflated posterior mass and can induce severe sticky behaviour in the Markov chain. The bias-correction term $-\widehat{\sigma}^2/2$ mitigates this effect by penalising large positive deviations; however, its impact differs markedly across the two schemes. Figure \ref{fig:scatter_error_sigma2_hat_empirical1} shows that, under TPD, large positive errors are typically accompanied by substantial corrections and are therefore strongly penalised, whereas under uniform sampling with $m=2$ positive deviations are only weakly corrected. 

\begin{figure}[H]
    \centering
    \includegraphics[width=0.8\textwidth]{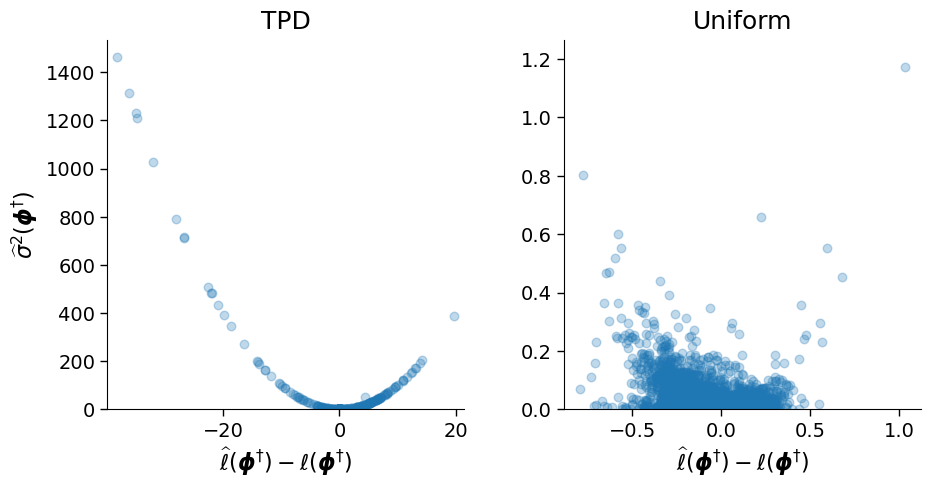}
    \caption{Results for the GARCH(1, 1) model with Student-$t$ errors fitted to the AAPL data. Scatter plots of log-likelihood errors $\mathrm{err}(\boldsymbol{\phi}^{\dagger})=\widehat{\ell}(\boldsymbol{\phi}^{\dagger})-\ell(\boldsymbol{\phi}^{\dagger})$ against variance estimates $\widehat{\sigma}^2(\boldsymbol{\phi}^{\dagger})$, evaluated at $\boldsymbol{\phi}^{\dagger}$ in \eqref{eq:phi_tune},  for truncated power-law decaying (TPD) sampling (left) and uniform sampling (right).}
    \label{fig:scatter_error_sigma2_hat_empirical1}
\end{figure}

This is further supported by Table \ref{tab:sigma_dependence}, which shows that overestimation is attenuated to a much greater extent under TPD than under uniform sampling with $m=2$. 

\begin{table}[H]
\centering
\caption{Results for the GARCH(1, 1) model with Student-$t$ errors fitted to the AAPL data. 
Summary statistics of the relationship between the log-likelihood errors $\mathrm{err}=\widehat{\ell}(\boldsymbol{\phi}^{\dagger})-\ell(\boldsymbol{\phi}^{\dagger})$ and the variance estimates $\widehat{\sigma}^2(\boldsymbol{\phi}^{\dagger})$ (complementing Figure \ref{fig:scatter_error_sigma2_hat_empirical1}) evaluated at $\boldsymbol{\phi}^{\dagger}$ in \eqref{eq:phi_tune}. 
}
\begin{tabular}{lcccc}
\toprule
Scheme 
& $\mathrm{Corr}\left(|\mathrm{err}|,\widehat{\sigma}^2\right)$ 
& $\mathbb{E}\left(\widehat{\sigma}^2 | \mathrm{err} > 0\right)$ 
& $\mathbb{E}\left(\widehat{\sigma}^2 | \mathrm{err} \le 0\right)$
& $\Pr(\mathrm{err} > 0)$ \\
\midrule
TPD 
& 0.895 
& 32.14 
& 1.40 
& 0.014 \\
Uniform 
& 0.546 
& 0.009 
& 0.026 
& 0.557 \\
\bottomrule
\end{tabular}\label{tab:sigma_dependence}
\end{table}

To summarise, the key difference is that under TPD, large positive errors are both rare and strongly penalised, whereas under uniform sampling with $m=2$, positive errors are common and only weakly attenuated. We emphasise that increasing the subsample size (e.g., $m=30$) removes problematic positive errors, but at a  computational cost that renders a subsampling approach effectively obsolete. The experiment above explains why we were unable to obtain usable pseudo-marginal chains under uniform sampling with $m=2$ in any of the examples in Section \ref{subsec:results_subsampling}. It also sheds light on why the TPD scheme can perform well even under extremely large variance tolerances, for example, in the GARCH(1, 1) with normal errors in Section \ref{subsec:results_subsampling}, with $V=8{,}000{,}000$.

\subsection{Conclusion and future research}

We consider an additional application of our methodology based on a posterior sampling approach (Markov chain Monte Carlo), complementing the optimisation-based approach (variational Bayes) framework in the main paper. Specifically, we consider subsampling MCMC to estimate GARCH models and a threshold extension thereof for modelling conditional volatility in the Apple Inc log-returns, with the aim of model selection via log predictive density scores.

The dataset in this application exhibits substantially higher persistence than that considered in the main paper, which introduces additional challenges for sampling. To address this, we construct the proposal in a reparameterisation that enforces positivity and stationarity constraints.

Our experiments show that the stabilised weighted subsampling approach can deliver substantial computational gains while maintaining accurate posterior inference in challenging recursive likelihood settings. While the variational Bayes approach considered in the main paper also offers substantially faster inference than its full-data counterpart, it provides a less accurate approximation of the posterior, particularly in terms of uncertainty quantification. In contrast, subsampling MCMC yields a more faithful characterisation of posterior uncertainty.

Moreover, we find that uniform sampling in the tuning regime for recursive likelihoods, which encourages a small subsample $m$ for computational reasons, renders the pseudo-marginal algorithm ineffective, as the resulting chains become stuck. We develop an empirical understanding of this behaviour. This stands in contrast to the variational Bayes setting, where uniform subsampling remains viable, albeit often producing poorer approximations and only modest computational gains; see Section \ref{sec:VB_application}.

We conclude by noting that the insights from the experiment in Section \ref{subsec:limitations_uniform_recursive} suggest that variance alone is not sufficient to characterise the behaviour of pseudo-marginal approaches in subsampling MCMC, such as \cite{quiroz2019speeding}. Previous work on tuning pseudo-marginal algorithms \citep{pitt2012some, doucet2015efficient} proposes variance targets in the range $1$-–$3$. These recommendations are derived by linking the variance of the estimator to the objective of interest under specific modelling and algorithmic assumptions. In particular, they rely on settings in which the estimator is approximately normally distributed and its variance is inversely proportional to computational effort. In our setting, however, these assumptions are not satisfied. The expected computational cost does not admit such a direct relationship with estimator variance, and the distribution of the estimator exhibits clear deviations from normality in our tuning regime, including asymmetry and heavy tails. As a result, variance alone does not adequately capture the behaviour of the pseudo-marginal algorithm. Developing pseudo-marginal tuning strategies tailored to the TPD scheme remains an open problem, which we leave for future research.

\newpage

\section{Comparison with stochastic gradient Markov chain Monte Carlo}\label{sec:comparison_SG_MCMC}

\subsection{Background, model, and data}\label{subsec:data_SG_MCMC}
This section compares the subsampling Markov chain Monte Carlo framework developed in \cite{quiroz2019speeding}, extended with the truncated sampling probabilities proposed in this paper, with stochastic gradient Markov chain Monte Carlo (SG-MCMC) methods \citep{chen2014stochastic,nemeth2021stochastic}, recently extended to general state space models by \cite{aicher2025stochastic}. SG-MCMC methods follow the inferential speed-up paradigm considered here, as they use data subsamples to construct stochastic approximations of full-data posterior quantities. However, they differ fundamentally in how this approximation is introduced, replacing exact gradient evaluations with noisy gradient estimates within the Markov chain dynamics and typically omitting a Metropolis--Hastings correction.

The SG-MCMC method in \cite{aicher2025stochastic} is developed for non-linear state space models and provides estimates of the static model parameters and the latent states. In contrast, our recursive likelihood formulation is more restricted, as it focuses on static parameter inference. Specifically, it applies to models for which the likelihood admits a recursive representation, including both models arising from latent-variable formulations with tractable marginal likelihoods (e.g., linear state space models via the Kalman filter) and models defined directly through recursive likelihoods, such as GARCH-type models. However, our method does not provide estimates of the latent states. We therefore restrict the comparison to the estimation of static parameters.

\cite{aicher2025stochastic} consider a latent-variable formulation of GARCH($1,1$), which they term GARCH with noise model. Our specification does not include this noise component; we reconcile the two formulations in Section \ref{subsec:reconcile_models}. In particular, we show that as the noise component tends to zero, their model coincides with ours. We use simulated data generated from the latent-variable GARCH formulation, with GARCH parameters calibrated to those obtained from fitting the model to the Dow Jones Industrial Average data in Section \ref{sec:VB_application}. The unconditional mean is set to zero. The parameter controlling the noise component is set to a sufficiently small level so that the noise-free GARCH full-data posterior is practically indistinguishable between datasets simulated with and without the latent noise component, holding all other sources of randomness fixed. This justifies using the noise-free full-data posterior as a reliable proxy for the corresponding posterior under the GARCH with noise likelihood, which is not computed in \cite{aicher2025stochastic}.

\subsection{Reconciling the model formulations}\label{subsec:reconcile_models}
The model formulation in \cite{aicher2025stochastic} differs from ours in that they obtain a GARCH specification from a latent state space representation, which they refer to as a GARCH with noise model, whereas we formulate it directly as a conditional variance model. We now reconcile the two model formulations to enable a direct comparison of the corresponding sampling algorithms using existing implementations of both approaches\footnote{SG-MCMC code available at \url{https://github.com/aicherc/sgmcmc_ssm_code/tree/master/sgmcmc_ssm}.}.

\cite{aicher2025stochastic} formulate a GARCH($1,1$) model with normal errors as a latent state space model. To connect their specification to ours, we first state our formulation as
\begin{align}\label{eq:GARCH11norm}
    y_t | \mathcal{F}_{t-1} & \sim \mathcal{N}(\mu,\sigma^2_t), \nonumber \\
    \sigma^2_t & = \omega + \alpha_1(y_{t-1}-\mu)^2 + \beta_1\sigma^2_{t-1}.
\end{align}
\cite{aicher2025stochastic} instead introduce a latent state $x_t$ and specify their model as
\begin{align}\label{eq:GARCH11norm_Aicher}
    y_t | x_t & \sim \mathcal{N}(x_t, \tau^2), \nonumber \\
    x_t | x_{t-1} &\sim \mathcal{N}(0,\sigma_t^2)  \nonumber \\
    \sigma^2_t & = \omega + \alpha_1x_{t-1}^2 + \beta_1\sigma^2_{t-1}. 
\end{align} 
The latent state $x_t$ corresponds to the centred return $y_t-\mu$ in our formulation. Integrating out $x_t$ in \eqref{eq:GARCH11norm_Aicher} yields the marginal observation model
\begin{align}\label{eq:GARCH11norm_Aicher_collapsed}
    y_t | \mathcal{F}_{t-1} & \sim \mathcal{N}(0, \sigma_t^2 + \tau^2), \nonumber \\
    \sigma^2_t & = \omega + \alpha_1x_{t-1}^2 + \beta_1\sigma^2_{t-1}, 
\end{align}
where $\mathcal{F}_{t-1}$ denotes the information set under the state space formulation. When $\tau^2\rightarrow 0$, the state space formulation reduces to the standard GARCH specification in \eqref{eq:GARCH11norm}, and the information sets coincide.

The connection with our formulation is immediate. Both models share the same conditional variance recursion, while the state space formulation introduces an additive observation variance component $\tau^2$. In the limit $\tau^2\rightarrow0$, the marginal observation model in \eqref{eq:GARCH11norm_Aicher_collapsed} reduces to $$y_t|\mathcal{F}_{t-1}\sim \mathcal{N}(0,\sigma^2_t),$$
which coincides with our specification in \eqref{eq:GARCH11norm} when $\mu=0$. In practical applications, when the posterior samples of $\tau^2$ are small relative to the unconditional variance of the process $\mathbb{E}(\sigma^2_t)$, the marginalised model in \eqref{eq:GARCH11norm_Aicher_collapsed} closely approximates the standard conditional GARCH model in \eqref{eq:GARCH11norm}, and posterior inference on $\omega,\alpha_1,\beta_1$ is therefore approximately comparable across the two model specifications.

To obtain a full-data posterior benchmark for comparison with \citet{aicher2025stochastic}, we exploit that the noise-free GARCH posterior is insensitive to small levels of latent noise in the data-generating process. Specifically, we simulate datasets from \eqref{eq:GARCH11norm_Aicher} for varying values of the noise component $\tau^2$, while holding all other sources of randomness fixed, and perform inference using the full-data posterior based on the (noise free) likelihood in \eqref{eq:GARCH11norm}. Figure \ref{fig:posterior_3tau2} shows that, for small noise levels, the posterior distributions under the noise-free GARCH specification are practically indistinguishable, while larger noise levels lead to visible deviations. This justifies using the noise-free full-data posterior as a ground truth benchmark for the comparison.

\begin{figure}[h]
    \centering
    \includegraphics[width=0.98\textwidth]{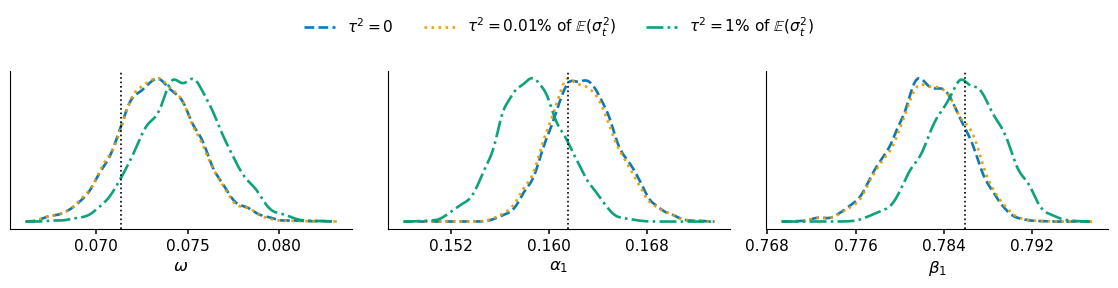}
    \caption{Marginal posterior distributions for the GARCH($1,1$) parameters obtained from the full-data posterior based on the likelihood in \eqref{eq:GARCH11norm}, using three simulated datasets. The datasets differ only in the noise component $\tau^2$, which is specified as a percentage of the unconditional variance of the observations $y_t$, i.e., $\mathbb{E}(\sigma^2_t)$. The vertical lines indicate the true parameter values.}
    \label{fig:posterior_3tau2}
\end{figure}

Finally, although the priors specified in \cite{aicher2025stochastic} are different from ours, the large sample size in the comparison ($T=100{,}000$) ensures that posterior inference is driven primarily by the likelihood, rendering the prior specification practically negligible.

\subsection{Specifications and settings}\label{subsec:specifications_SG_MCMC}
We use the dataset corresponding to the small noise level illustrated in Figure \ref{fig:posterior_3tau2}, i.e.,\ $\tau^2 = 0.01\%$ of the unconditional variance $\mathbb{E}(\sigma^2_t)=\omega/(1-\alpha_1-\beta_1)$. We use the same prior specifications as in Sections \ref{sec:VB_application} and \ref{sec:MCMC_application}, however, with an identity transformation for all parameters.

We run subsampling MCMC with truncated sampling probabilities determined by the tuning procedure in Section \ref{subsec:balancing_variance_and_computational_cost} with $R_{\max}=100$. We use a standard random-walk proposal preconditioned by a scaled covariance matrix from a Laplace approximation. For stochastic gradient MCMC, we use the implementation and default tuning choices provided by \citet{aicher2025stochastic}, fixing the subsequence length to $40$, buffer length to $10$, and number of particles to $1{,}000$. Their implementation uses the particle filtering methodology in \cite{poyiadjis2011particle}. Following the experiments in \cite{aicher2025stochastic}, we vary only the stochastic gradient Langevin dynamics step size $h$ over the grid $\{1, 0.1, 0.01, 0.001\}$. To assess the accuracy of the resulting posterior distributions, we also run a full-data Markov chain Monte Carlo algorithm using the same settings as in Section \ref{subsec:spec_settings_App1}, including an adaptive proposal. All chains are initialised at the MAP estimate and run for $10{,}000$ samples after a burn-in period of $2{,}000$ iterations.
\subsection{Results}
Figure \ref{fig:kde_SG_MCMC_subsampling_full_data} shows kernel density estimates from subsampling MCMC and SG-MCMC for several values of $h$, together with full-data MCMC estimates as a ground truth benchmark. Subsampling MCMC provides accurate posterior approximations for all parameters. In contrast, SG-MCMC does not accurately capture the posterior distribution in this setting. 

\begin{figure}[h]
    \centering
    \includegraphics[width=0.98\textwidth]{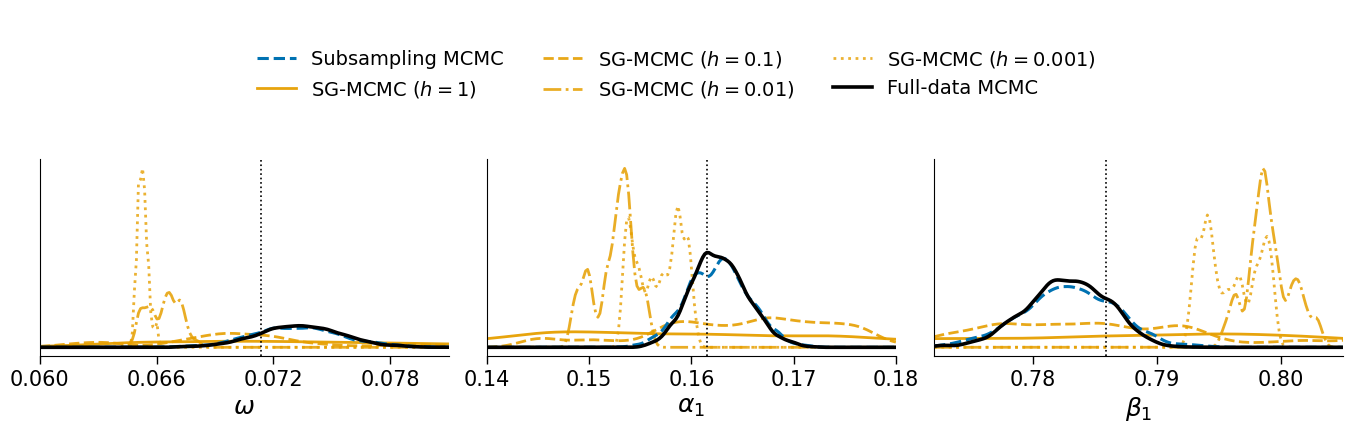}
    \caption{Posterior marginal distributions for the GARCH($1,1$) model fitted to simulated data calibrated to the Dow Jones Industrial Average data. The subsampling Markov chain Monte Carlo (MCMC) method with truncated decaying sampling probabilities is shown for $R_{\max}=100$, corresponding to a variance tolerance $V\approx 0.20$. Stochastic gradient MCMC (SG-MCMC) is shown for several values of the step size $h$. The solid black curves represent the full-data MCMC posterior for reference, and the vertical lines indicate the true parameter values.}
    \label{fig:kde_SG_MCMC_subsampling_full_data}
\end{figure}

Finally, we follow \cite{aicher2025stochastic} and assess sample quality using the kernel Stein discrepancy (KSD; \citealp{gorham2017measuring}). The discrepancy is estimated using the same implementation, employing the inverse multiquadratic kernel. We use $500$ samples obtained by thinning the post burn-in output (burn-in of $2{,}000$ iterations, thinning factor $20$). For all methods, the score function is estimated following \cite{aicher2025stochastic}, using a particle-based approximation with $10{,}000$ particles, subsequence length $1{,}000$, and buffer length $10$. This score is based on the latent-variable formulation and therefore does not align exactly with our model formulation. As such, the evaluation is, if anything, favourable to the SG-MCMC approach. Table \ref{tab:SG_MCMC_vs_subsampling} shows the estimated KSD values for each method, together with the running time to generate $12{,}000$ iterations. As expected, full-data MCMC provides the highest sample quality, but at a substantially higher computational cost. Subsampling MCMC achieves KSD values very close to those of full-data MCMC, while reducing computational cost by approximately a factor of $76$. The SG-MCMC methods are substantially faster still, requiring only around one quarter of the computational time of subsampling MCMC, but at the expense of noticeably larger KSD values, indicating a poorer approximation of the target posterior (as also illustrated in Figure \ref{fig:kde_SG_MCMC_subsampling_full_data}).

\begin{table}[H]
\centering
\caption{Comparison of sample quality measured using the kernel Stein discrepancy (KSD) on the $\log_{10}$ scale. The log$_{10}$ KSD is reported as mean (with standard deviation in parentheses) and range over $n_{\mathrm{rep}} = 6$ independent evaluations, using common random seeds across methods. The table also shows the computational cost of running $12{,}000$ iterations of each algorithm, measured as wall-clock time (CPU, in seconds). Methods are ordered by increasing $\log_{10}$ KSD (lower values indicate higher sample quality).}
\begin{tabular}{lccc}
\toprule
Method & log$_{10}$ KSD & log$_{10}$ range & CPU \\
\midrule
Full-data MCMC & \makecell{1.055 \\ (0.280)} & [0.534, 1.398] & 69750.5 \\
Subsampling MCMC & \makecell{1.101 \\ (0.216)} & [0.695, 1.369] & 922.7 \\
SG-MCMC $h = $0.1 & \makecell{1.513 \\ (0.137)} & [1.340, 1.722] & 232.5 \\
SG-MCMC $h = $1 & \makecell{1.619 \\ (0.103)} & [1.521, 1.784] & 235.7 \\
SG-MCMC $h = $0.001 & \makecell{1.738 \\ (0.081)} & [1.616, 1.841] & 233.2 \\
SG-MCMC $h = $0.01 & \makecell{1.950 \\ (0.054)} & [1.882, 2.024] & 232.9 \\
\bottomrule
\end{tabular}\label{tab:SG_MCMC_vs_subsampling}

\end{table}

Overall, these results indicate that subsampling MCMC provides an accurate and robust approximation to the full-data posterior, while SG-MCMC, although faster, is sensitive to tuning and does not reliably approximate the target posterior in this setting.

\clearpage

%\newpage

\section{Comparison with divide-and-conquer Markov chain Monte Carlo}\label{sec:comparison_divide-and-conquer_MCMC}

\subsection{Background, model, and data}

This section compares the subsampling Markov chain Monte Carlo framework developed in \cite{quiroz2019speeding}, extended with the truncated sampling probabilities proposed in this paper, with divide-and-conquer Markov chain Monte Carlo methods \citep{scott2017consensus, scott2022bayes}. Divide-and-conquer methods were recently extended to time series data by \cite{ou2025scalable} through the divide-and-conquer Bayesian time series (DC-BATS) algorithm; see also \cite{wang2026spectral} for an approach based on the frequency domain. Divide-and-conquer methods follow a substantially different inferential speed-up paradigm. Rather than using data subsamples to construct stochastic approximations of full-data posterior quantities, the data are partitioned across shards, where a tempered posterior is fitted independently on each shard, and the resulting draws are subsequently merged to approximate samples from the full-data posterior. Under suitable regularity conditions, \cite{ou2025scalable} establish that the resulting approximation is asymptotically consistent as the shard size $b = T/K$ increases, with the approximation error decreasing at a rate of order $\mathcal{O}(1/b)$, where $T$ denotes the total number of observations and $K$ the number of shards.

To ensure comparability with the stochastic-gradient Markov chain Monte Carlo example in Section \ref{sec:comparison_SG_MCMC}, we first use the same model and dataset, that is, a GARCH($1,1$) model with normal errors and a simulated dataset from a GARCH($1,1$) with noise model calibrated to the Dow Jones Industrial Average; see Section \ref{subsec:data_SG_MCMC} for details. We then investigate the scalability of the proposed subsampling methodology as the parameter dimension increases. Since standard GARCH models are naturally low-dimensional, we additionally consider a regression-GARCH($1,1$) model with Student-$t$ errors,
\begin{equation*}
\begin{aligned}
    y_t &= \mathbf{x}_t^\top \boldsymbol{\beta} + z_t,
    \qquad
    z_t = \sigma_t \varepsilon_t,
    \qquad
    \varepsilon_t \overset{\mathrm{ind}}{\sim} t_\nu(0,1), \\
    \sigma_t^2 &= \omega + \alpha z_{t-1}^2
    + \beta_{\mathrm{GARCH}}\sigma_{t-1}^2,
\end{aligned}
\end{equation*}
where $\mathbf{x}_t\in\mathbb{R}^{d_x}$ denotes the vector of regressors and
$\boldsymbol{\beta}\in\mathbb{R}^{d_x}$ the corresponding regression coefficients. We consider $d_x=25$ and $d_x=50$, corresponding to total parameter dimensions
$29$ and $54$, respectively, and simulate $T=100{,}000$ observations in each setting. The model contains no intercept. The regressors are correlated with correlation parameter $\rho=0.5$, and the regression coefficients are scaled so that the variance of the regression signal $\mathbf{x}_t^\top\boldsymbol{\beta}$ equals one.

\begin{remark}
The likelihood derivatives required for constructing the Taylor control variates are obtained by differentiating the regression residual $z_t=y_t-\mathbf{x}_t^\top\boldsymbol{\beta}$ jointly with the GARCH recursion. The resulting gradient and Hessian follow the same recursive construction as for the GARCH models considered in Appendix~\ref{app:derivatives_GARCH}, with additional blocks corresponding to the regression coefficients. We omit these expressions here; the complete implementation is provided in the accompanying replication code.
\end{remark}

\subsection{Specifications and settings}

For the GARCH($1,1$) comparison, full-data MCMC and subsampling MCMC use the same settings as in the comparison with the stochastic-gradient Markov chain Monte Carlo algorithm in Section \ref{subsec:specifications_SG_MCMC}. For the regression-GARCH experiments, we retain the same settings wherever applicable, including the priors and stationarity restriction for the GARCH parameters, and use independent $\mathcal{N}(0,1)$ priors for the regression coefficients and $\nu-2\sim\mathrm{Gamma}(2,1)$. Subsampling MCMC is again tuned using $R_{\max}=100$.

For the DC-BATS implementation we partition the data, with $T = 100{,}000$, into $K$ consecutive shards and fit a powered posterior on each shard using Hamiltonian Monte Carlo as implemented in \texttt{CmdStanR}, the \texttt{R} interface to \texttt{CmdStan} \citep{cmdstanr}. The shards are constructed to be as equal in size as possible, with all observations retained; if shard $k$ contains $n_k$ observations, its likelihood is raised to the power $T/n_k$, which reduces to $K$ for equal-sized shards, following \cite{ou2025scalable}\footnote{DC-BATS code available at \url{https://github.com/astfalckl/dcbats/}.}. DC-BATS uses the same prior as subsampling MCMC and full-data MCMC. We use $4$ chains with $5{,}000$ warm-up iterations and $5{,}000$ post-warm-up draws per chain, initialised at the MAP estimate obtained from the full-data posterior. Since the data are simulated, each shard is initialised using the true initial values for the recursion (pre-sample) to ensure consistency across shards. The shard posterior draws are subsequently aggregated using a Wasserstein barycentre constructed from marginal quantile functions, from which $10{,}000$ samples are generated for posterior inference to match our post burn-in $10{,}000$ samples.

For the GARCH($1,1$) experiment, we consider $K = 10, 20, 76$. The first two settings align with the experimental setup in \cite{ou2025scalable}. The final setting is included as a stress scenario, chosen to approximately match, under idealised conditions without overhead costs, the computational speed-up of DC-BATS relative to full-data MCMC ($K=1$) with the speed-up achieved by subsampling MCMC relative to full-data MCMC under our implementation. Table \ref{tab:SG_MCMC_vs_subsampling} shows that subsampling MCMC is approximately $76$ times faster than full-data MCMC. For the regression-GARCH experiment, we again consider $K=10$ and $K=20$, together with a computationally matched stress setting chosen as above. This gives $K=37$ for $d_x=25$ and $K=36$ for $d_x=50$.

\subsection{Results: GARCH}
Figure \ref{fig:kde_DC_BATS_subsampling_full_data} shows kernel density estimates from subsampling MCMC and DC-BATS for several values of $K$, together with full-data MCMC estimates as a ground truth benchmark. Both methods are substantially more accurate than SG-MCMC; recall Figure \ref{fig:kde_SG_MCMC_subsampling_full_data}. Subsampling MCMC is virtually indistinguishable from the full-data posterior across all parameters. DC-BATS yields a very accurate posterior when $K=10$ and, as expected from the theoretical results in \cite{ou2025scalable}, the bias increases as $K$ increases. The deterioration is most noticeable for the persistence parameter $\beta_1$, which reflects the long memory in the volatility process and is therefore particularly sensitive to the partitioning of the time series. However, the bias remains small in absolute value (see the scale of the $x$-axis in Figure \ref{fig:kde_DC_BATS_subsampling_full_data}).

\begin{figure}[H]
    \centering
    \includegraphics[width=0.80\textwidth]{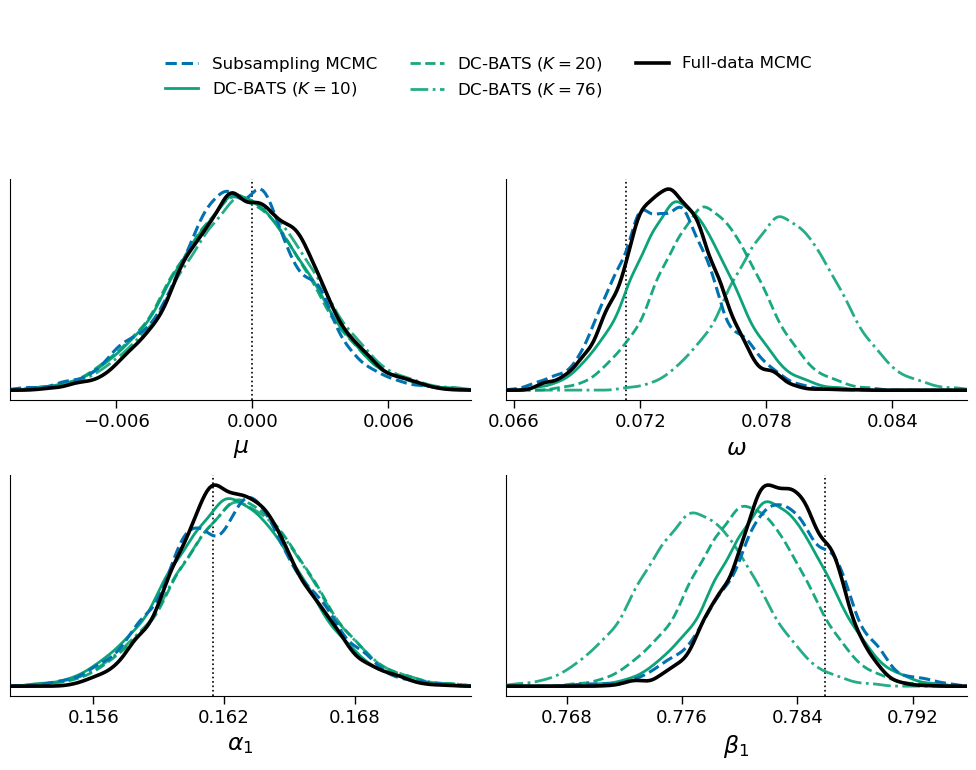}
    \caption{Posterior marginal distributions for the GARCH($1,1$) model fitted to simulated data calibrated to the Dow Jones Industrial Average data. The subsampling Markov chain Monte Carlo (MCMC) method with truncated decaying sampling probabilities is shown for $R_{\max}=100$, corresponding to a variance tolerance $V\approx 0.20$. Divide-and-conquer Bayesian time series (DC-BATS) MCMC is shown for several values of the number of shards $K$. The solid black curves represent the full-data MCMC posterior for reference, and the vertical lines indicate the true parameter values.}
    \label{fig:kde_DC_BATS_subsampling_full_data}
\end{figure}

Turning to computational cost, the DC-BATS implementation exhibits approximately linear computational savings in $K$ under idealised conditions. In particular, posterior approximations for $K=10$ and $K=20$, the regimes considered in \cite{ou2025scalable}, can be obtained approximately 10 and 20 times faster than the corresponding full-data implementation ($K=1$). Subsampling MCMC achieves a reduction in CPU time of approximately 76 times relative to full-data MCMC under our implementation, while retaining an accurate approximation. Achieving a comparable speed-up with DC-BATS would require a value of $K=76$, which may result in a deterioration in accuracy, as observed above.

\subsection{Results: Regression-GARCH}

Figures \ref{fig:regression_GARCH_25} and \ref{fig:regression_GARCH_50} show selected posterior marginal distributions for the regression-GARCH model with $d_x=25$ and $d_x=50$, respectively. Subsampling MCMC remains close to the full-data posterior as the parameter dimension increases. DC-BATS also provides accurate approximations for $K=10$ and $K=20$, while discrepancies become more visible at the computationally matched settings $K=37$ and $K=36$, particularly for some of the GARCH and Student-$t$ parameters. These results indicate that the proposed subsampling methodology continues to provide accurate posterior approximations in substantially higher-dimensional settings.

\begin{figure}[H]
    \centering
    \includegraphics[width=0.75\textwidth]{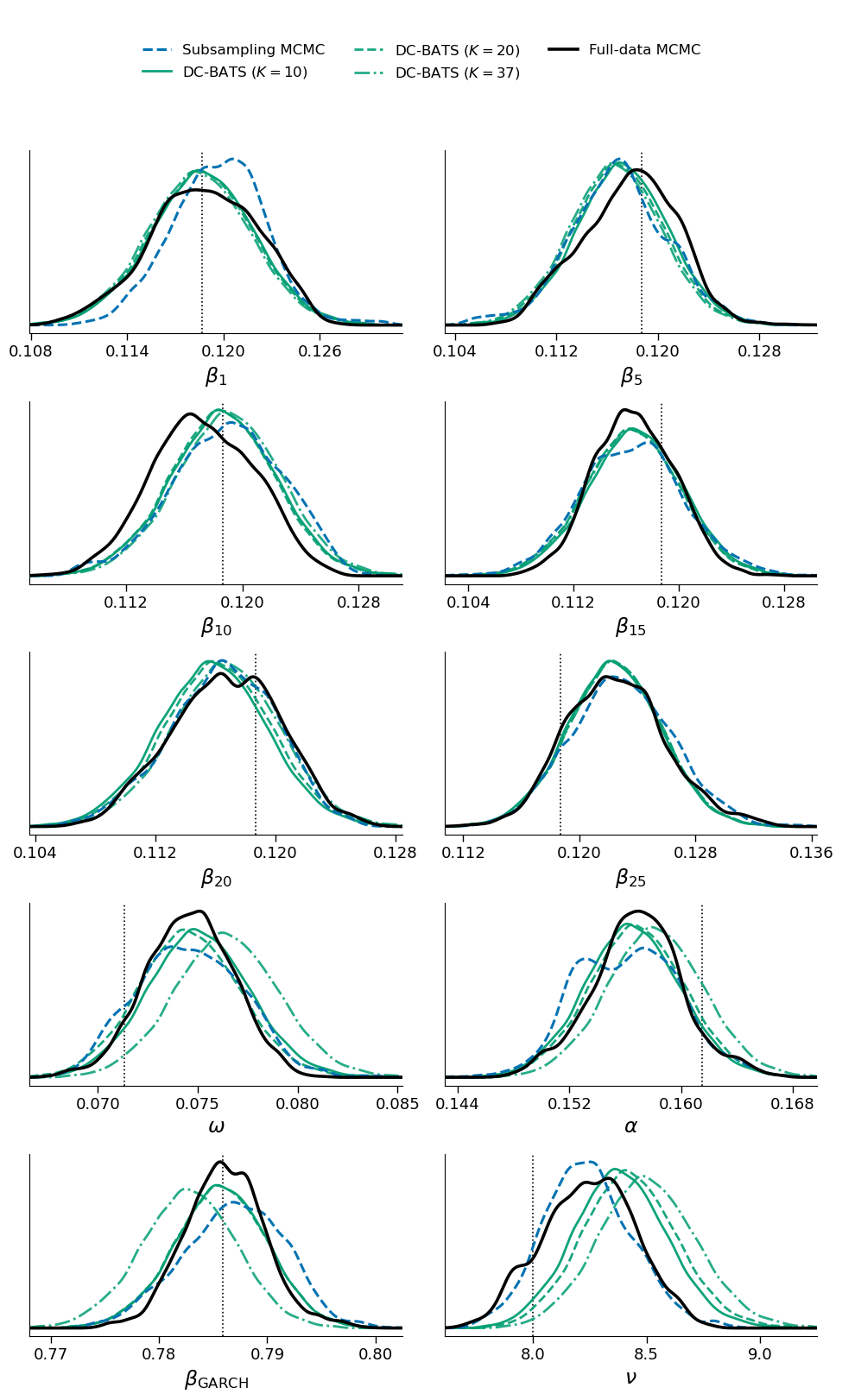}
    \caption{Selected posterior marginal distributions for the regression-GARCH($1,1$) model with $d_x=25$ regressors, corresponding to a total parameter dimension of $29$. Subsampling Markov chain Monte Carlo (MCMC) and divide-and-conquer Bayesian time series (DC-BATS) MCMC are compared with full-data MCMC, with DC-BATS shown for $K=10$, $K=20$, and the computationally matched setting $K=37$. The vertical lines indicate the true parameter values.}
    \label{fig:regression_GARCH_25}
\end{figure}

\begin{figure}[H]
    \centering
    \includegraphics[width=0.75\textwidth]{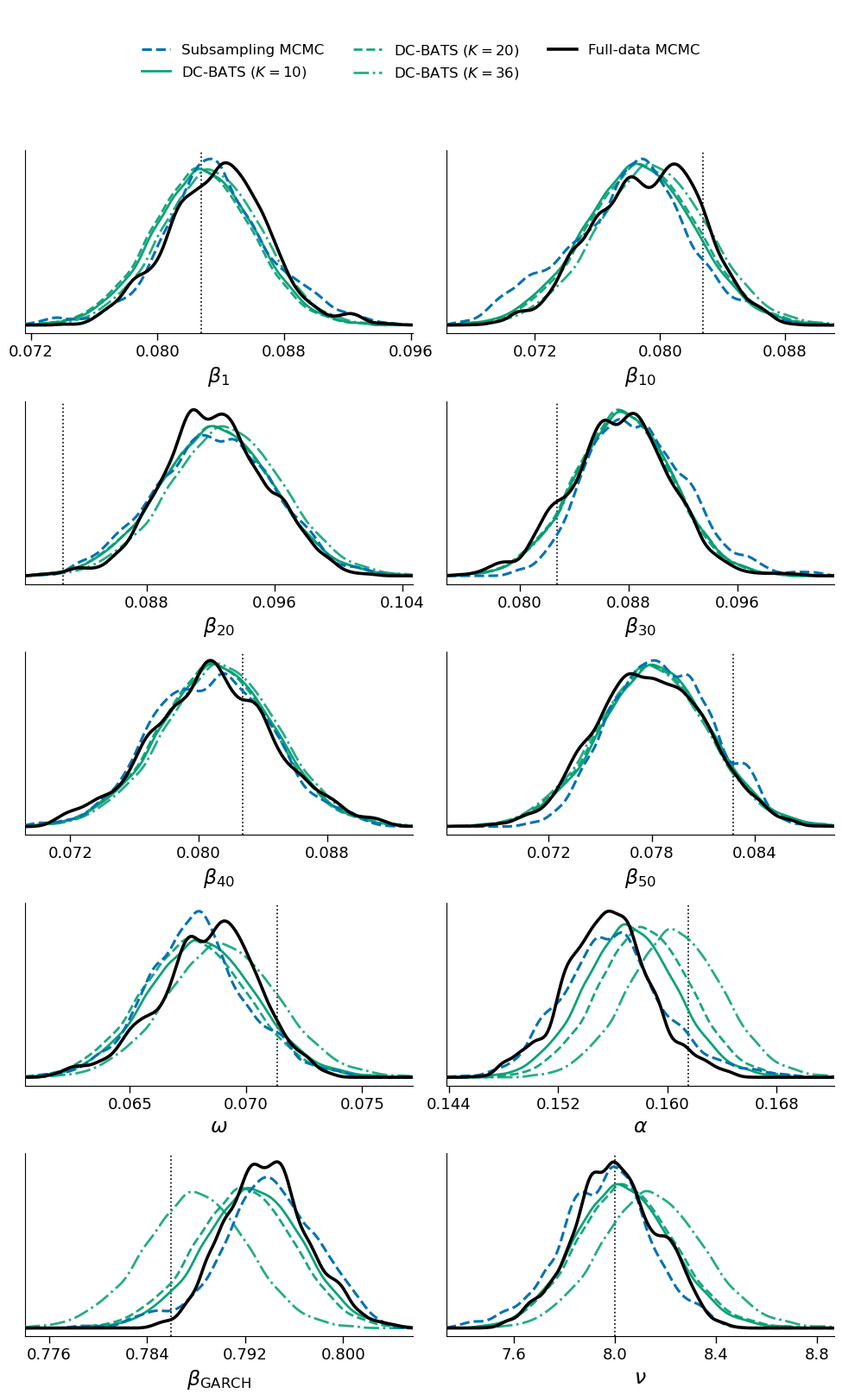}
    \caption{Selected posterior marginal distributions for the regression-GARCH($1,1$) model with $d_x=50$ regressors, corresponding to a total parameter dimension of $54$. Subsampling Markov chain Monte Carlo (MCMC) and divide-and-conquer Bayesian time series (DC-BATS) MCMC are compared with full-data MCMC, with DC-BATS shown for $K=10$, $K=20$, and the computationally matched setting $K=36$. The vertical lines indicate the true parameter values.}
    \label{fig:regression_GARCH_50}
\end{figure}

\putbib
\end{bibunit}
\end{document}